\documentclass[11pt,a4paper]{article}

\usepackage[margin=1in]{geometry}
\usepackage{algpseudocode}
\usepackage{algorithm}
\usepackage{amsmath}
\usepackage{amssymb}
\usepackage{amsthm}
\usepackage{mathrsfs}
\usepackage{amsbsy}
\usepackage{upgreek}
\usepackage[table]{xcolor}
\usepackage{enumerate}
\usepackage{todonotes}
\usepackage{graphicx}
\usepackage{subcaption}
\usepackage[bookmarks=true,pdfborder={0 0 0}]{hyperref}
\def\S{{\mathcal S}}

\def\G{{\mathcal G}}
\newcommand{\UR}{\textsf{\sc UseRnwMove}}
\newcommand{\DL}{\textsf{\sc DnwLseMove}}
\newcommand{\STR}{\textsf{\sc STRnwMove}}

\usepackage{authblk}
\usepackage{mathtools}
\DeclareMathOperator{\interior}{int}
\usepackage{lipsum}
\usepackage{lineno}
 \newtheorem{claim}{Claim}

\newtheorem{theorem}{Theorem}
 \newtheorem{lemma}{Lemma}

\newtheorem{observation}{Observation}

\newcommand{\IR}{\mathbb{R}}

\providecommand{\keywords}[1]
{
  \small	
  \textbf{Keywords:} #1
}

\algrenewcommand\algorithmicrequire{\textbf{Input:}}
\algrenewcommand\algorithmicensure{\textbf{Output:}}

\usepackage{color}

\def\S{{\mathcal S}}
\def\Q{{\cal Q}}

\def\R{{\mathcal R}}

\def\I{{\mathcal I}}

\def\P1{{\mathcal P_1}}
\def\P2{{\mathcal P_2}}
\def\P{{\mathcal P}}

\def\O{{\mathcal O}}
\title{Online Class Cover Problem\thanks{
      This paper has been accepted in Computational Geometry: Theory and Applications. DOI: \url{https://doi.org/10.1016/j.comgeo.2024.102120}. Work on this paper by M. De has been partially supported by SERB MATRICS Grant MTR/2021/000584, work by A. Maheshwari has been supported by NSERC, and work by R. Mandal has been supported by CSIR, India, File Number- 09/0086(13712)/2022-EMR-I.}}

\author[]
       {Minati De$^1$, Anil Maheshwari$^2$,  and Ratnadip Mandal$^{1}$
       \\
       $^1$Dept. of Mathematics, Indian Institute of Technology Delhi, New Delhi, India\\
       minati@maths.iitd.ac.in, maz218522@iitd.ac.in\\ 
       $^2$School of Computer Science, Carleton University, Ottawa, Canada\\
       anil@scs.carleton.ca\\ 
       }

\begin{document}

\maketitle

\begin{abstract}
    In this paper, we study the online class cover problem where a (finite or infinite) family $\cal F$ of geometric objects and a set ${\cal P}_r$ of red points in $\mathbb{R}^d$ are given a prior, and blue points from $\mathbb{R}^d$ arrives one after another. Upon the arrival of a blue point, the online algorithm must make an irreversible decision to cover it with objects from $\cal F$ that do not cover any points of ${\cal P}_r$. The objective of the problem is to place a minimum number of objects. When $\cal F$ consists of axis-parallel unit squares in $\mathbb{R}^2$, we prove that the competitive ratio of any deterministic online algorithm is $\Omega(\log |{\cal P}_r|)$, and also propose an $O(\log |{\cal P}_r|)$-competitive deterministic algorithm for the problem.
\end{abstract}


\section{Introduction}
Let $\P_b$ and $\P_r$ be two sets of blue and red points in $\mathbb{R}^d$, where $|\P_b|=n$ and $|\P_r|=m$. In this scenario, the set $\P=\P_b \cup\P_r$ is called a bi-colored point set, and $\P_b$ and $\P_r$ are referred to as the blue and red point classes, respectively. A geometric object (for example, ball, hypercube, etc.)   is considered \emph{$\P_r$-empty} if it does not contain any point from the set $\P_r$. For a given  (finite or infinite) family $\cal F$ of geometric objects in $\IR^d$ and a bi-colored set $\P=\P_b\cup\P_r$, the objective of the \emph{class cover problem} is to find a minimum cardinality subset  $\cal S\subseteq \cal F$, consisting of only $\P_r$-empty geometric objects,  that covers the blue point class (i.e., every point in $\P_b$ is contained in at least one of the objects without covering any points in $\P_r$).

For our research, we are considering the  \emph{online class cover problem} where we only know the red point class $\P_r$ and don't know the blue point class in advance. The blue points will come one by one, and upon the arrival of a blue point, we need to make an irreversible decision to cover it with  $\P_r$-empty objects. Again, the objective of the problem is to find a minimum cardinality set of $\P_r$-empty objects.
 In this paper, we consider the problem when $\cal F$ is the family of axis-parallel unit squares
 % of all translated copies of a rectangle
 in $\mathbb{R}^2$, and $\P_r$ is a set of points in $\IR^2$.

The quality of our online algorithm is analyzed using competitive analysis \cite{Borodin98}. The \emph{competitive ratio} of an online algorithm is  $\sup_{\I} \left(\frac{\text{ALG}(\I)}{\text{OPT}(\I)}\right)$, where $\I$ is an input sequence, and $\text{OPT}(\I)$ and $\text{ALG}(\I)$ are the cost of an optimal offline algorithm and the solution produced by the online algorithm, respectively, over $\I$. Here, the supremum is taken over all possible input sequences $\I$.

Online computation models capture many real-world phenomena: the irreversibility of time (decisions), the inaccessibility of complete data, and the ambiguity of the future. As a result, online models of computation have attracted many researchers. In computational geometry, much work has been conducted on the online model for various problems such as the set cover problem \cite{CharikarCFM04, DumitrescuGT20, DumitrescuT22}, hitting set problem \cite{DeS24, EvenS14,  ArindamLRSW23}, independent set problem~\cite{CaragiannisFKP07, OliverHKSV14}, coloring problem~\cite{ChenFKLMMPSSWW07,ChenKS09} and the dominating set problem~\cite{DeKS23, Eidenbenz}. There is a lack of literature on the online class cover problem. This is a natural problem, and it has various applications in data mining, pattern recognition, scientific computation, visualization and computer graphics~\cite{AgarwalS98, Devinney03, Mitchell93}.
An application example of the online class cover problem inspired by Alon et al.~\cite{AlonAABN09} is as follows. Consider, for instance, network servers that provide a service. A set of clients may require the service, and each server can provide it to a subset of those clients. The requests of clients arrive one by one. The network administrator must determine where to deploy a server upon the introduction of a client so that the client receives the requested service.  However, there may be some customers who utilize the services of a competitor. Therefore, the network administrator must install the server to exclude rival clients and cover the area where a potential future client may arrive.

% Consider, for example, servers in a network that provide a service. There is a set of potential clients that may need the service and each server can provide the service to a subset of them. (E.g., the subset is determined by the distance from the server.) There is a setup cost, or activation cost, associated with the operation of a server. The clients arrive one by one. Upon the arrival of a client, the network manager has to decide which server to activate so that the client receives the service it requested. The network manager knows in advance the set of potential clients and the set of servers, however, it does not know in advance which clients will indeed request the service.

\subsection{Related Work}
We aren't aware of any work in the online setting of the class cover problem. But in the offline setting, some of the key results are highlighted below. Inspired by the Cowen–Priebe method for the classification of high-dimensional data (see~\cite{cannon1998approximate,cowen1997randomized}), Cannon and Cowen~\cite{CannonC04} introduced the class cover problem as follows.
\begin{align*}
\text{Minimize} \quad
&k \\
\text{Subject to,} \quad
&\max_{u\in \P_b} d(u,S) \leq \min_{u\in \P_r} d(u,S), \\
&\text{where}\ S\subseteq \P_b,|S|=k.
\end{align*}

Here the point-to-set distance $d(u, S)$ is defined as $\min_{s\in S} d(u, s)$. In other words, given a bi-colored set $\P=\P_b \cup\P_r$, the objective is to find a minimum cardinality set of blue centers to cover the blue points with a set of blue-centered
balls of equal radius, such that no red points lie in these balls. They showed that this problem is NP-hard and gave an $(1 + \ln |\P|)$-approximation algorithm for general metric spaces that runs in cubic time. They presented a polynomial-time approximation scheme (PTAS) in the Euclidean setting for constant dimension.

In~\cite{BeregCDPSV12}, Bereg et al. considered the class cover problem for axis-aligned rectangles (i.e., boxes) and called it the \emph{Boxes Class Cover problem} (BCC problem). They proved the NP-hardness of the problem and gave an $O(\log |\text{OPT}|)$-approximation algorithm, where OPT is an optimal covering. They also study some variants of the problem. If the geometric objects are axis-parallel strips and half-strips oriented in the same direction, they gave an exact algorithm running in $O(m \log m +n \log n + \sqrt{mn})$-time and $O((m + n)\log(\min\{m, n\}))$-time, respectively. However, if half-strips are oriented in any of the four possible orientations, the problem is NP-hard. But in this case, they showed that there exists an $O(1)$-approximation algorithm. In the last variant of the BCC problem, they considered axis-aligned squares. Here, they proved the NP-hardness of the problem and presented an $O(1)$-approximation algorithm. Shanjani~\cite{Shanjani20} proved that the BCC problem is APX-hard, whereas a PTAS  exists for the variant of the problem for disks and axis-parallel squares~\cite{AschnerKMY13}.

In~\cite{CardinalDI21}, Cardinal et al. considered a symmetric version of the BCC problem and called it \emph{Simultaneous BCC problem} (SBCC problem), which is defined as follows. Given a set $\P$ of points in the plane, each colored red or blue, find the smallest cardinality set of axis-aligned boxes, which together cover $\P$ such that each box cover only points of the same color and no box covering a red point intersects the interior of a box covering a blue point \cite[Dfn.~2]{CardinalDI21}. They showed that this problem is APX-hard and gave a constant-factor approximation algorithm.

% As of our knowledge, this is the first work related to the online class cover problem for any geometric objects.
 If the red point class $\P_r$ is empty, and $\cal F$ consists of all possible translated copies of an object, then the online class cover problem is equivalent to the online unit cover problem. 
 Charikar et al.~\cite{CharikarCFM04} proposed an $O(2^dd\log d)$ competitive algorithm for the online unit cover for balls in $\IR^d$.  Later, Dumitrescu et al.~\cite{DumitrescuGT20} obtained $O({1.321}^d)$ competitive algorithm and proved that the lower bound is $\Omega(d+1)$. In particular, for balls in $\IR^2$ and $\IR^3$, they obtained competitive ratios of 5 and 12, respectively. 
Dumitrescu and Tóth~\cite{DumitrescuT22} studied the unit cover problem for hyper-cubes in $\mathbb{R}^d$. They proved that the lower bound of this problem is at least $2^d$, and also this ratio is optimal, as it is obtainable through a simple deterministic algorithm that allocates points to a predefined set of hyper-cubes \cite{ChanZ09}. Recently, De et al.~\cite{DeJKS24} gave a deterministic online algorithm of the unit cover problem for $\alpha$-aspect$_\infty$  objects in $\mathbb{R}^d$ with a competitive ratio at most $\left(\frac{2}{\alpha}\right)^d((1+\alpha)^d-1)\log_{(1+\alpha)} \left(\frac{2}{\alpha}\right) +1$. To summarise existing results related to the online unit cover problem using translates of an object of various shapes, we refer to {Table of~\cite{DeJKS24}}.

Consider the online class cover problem when  $\cal F$ is an infinite family of intervals on $\mathbb{R}$ and $\P_r$ is a set of $m$ red points in $\mathbb{R}$. Observe that the $m$ red points partition $\mathbb{R}$ into $m+1$ open intervals. As a result of this, any $\P_r$-empty interval must lie inside one of the $m+1$ open intervals. Therefore, the online class cover problem for intervals in $\mathbb{R}$ reduces to $m+1$ distinct online set cover problems for intervals in $\mathbb{R}$. For this case, if  $\cal F$ is the set of all possible translations of an interval, then an optimum $2$-competitive algorithm for the problem can be obtained~\cite{ChanZ09, CharikarCFM04}.

\subsection{Our Contributions}
This paper primarily presents two results on the online class cover problem when $\cal F$ is the set of all possible axis-parallel unit squares in $\mathbb{R}^2$, and $\P_r$ is a set of $m$ points in $\IR^2$. We show, first, that the competitive ratio of every deterministic online algorithm of the problem is $\Omega(\log m)$ (Theorem~\ref{theo:lowerbound}). To prove this, we consider an adaptive adversary. Initially, the adversary constructs a set $\P_r$ of $m$ red points such that a unit square contains $\P_r$. In each of $O(\log m)$ rounds, the adversary maintains a subset $P \subseteq \P_r$ of red points and places a blue point in a grid cell, obtained by using the points in $P$, such that any online algorithm must place a new square to cover it,  whereas an offline optimal covers all the blue points by a single square.

Next, we provide a deterministic online algorithm with a competitive ratio of $O(\log m)$ (Theorem~\ref{theo:upperbound} and Theorem~\ref{theo:upperbound_improved}). Our algorithm, in a nutshell, is as follows. When a new blue point is introduced to the algorithm that is uncovered by the previously selected squares, the algorithm chooses at most five squares, called candidate squares (see Section~\ref{sec:algo} for definition) to cover it. The candidate squares are constructed using a binary search over potential squares. The main challenge and the novelty is to show that the candidate squares are sufficient to establish the competitive ratio. As stated previously, if $\P_r= \emptyset$, the online class cover problem for translated copies of an object reduces to the online unit cover problem. In that case, for axis-parallel squares, our algorithm works like the well-known algorithm, namely \texttt{Algorithm Centered}, and thus achieves the optimal competitive ratio of 4 \cite{ChanZ09, DumitrescuT22}. For simplicity, throughout the paper, we consider that the set $\P_r$ contains at least one point, i.e., $m\geq 1$, unless mentioned.
% Finally, we discuss that the problem's upper bound $O(\log m)$ is also achievable when the objects are translated copies of a rectangle (Theorem~\ref{theo:upperbound_rect}).
% Furthermore, we show that the same upper bound of the online class cover problem applies to translated copies of a rectangle (Theorem~\ref{theo:upperbound_rect}).

\subsection{Outline of the Paper}
In Section~\ref{sec:nota}, we introduce terminology that will be used in subsequent sections. Section~\ref{sec:lower} demonstrates the lower bound for the online class cover problem for axis-parallel unit squares. Then, in Section~\ref{sec:algo}, we present an online algorithm for the problem. In Section~\ref{sec:ana_algo}, we give our algorithm's correctness and competitive analysis. 
% In Section~\ref{sec:rect}, we examine the online class cover problem for translated copies of a rectangle. 
Finally, we conclude in Section~\ref{sec:conclud}.

% In Section~\ref{sec:nota}, we give some definitions and introduce some basic terminology that will be used later on. Next in Section~\ref{sec:lower}, we prove the lower bound of the online class cover problem for axis-parallel unit squares. Then in Section~\ref{sec:upper}, we discuss our algorithm and the analysis of it. At last in Section~\ref{sec:rect}, we consider the online class cover problem for translated copies of a rectangle.

\section{Notation and Preliminaries}\label{sec:nota}

\begin{figure}[htbp]
    \centering
     \begin{subfigure}[b]{.32\textwidth}
        \centering
        \includegraphics[page=1, width=48mm]{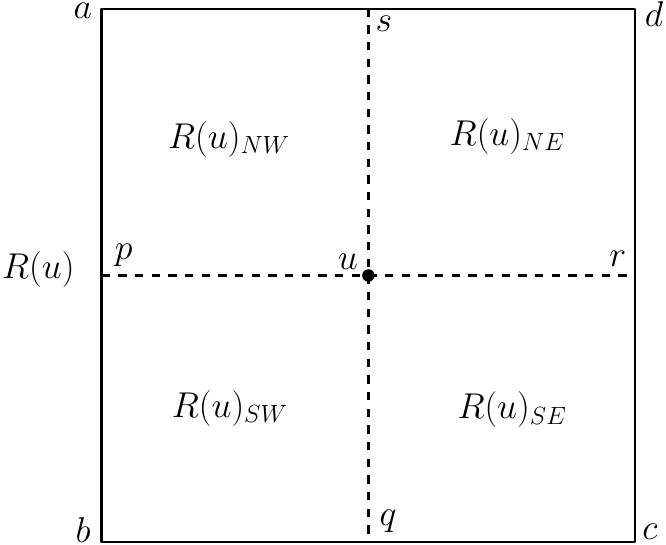}
        \caption{}
        \label{fig:unit_square}
    \end{subfigure}
    \hfill
        \begin{subfigure}[b]{.30\textwidth}
        \centering
        \includegraphics[page=2, width=47mm]{Preliminaries.pdf}
        \caption{}
        \label{fig:hori_trans}
    \end{subfigure}
    \hfil
    \begin{subfigure}[b]{.32\textwidth}
        \centering
        \includegraphics[page=3, width=48mm]{Preliminaries.pdf}
        \caption{}
        \label{fig:grid}
    \end{subfigure}
    \caption{(a) An axis-parallel unit square $R(u)=\square abcd$ centered at $u$. (b) The square $P+(v_1, 0)$ (red) is a horizontal translated copy of the square $P$ (black). (c) An $8\times 8$ {grid} $\G(Q)$ using $9$ points of $Q=\{u_1, u_2,\dots, u_9\}$ in a plane.}
    \label{fig:unit_square_1}    
\end{figure}

We use $[n]$ to represent the set  $\{1,2,\ldots,n\}$ where $n$ is a positive integer. We use the phrase `\emph{$X$ contains $Y$}' to signify that $Y\subseteq X$.
For any point $u$, we use $u(x)$ and $u(y)$ to denote the $x$ and $y$-coordinates of $u$, respectively. A point $u$ is said to be lying on the left side of a vertical line segment $l$ if $u$ lies on the left side of the line $L$, obtained by extending the line segment $l$ in both directions. Similarly, we define the other cases such as $u$ lying on the right, the above and the bottom side of a horizontal/vertical line segment~$l$. Let $R(u)=\square abcd$ be an axis-parallel unit square centred at a point $u$ (see Figure~\ref{fig:unit_square}). Let $p,q,r$ and $s$ be the midpoints of the edges $ab, bc, cd$ and $ad$ of $R(u)$, respectively. The vertices $a, b, c$ and $d$ will be called $NW, SW, SE$ and $NE$ vertices of $R(u)$, respectively. The edges $ab, bc, cd$ and $ad$ will be called the left, the bottom, the right and the top edge of $R(u)$, respectively. Now, the two line segments $\overline{pr}$ and $\overline{qs}$ partition $R(u)$ into four \emph{sub-squares} $R(u)_{NW}=\square apus, R(u)_{SW}=\square pbqu, R(u)_{SE}=\square uqcr$ and $R(u)_{NE}=\square surd$ each of side length $\frac{1}{2}$. The boundary of $R(u)$ will be denoted by $\partial R(u)$, and the interior by $\interior(R(u)) = \{v\in R(u)\ |\ v \notin \partial R(u)\}$. A point $v$ is said to be \emph{covered/contained} by a square $R(u)$ if $v \in \interior{(R(u))}$. Throughout the paper, a square means an axis-parallel unit square unless mentioned.

% Let $\tau_{xy}=(x, y)$ be a vector in $\mathbb{R}^2$. We denote $\tau_{x0}=(x, 0)$ and $\tau_{0y}=(0, y)$ as $\tau_x$ and $\tau_y$, respectively. 
A square $P$ is said to be a \emph{translated copy} of a square $P'$ if $P=P'+v=\{w+v\ |\ w\in P'\}$, where  $v=(v_1, v_2)$ is a vector in $\mathbb{R}^2$. We say that the square $P$ is a \emph{horizontal translated copy} of $P'$ if the vector $v$ is of the form $(v_1, 0)\in \mathbb{R}^2$ (see Figure~\ref{fig:hori_trans}). Similarly, if the vector $v$ is of the form $(0, v_2)\in \mathbb{R}^2$, the square $P$ is said to be a \emph{vertical translated copy} of $P'$. A square $P$ is moved \emph{rightwards}, meaning that we consider a horizontal translated copy $P'=P+(v_1, 0)$ of $P$, where the value of $v_1>0$ and it is continuously increasing. Similarly, we can define the movement of $P$ in other directions, such as leftwards, upwards, and downwards.

Let $Q=\{u_1, u_2,\dots, u_q\}$ be a set of points in a plane having distinct $x$ and $y$-coordinates. Consider the horizontal and vertical lines passing through each point $u\in Q$. They form a $(q-1)\times(q-1)$ grid, denoted by $\G(Q)$ (see~Figure~\ref{fig:grid}). We consider the bottom-most row of $\G(Q)$ as $1$st row and the right-most column of $\G(Q)$ as $1$st column. Let $G_{ij}$ be the cell in the grid that appears in the $i$th row and the $j$th column for $i, j\in [q-1]$. Note that $G_{11}$ and $G_{(q-1)(q-1)}$ are the $SE$ and the $NW$ corner cells in the grid. {The \emph{height} (respectively, \emph{width}) of the grid $\G(Q)$ means the vertical\footnote{For any two points $a, b\in \mathbb{R}^2$, the vertical distance between $a$ and $b$ is defined by $|b(y)-a(y)|$. Similarly, the horizontal distance between $a$ and $b$ is $|b(x)-a(x)|$.} (respectively, horizontal) distance between the top-most and the bottom-most points (respectively, the left-most and the right-most points) of $Q$}.
A grid cell $G_{ij}$ is said to be \emph{cover-free} with respect to a set $\{R_1,\dots,R_k\}$ of squares, if there exists a point $w\in G_{ij}$ such that $w\notin \cup_{i=1}^{k} R_i$.

\begin{figure}[htbp]
    \centering
    \begin{subfigure}[b]{.33\textwidth}
        \centering
        \includegraphics[page=4, width=40mm]{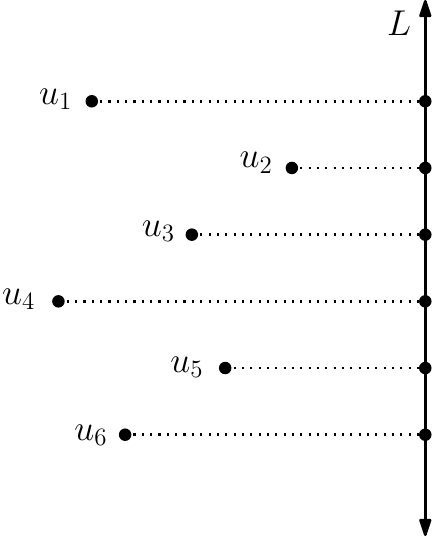}
        \caption{}
        \label{fig:nearest_point}
    \end{subfigure}
    \hfil
    \begin{subfigure}[b]{.30\textwidth}
        \centering
        \includegraphics[page=5, width=47mm]{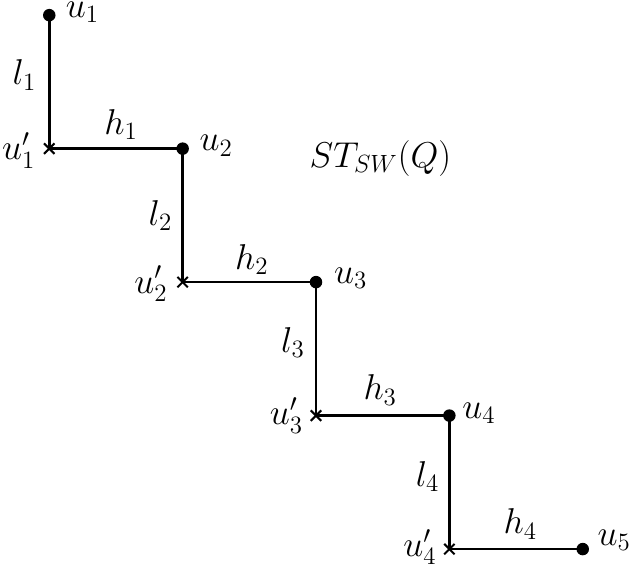}
        \caption{}
        \label{fig:staircase}
    \end{subfigure}
    \hfil
    \begin{subfigure}[b]{.33\textwidth}
        \centering
        \includegraphics[page=6, width=35mm]{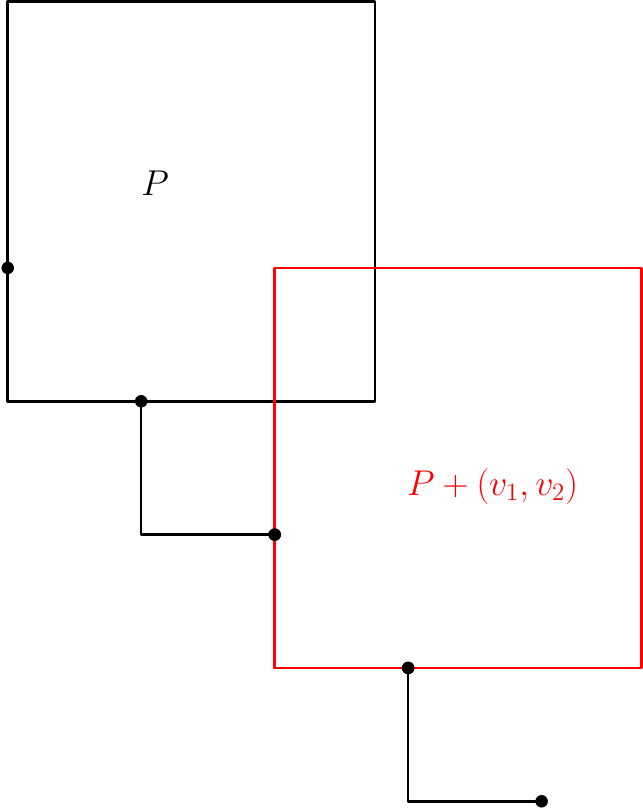}
        \caption{}
        \label{fig:staircase_movement}
    \end{subfigure}
    \caption{ (a) The point $u_2\in Q$ is nearest to the line $L$, where $Q=\{u_1, u_2,\dots, u_6\}$. (b) Here, $ST_{SW}(Q)$ is the sequence $\sigma=\{l_1, h_1, \dots,$ $l_{4}, h_{4}\}$, where $Q=\{u_1, u_2,\dots, u_5\}$. (c) The square $P+(v_1, v_2)$ (red) is obtained by moving the square $P$ (black) rightwards along the staircase.}
    \label{fig:unit_square_2}    
\end{figure}

For any two points $a, b\in \mathbb{R}^2$, we use $d(a, b)$ to represent the distance between the points $a$ and $b$ under the Euclidean norm. Let $L$ be any horizontal or vertical line. Define, $d(a, L)=\min_{u\in L} d(a,u)$. Consider a set $Q$ of points. Now, a point $u\in Q$ is said to be the \emph{nearest point} to the line $L$, if $d(u, L)= \min_{a\in Q} d(a, L)$ (see Figure~\ref{fig:nearest_point}). With slight abuse, a point $u\in Q$ is said to be the nearest point to a line segment $l$ if $u$ is the nearest point to the line $L$, obtained by extending the line segment $l$ in both directions.

For a set $Q$ of points, define the set $D_{SW}(Q)=\{q\in Q\ |\ \nexists\ q'\in Q$ such that $q(x)<q'(x)$ and $q(y)<q'(y)\}$. We call $D_{SW}(Q)$ the \emph{set of $SW$ dominating points} of $Q$. Similarly, we can define the other dominating sets of $Q$ such as $D_{NW}(Q), D_{NE}(Q)$ and $D_{SE}(Q)$. Let $Q'=D_{SW}(Q)=\{u_1, u_2,\dots, u_q\}$ be such that $u_i(x)\leq u_{i+1}(x)$ and $u_i(y)\geq u_{i+1}(y)$ for $i\in [q-1]$. Now, we define a \emph{SW staircase} of $Q$, denoted by $ST_{SW}(Q)$, as an alternating sequence $\{l_1, h_1, \dots,$ $l_{q-1}, h_{q-1}\}$ of vertical and horizontal line segments such that $l_i=\overline{u_iu'_i}$ and $h_i=\overline{u'_{i}u_{i+1}}$, where $u'_i=(u_{i}(x), u_{i+1}(y))\in \mathbb{R}^2$ for $i\in [q-1]$ (see Figure~\ref{fig:staircase}). The two points $u_1$ and $u_q$ are called the \emph{initial point} and the \emph{terminal point} of $ST_{SW}(Q)$, respectively. Similarly, we can define the other staircases of $Q$ such as $ST_{NW}(Q)$, $ST_{NE}(Q)$ and $ST_{SE}(Q)$. Let $P$ be a square such that the $SW$ vertex of $P$ coincides with $u'_i$ for some $i\in [q-1]$. Now, $P$ is moved \emph{rightwards along the staircase $ST_{SW}(Q)$}, meaning that we consider a translated copy $P'=P+(v_1, v_2)$ of $P$ with $v_1\geq 0, v_2\leq 0$ such that the $SW$ vertex of $P'$ always lies either on $l_j$ or on $h_j$ for some $j\in [q-1]$, and either the value of $v_1$ is continuously increasing or the value of $v_2$ is continuously decreasing (see Figure~\ref{fig:staircase_movement}). Similarly, we can define the movement of $P$ along the staircase $ST_{SW}(Q)$ in the leftward direction.

% The two points $u_1$ and $u_q$ are called the \emph{initial point} and the \emph{terminal point} of $ST_{SW}(Q)$, respectively.

\section{Lower Bound Construction}\label{sec:lower}
In this section, we present a lower bound of the online class cover problem for squares.

\begin{theorem}\label{theo:lowerbound}
The competitive ratio of every deterministic online algorithm for the class cover problem for squares is at least $\max\{4, \lfloor \log_2 m\rfloor + 1\}$, where $m\ (\geq 1)$ is the number of red points.
\end{theorem}
\begin{figure}[htbp]
    \centering
    \begin{subfigure}[b]{.49\textwidth}
        \centering
        \includegraphics[page=1, width=65mm]{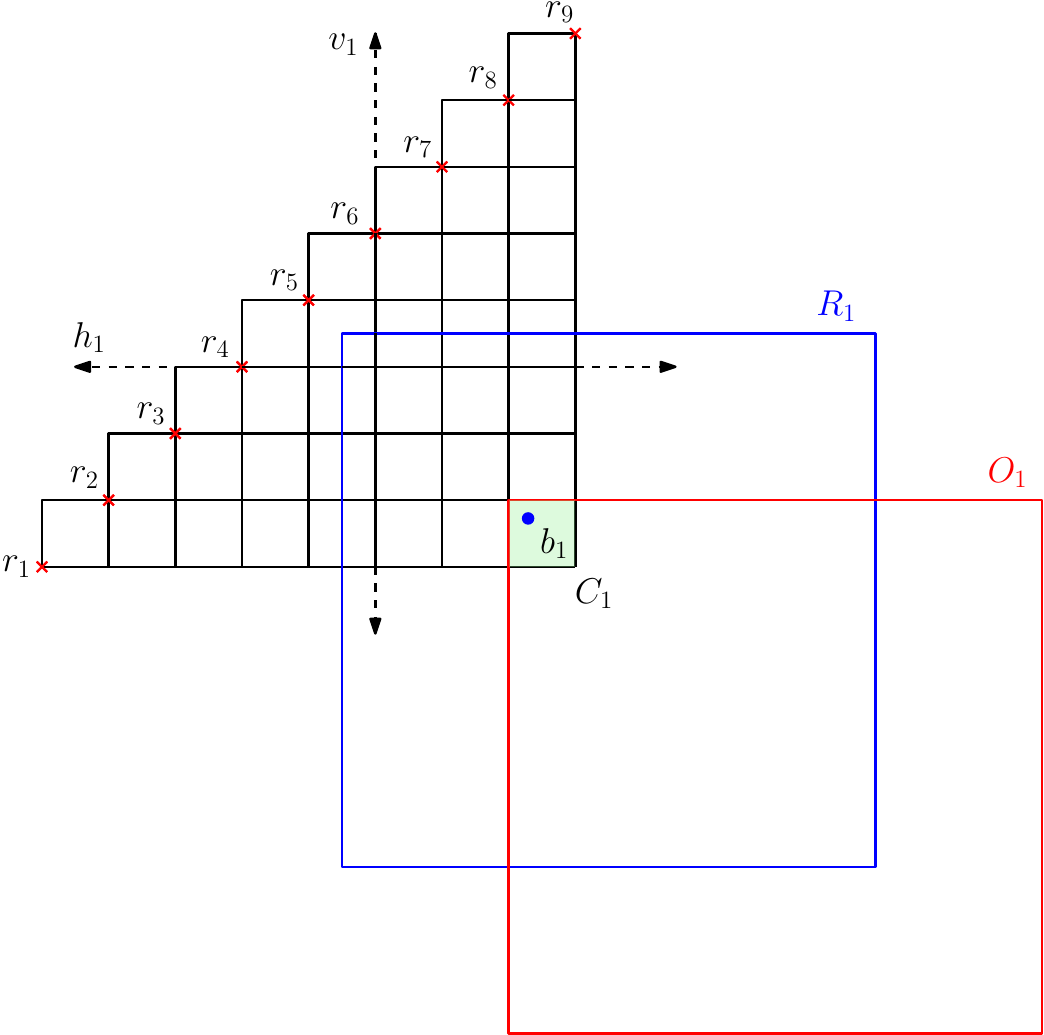}
        \caption{}
        \label{fig:lower_ex_1}
    \end{subfigure}
    \hfill
    \begin{subfigure}[b]{.49\textwidth}
        \centering
        \includegraphics[page=2, width=70mm]{Lower_Bound.pdf}
        \caption{}
        \label{fig:lower_ex_2}
    \end{subfigure}
    % \hfill
    % \begin{subfigure}[b]{.49\textwidth}
    %     \centering
    %     \includegraphics[width=70 mm]{Lower_Ex_3.pdf}
    %     \caption{}
    %     \label{fig:lower_ex_3}
    % \end{subfigure}
    % \hfill
    % \begin{subfigure}[b]{.49\textwidth}
    %     \centering
    %     \includegraphics[width=70 mm]{Lower_Ex_4.pdf}
    %     \caption{}
    %     \label{fig:lower_ex_4}
    % \end{subfigure}
    \caption{Example of the lower bound of the online class cover problem using axis-parallel unit square where $m=9$. (a) $O_1$ is the red-colored square. The algorithm places the square $R_1$ (blue) to cover $b_1$. (b) $O_2$ is the red-colored square that covers both $b_1$ and $b_2$. The algorithm places the square $R_2$ (green) to cover $b_2$.}
    \label{fig:lower_ex}
\end{figure}

    % (c) To cover $b_3$, the algorithm places the square $R_3$ (orange). (d) To cover $b_4$, the algorithm places the square $R_4$ (yellow). But an offline algorithm just places $R_4$ to cover all the blue points.

\begin{proof}
     Dumitrescu and Tóth~\cite{DumitrescuT22} proved that the competitive ratio of every deterministic online algorithm for the set cover problem using axis-parallel unit squares in a plane is at least~$4$. {The same lower bound construction can be done for the class cover problem, where the adversary can completely ignore the set of red points as follows. The adversary will place blue points far enough from the given set of red points so that any square covering a blue point never contains a red point. As a result, the lower bound of~\cite{DumitrescuT22} applies to the online class cover problem for any value of $m\geq 1$.} 
     % Ignoring the red points, the adversary can place the blue points far enough from the red points so that any square covering a blue point never contains a red point. This allows the adversary to apply the lower bound construction of~\cite{DumitrescuT22} for the online class cover problem of squares, resulting in a lower bound of $4$ for any value of $m\geq 1$.
     % Since the online class cover problem for squares is a general version of the online set cover problem for squares, this lower bound applies to the online class cover problem. 
     Next, to prove the theorem, we construct a lower bound $\lfloor \log_2 m\rfloor + 1$ of the problem for $m\geq 3$.
     
     % Hence, the lower bound of the online class cover problem can be expressed as $\max\{4, \lfloor \log_2 m\rfloor + 1\}$ for $m\geq 0$.
     
     Let $\P_r=\{r_1, r_2, \dots, r_m\}$ be a set of $m\ (\geq 3)$ red points in a plane lying on the line $x=y$ such that the distance between $r_1$ and $r_m$ is $\sqrt{2}$, and $r_i(x)< r_{i+1}(x)$ for $i\in [m-1]$. Let $p=\lfloor \log_2 m\rfloor$. To prove the lower bound, we construct an input sequence of $p + 1$ many blue points $b_1, b_2, \dots, b_{p+1} \in \P_b$ using an adaptive adversary; for which an offline optimal algorithm needs just one square, whereas any online algorithm needs at least $p+1$ many squares to cover them without containing any red point. See Figure~\ref{fig:lower_ex} for an example when $m=9$. For each $i\in [p+1]$, the adversary maintains a set $\P_r^{(i)} \subseteq \P_r$, and places the blue point $b_i$ inside a cell $C_i$ of the grid $\G(\P_r^{(i)})$. {Let $O_i$ be a square such that the left and top edges of $C_i$ are contained in the left and top edges of $O_i$, respectively.} We set $\P_r^{(0)}=\P_r$ and $O_0=\emptyset$. For each $i\in[p+1]$, let $R_i$ be a square placed by an online algorithm in $i$th iteration to cover $b_i$ and $R_0=\emptyset$. We construct the sequence $\{b_1, b_2, \dots, b_{p+1}\}$ inductively so that the following invariants will always be satisfied for $i = 1, 2, \dots, p+1$.
    \begin{itemize}
        \item[(1)] $\P_r^{(i)} \subseteq \P_r^{(i-1)}$ and $|\P_r^{(i)}|\geq \frac{m}{2^{i-1}}$.
        \item[(2)] Each grid cell of $\G(\P_r^{(i)})$ is cover-free with respect to the set $\{R_0,R_1,\dots,R_{i-1}\}$ of squares.
        \item[(3)] The blue point $b_i$ is placed inside the $SE$ corner cell, say $C_i$, of the grid $\G(\P_r^{(i)})$ such that $b_i\notin\cup_{j=0}^{i-1} R_j$.
        \item[(4)] All blue points $b_1,\dots,b_{i}$ are covered by the square $O_i$ that contains no red points.
    \end{itemize}

    In the first iteration, we consider $\P_r^{(1)}=\P_r$, and it is easy to see that invariant (1) holds. Now, the adversary places the blue point $b_1$ inside the $SE$ corner grid cell, say $C_1$, of $\G(\P_r^{(1)})$. Thus invariants (2) and (3) hold. Here, the square $O_1$ contains $b_1$. Therefore, invariant~(4) also holds.

    Let us assume that all the invariants hold for $i=1,2,\dots,k$ (where, $k \leq p$). Now, we will prove that the invariants hold for $i=k+1$.  Let $a$ be the $NW$ vertex of $R_k$. If $a$ does not lie on any grid edge, then it lies inside a cell of $\G(\P_r^{(k)})$. Let $G_{uv}$ be that cell. Otherwise, we define $u$ and $v$ as follows: $u = \max\{i\in [m-1] \ |\ a\in G_{ij}\}$ and $v = \max\{j\in [m-1]\ |\ a\in G_{ij}\}$. Suppose $v_k$ (respectively, $h_k$) is the vertical line (respectively, horizontal line) passing through the right edge (respectively, bottom edge) of $G_{uv}$  (see Figure~\ref{fig:hori_verti}). Let $V_k$ (respectively, $H_k$) be the half-plane lying on the left side (respectively, top side) of the line $v_k$ (respectively, $h_k$). 
    % Note that $a \in V_k\cap H_k$. 
    Now, the adversary considers two subsets $\R_V$ and $\R_H$ of $\P_r^{(k)}$ as follows: $\R_V=\{r_i\in\P_r^{(k)}\ |\ r_i$ lies on the half-plane $V_k \}$, and $\R_H=\{r_i\in\P_r^{(k)} \ |\ r_i$ lies on the half-plane $H_k\}$.

    \begin{itemize}

    \begin{figure}
        \centering
        \includegraphics[page=3, width=40mm]{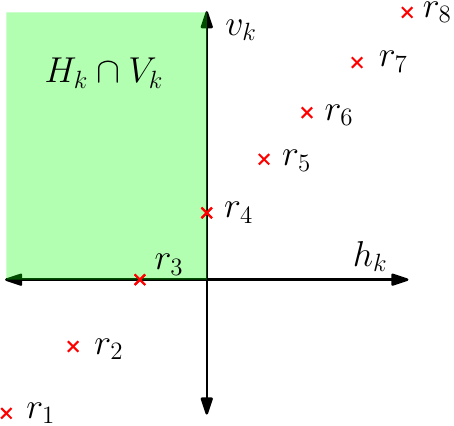}
        \caption{Here, $\R_V=\{r_1, r_2, r_3, r_4\}$ and $\R_H=\{r_3, r_4, r_5, r_6, r_7, r_8\}$. The green shaded region denotes the common region of the two half-planes $H_k$ and $V_k$. Also, $\R_V\cup\R_H=\P_r^{(k)}=\{r_1, r_2, \dots, r_8\}$.}
        \label{fig:hori_verti}
    \end{figure}

        \item  Let $\P_r^{(k+1)}$ be the set among $\R_V$ and $\R_H$ with the maximum size. Since $R_k$ contains no red points, the $NW$ vertex $a$ of $R_k$ must lie in $V_k\cap H_k$, i.e., $a \in V_k\cap H_k$. As a result, we have that for each $r\in \P_r^{(k)}$, either $r\in \R_V$ or $r\in \R_H$. Therefore, $\R_V\cup\R_H=\P_r^{(k)}$  (see Figure~\ref{fig:hori_verti}). Hence, either $|\R_V|\geq \frac{\P_r^{(k)}}{2}$ or $|\R_H|\geq \frac{\P_r^{(k)}}{2}$. This implies that $|\P_r^{(k+1)}|\geq \frac{\P_r^{(k)}}{2}\geq \frac{m}{2^{k}}$, and invariant~(1) holds.
        
        % If $|\R_1|\geq \frac{m}{2^{k}}$, then $|\P_r^{(k+1)}|\geq \frac{m}{2^{k}}$. Otherwise, let $|\R_1|<\frac{m}{2^{k}}$. Let $r\in \P_r^{k}\cap \overline{V_k}$ be a point nearest to the line $v_k$  (see Figure~\ref{fig:hori_verti}). Since $R_k$ does not contain any red points and $a \in V_k\cap H_k$, we get, $a(y)\leq r(y)$. Thus the horizontal line $h_k$ can not be above the point $r$. Hence, $\P_r^{k}\setminus\R_1\subseteq\R_2$. As a result of this, $|\R_2|\geq|\P_r^{k}\setminus\R_1|= |\P_r^{k}|-|\R_1|> \left(\frac{m}{2^{k-1}}-\frac{m}{2^k}\right) = \frac{m}{2^{k}}$. Therefore, $|\P_r^{(k+1)}|\geq \frac{m}{2^{k}}$ and invariant~(1) holds.

        \item Since $\P_r^{(k+1)}\subseteq\P_r^{(k)}$ and each grid cell of $\G(\P_r^{(k)})$ is cover-free with respect to the set $\{R_1, R_2,\dots,$ $R_{k-1}\}$, each of the grid cell of $\G(\P_r^{(k+1)})$ is also cover-free with respect to $\{R_1, R_2,\dots,R_{k-1}\}$. {Now, from the construction of $\P_r^{(k+1)}$, we can see that if $\P_r^{(k+1)} = \R_V$ (respectively, $\P_r^{(k+1)} = \R_H$), then all the cell $G_{iv}$ for $i\in [u]$ (respectively, $G_{uj}$ for $j\in [v]$) are the only cells  of the grid $\G(\P_r^{(k+1)})$ that may have some non-empty intersection with $R_k$.} But each of the cell $G_{iv}$ (respectively, $G_{uj}$) is cover-free with respect to the set $\{R_k\}$. As a result, each cell of the grid $\G(\P_r^{(k+1)})$ is cover-free with respect to $\{R_k\}$. Therefore, invariant~(2) holds.
        \item Due to invariant (2), a blue point $b_{k+1}$ can be placed inside the $SE$ corner cell, say $C_{k+1}$, of the grid $\G(\P_r^{(k+1)})$ such that $b_{k+1}\notin\cup_{j=0}^{k} R_j$. Therefore, invariant~(3) holds.
        \item Since $O_{k+1}$ covers $C_{k+1}$, we have, $b_{k+1}\in O_{k+1}$. Observe that if $\P_r^{(k+1)} = \R_V$ (respectively, $\P_r^{(k+1)} = \R_H$), then $C_{k}$ and $C_{k+1}$ will be in the same row (respectively, column) of the grid $\G(\P_r^{(k)})$. {Also, note that the height and width of the grid $\G(\P_r^{(k)})$ is at most $1$.} As a result of this, we have, $O_k\cap \G(\P_r) \subseteq O_{k+1}\cap \G(\P_r)$. Thus, $O_{k+1}$ also contains $b_1,\dots,b_k$. Therefore, invariant~(4) holds.
    \end{itemize}
    Hence, all the invariants also hold for $i=k+1$. Therefore, the theorem follows.
\end{proof}

% \vspace{2mm}\noindent\textbf{Remark:} Dumitrescu and Tóth \cite{dumitrescu2022online} proved that the competitive ratio of every deterministic online algorithm for the set cover problem using axis-parallel unit squares in a plane is at least~$4$. This lower bound also applies to the online class cover problem for $m\geq 0$. Hence, the lower bound of the online class cover problem can be expressed as $\max\{4, \lfloor \log_2 m\rfloor + 1\}$ for $m\geq 0$.

\section{An Online Algorithm}\label{sec:algo}
% In this section, we present an online algorithm for the class cover problem for the translates of a square. If a blue point $u$, introduced to the algorithm, is already covered by previously selected squares, the algorithm doesn't introduce additional squares in this step. Otherwise, we define at most five squares, called the \emph{candidate squares}, for the point $u$ as given in Section~\ref{sec:all_candidate_square}, and our algorithm places each of the squares to cover the point $u$. Though any of the candidate squares is sufficient to cover $u$, we need all the candidate squares to derive the bound on the competitive ratio. Before defining the set of candidate squares, we will define another set, called \emph{the set of staircase squares}, for a blue point $u$ in Section~\ref{sec:possible_candidate_square}, that will help us to determine the candidate squares later. In Section~\ref{sec:algo_desc}, we formally describe our algorithm.

In this section, we present an online algorithm for the class cover problem for the axis-parallel unit squars.
% the translates of a square. 
If a blue point $u$, introduced to the algorithm, is already covered by previously selected squares, the algorithm doesn't introduce additional squares in this step. Otherwise, we define at most five squares, called the \emph{candidate squares},  $\S(u)=\{R_1, R_2, R_3, R_4, R_5\}$ for the point $u$ as given in Section~\ref{sec:all_candidate_square}, and our algorithm places each of these candidate squares to cover the point $u$. Though any of the candidate squares is sufficient to cover $u$, we need all the candidate squares to derive the bound on the competitive ratio. Before illustrating the set of candidate squares, for a blue point $u$, we demonstrate an important notion of \emph{staircase squares} of $u$ in Section~\ref{sec:possible_candidate_square}. 
Then we illustrate how to construct three candidate squares $R_1, R_2$ and $R_3$ from the staircase squares in Section~\ref{sec:type1}. Later on, using this tool, we demonstrate how to obtain all five candidate squares in Section~\ref{sec:all_candidate_square}.
 In Section~\ref{sec:algo_desc}, we formally describe our algorithm. Moreover, we characterize the candidate squares in two types, \emph{Type 1} and \emph{Type 2}, depending on whether it is constructed by using the staircase squares or not. We will construct the candidate squares so that $R_1, R_2$ and $R_3$ are always Type 1 candidate squares, and $R_4$ and $R_5$ are always Type 2 candidate squares.

Let $u$ be a blue point, and $R(u) =\square abcd$ be the square centred at $u$. Let $R(u)_{NW}=\square apus,$ $ R(u)_{SW}=\square pbqu,$ $R(u)_{SE}=\square uqcr$ and $R(u)_{NE}=\square surd$ be the four sub-squares of $R(u)$ (as defined in Section~\ref{sec:nota}). The construction of the candidate squares solely relies on the following observation.

\begin{observation}\label{obs:main}
    Let $u$ be any blue point and $P$ be any square such that $u\in P$. Then, $P$ must contain at least one sub-square of $R(u)$.
\end{observation}

% \noindent From Observation~\ref{obs:main}, we can see that if the blue point $u$ can be covered by a red point free optimum square, then at least one of the sub-square of $R(u)$ is red point free. Based on this fact, we define the candidate squares in such a way that each of them contains the blue point $u$ and at least one red point free sub-square of $R(u)$.
% and also satisfy some addition properties with respect to the optimum square $\R$ (see Lemma~\ref{lema:candidate_squares}). 

\subsection{Staircase Squares}\label{sec:possible_candidate_square}
Consider a blue point $u$ and the square $R(u)$ centered at it. Suppose that $R(u)_{SW}$ contains some red points, but $R(u)_{NE}$ is red point free. Now, if $R(u)_{NW}\cap \P_r\neq \emptyset$, let $r_1 \in R(u)_{NW}\cap \P_r$ be a red point nearest to the line segments $\overline{us}$; otherwise, let $r_1=a$, the $NW$ vertex of $R(u)$. If $R(u)_{SE}\cap \P_r\neq \emptyset$, let $r_2 \in R(u)_{SE}\cap \P_r$ be a red point nearest to the line segments $\overline{ur}$; otherwise, let $r_2=c$, the $SE$ vertex of $R(u)$. Let $R'$ be a square whose left edge and bottom edge pass through the points $r_1$ and $r_2$, respectively (see Figure~\ref{fig:possible_square_staircase}).

    \begin{figure}[htbp]
    \centering
    \begin{subfigure}[b]{.49\textwidth}
        \centering
        \includegraphics[page=1, width=75mm]{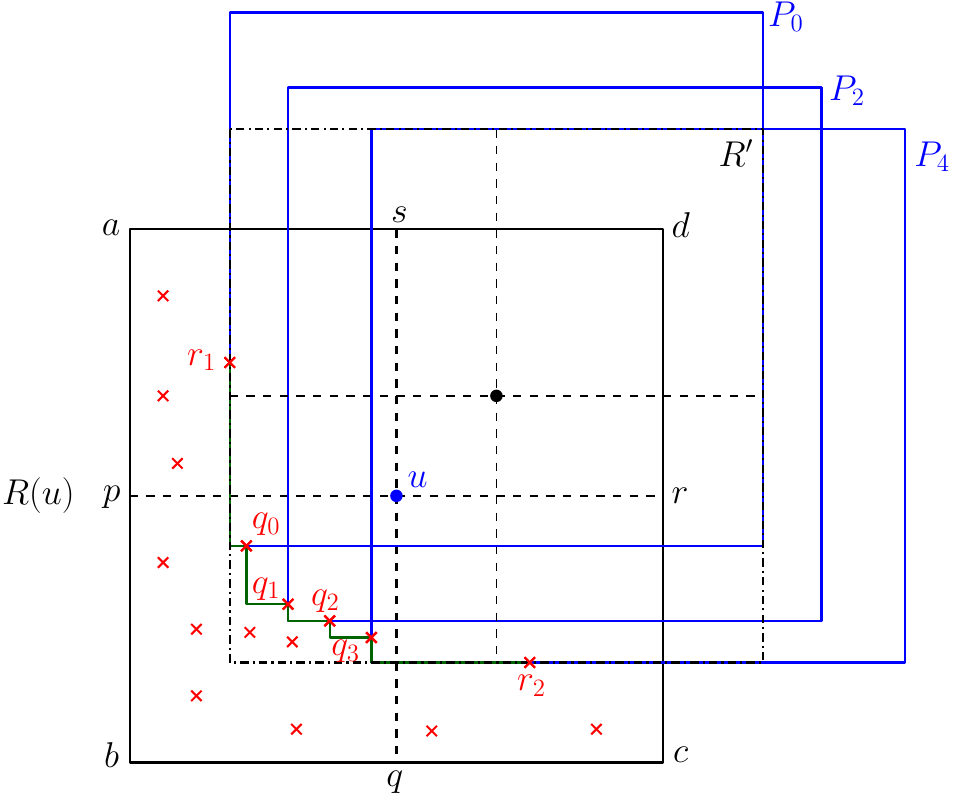}
        \caption{}
        \label{fig:possible_square_staircase}
    \end{subfigure}
    \hfill
    \begin{subfigure}[b]{.49\textwidth}
        \centering
        \includegraphics[page=3, width=75mm]{Candidate_Squares_Construction.pdf}
        \caption{}
        \label{fig:unit_square_case1.1}
    \end{subfigure}
    \caption{Let $R(u)$ (black) be the square centered at the blue point $u$. (a) The sequence of $SW$ staircase points $\Q(u)_{SW}=(r_1,q_0,\dots,q_3,r_2)$ forms the staircase $ST_{SW}(\Q(u)_{SW})$ (dark green), and $P_0,P_2,P_4$ (blue) are some staircase squares of $u$. Also, $u\in P_0,P_2,P_4$. (b) Case 1.1: $R'_{SW}$ contains no red points. Here, $R_4$ is a candidate square (blue) of $u$ obtained by moving $R'$ (black/dashed).}
    \label{fig:unit_square_cases1}
    \end{figure}

Let $Q'=R(u)_{SW}\cap R'_{SW}\cap \P_r$, and $Q''=D_{SW}(Q')$, the set of $SW$ dominating points of $Q'$. Now, consider a sequence $(q_0, q_1, \dots, q_k)$ of the elements of $Q''$ such that $q_i(x)\leq q_{i+1}(x)$ and $q_i(y)\geq q_{i+1}(y)$ for all~$i\in \{0\}\cup [k-1]$. Let $\Q(u)_{SW}$ be the sequence $(r_1,q_0, q_1, \dots, q_k,r_2)$,  and we call it \emph{the sequence of $SW$ staircase points} of $u$. With slight abuse of notation, we also use the term $\Q(u)_{SW}$ to denote the set $\{r_1,q_0, q_1, \dots, q_k,r_2\}$. The set $\Q(u)_{SW}$ forms a staircase $ST_{SW}(\Q(u)_{SW})$ (see Figure~\ref{fig:possible_square_staircase}). Here, the two points $r_1$ and $r_2$ are the initial and the terminal point of $ST_{SW}(\Q(u)_{SW})$, respectively.
% Similarly, we can define the set of staircase points $\Q_{NW}(u)$ (respectively, $\Q_{NE}(u)$ and $\Q_{SE}(u)$) and the staircase $ST_{NW}(\Q_{NW}(u))$ (respectively, $ST_{NE}(\Q_{NE}(u))$ and $ST_{SE}(\Q_{SE}(u))$) when $R(u)_{NW}$ (respectively, $R(u)_{NE}$ and $R(u)_{SE}$) contains some red points.

For each pair of consecutive red points $(q_j, q_{j+1})$, where $0\leq j\leq k-1$, we have a square $P_{j+1}$ such that $q_j$ and $q_{j+1}$ lie on the left and the bottom edge of $P_{j+1}$, respectively. We also define two squares $P_0$ and $P_{k+1}$ such that $r_1$ and $q_0$ lie on the left edge and bottom edge of $P_0$, and $q_k$ and $r_2$ lie on the left edge and bottom edge of $P_{k+1}$, respectively.

\begin{observation}\label{obs:Pj_SW}
    For each $j \in \{0\}\cup [k+1]$, the square $P_j$ contains the point $u$, and the sub-square $P_{j(SW)}$ contains no red point.
\end{observation}

Let us assume that $T=\{P_j\ |\ P_{j(NE)}\ \text{contains no red points for}\ j \in \{0\}\cup [k+1]\}$. Now, consider a sequence $\mathcal{PS}(u)_{SW}=(P_{j_1}, P_{j_2}, \dots, P_{j_{l}})$ of the squares of $T$ such that $(j_1,j_2,\dots,j_{l})$ is a sub-sequence of the sequence $(0,1,\dots,k+1)$ (see Figure~\ref{fig:possible_square_staircase}). We call $\mathcal{PS}(u)_{SW}$ \emph{the sequence of $SW$ staircase squares} of $u$. With slight abuse of notation, we also use the term $\mathcal{PS}(u)_{SW}$ to denote the set $\{P_{j_1}, P_{j_2}, \dots, P_{j_{l}}\}$. Also, each square $P_{j_i}\in \mathcal{PS}(u)_{SW}$ will be called a \emph{staircase square} of $u$. Similarly, we can define $\Q(u)_{NW}, \mathcal{PS}(u)_{NW}$ and so on.

% \subsection{Candidate Squares}\label{sec:candidate_square}
% Let $\S(u)=\{R_1, R_2, R_3, R_4, R_5\}$ be the set containing all the candidate squares of $u$. Moreover, we characterize the candidate squares in two types, \emph{Type 1} and \emph{Type 2}, depending on whether it is constructed by using the set of staircase squares or not. We will construct the candidate squares so that $R_1, R_2$ and $R_3$ are always Type 1 candidate squares, and $R_4$ and $R_5$ are always Type 2 candidate squares. This will help us establish the competitive ratio of our algorithm.

\subsection{Construction of Type~1 Candidate Squares}\label{sec:type1}
In this section, we demonstrate how to obtain Type~1 candidate squares for a blue point $u$ from the sequence of staircase squares of $u$. W.l.o.g, assume that $R(u)_{SW}$ contains some red points, but $R(u)_{NE}$ is red point free. So, as defined in Section~\ref{sec:possible_candidate_square}, consider the sequence of $SW$ staircase points $\Q(u)_{SW}=(r_1,q_0, \dots, q_k,r_2)$ and the sequence of $SW$ staircase squares $\mathcal{PS}(u)_{SW}=(P_{1}, P_{2}, \dots, P_l)$ of $u$.

    \begin{algorithm}[htbp]
    \caption{\UR$(u, H)$}\label{sub:1}
    \begin{algorithmic}[1]
    \Require{A blue point $u$, and a translated copy $H$ of $R(u)$ such that $H_{SW} \subseteq R(u)$.}
    \Ensure{A square $P$.}
    \While{$(H_{NW}\cup H_{SE})$ is not red point free and $H$ contains $u$}
        \If{$H_{SE}$ contains some red points}
            \State Move $H$ upwards until $H_{SE}$ becomes red point free
        \EndIf
        \If{$H_{NW}$ contains some red points}
            \State Move $H$ rightwards until $H_{NW}$ becomes red point free
        \EndIf
    \EndWhile
    \If{$H$ is red point free and $u\in H$}
        \State Return $P=H$
    \Else \State Return $P=\emptyset$
    \EndIf
    \end{algorithmic}
    \end{algorithm}

    \begin{algorithm}[htbp]
    \caption{\STR$(u, H)$}\label{sub:2}
    \begin{algorithmic}[1]
    \Require{A blue point $u$, and a translated copy $H$ of $R(u)$ such that the $SW$ vertex of $H$ lies on the staircase $ST_{SW}(\Q(u)_{SW})$.}
    \Ensure{A square $P$.}
    \If{$H_{NW}$ contains some red points}
        \State Move $H$ rightwards along the staircase $ST_{SW}(\Q(u)_{SW})$ until $(H_{NW}\cup H_{NE})$ becomes red point free or does not contain $u$
    \EndIf
    % \While{$H$ is not red point free or $H$ does not contain $u$}    
    \If{$H_{SE}$ contains some red points and $u\in H$}
        \State Return $P=$\UR$( u, H)$
        \State Break
    % \EndWhile
    \ElsIf{$H$ is red point free and $u\in H$}
        \State Return $P=H$
    \Else \State Return $P=\emptyset$
    \EndIf
    \end{algorithmic}
    \end{algorithm}

If $\mathcal{PS}(u)_{SW}=\emptyset$, then we set $R_i=\emptyset$ for each $i\in [3]$. W.l.o.g, assume that $\mathcal{PS}(u)_{SW}$ is non-empty. Let $H_1 =P_1, H_2 = P_{\lceil \frac{l}{2}\rceil}$ and $H_3 = P_l$ (they may not be distinct). Note that $H_{i(SW)}\cup H_{i(NE)}$ does not contain any red point. Now, if $H_{i(SE)}\cup H_{i(NW)}$ contains no red point, we set $R_i=H_i$. Otherwise, depending on which sub-squares among $H_{i(SE)}$ and $H_{i(NW)}$ contain red point, we define $R_i$ as follows: (1) If both $H_{i(SE)}$ and $H_{i(NW)}$ contain some red points, we set $R_i=\UR(u, H_i)$\footnote{Here, U and R stand for upwards and rightwards, respectively.} as defined in Algorithm~\ref{sub:1}, and it works as follows. If $H_{i(NW)}\cup H_{i(SE)}$ contains some red points, first, we move $H_i$ upwards until $H_{i(SE)}$ becomes red point free, and then move it rightwards until $H_{i(NW)}$ contains no red point. Repeat the process until $H_i$ becomes red point free or does not contain $u$. In the former case, we set $R_i=H_i$. In the latter case, we set $R_i=\emptyset$. (2) Now, consider the case when exactly one of $H_{i(NW)}$ and $H_{i(SE)}$ contains some red points. W.l.o.g., assume that only $H_{i(NW)}$ contains some red points. The other case is similar. In this case, we set $R_i=\STR(u, H_i)$\footnote{Here, ST and R stand for staircase and rightwards, respectively.} as defined in Algorithm~\ref{sub:2}, and it works as follows. First, we move $H_i$ rightwards along the staircase $ST_{SW}(\Q(u)_{SW})$ until $H_{i(NW)}\cup H_{i(NE)}$ becomes red point free or does not contain $u$. After that, if $H_{i(SE)}$ is red point free and $u\in H$, we set $R_i=H_i$, otherwise, if $H_{i(SE)}$ contains some red points and $u\in H$, we set $R_i=\UR(u, H_i)$ as defined in Algorithm~\ref{sub:1}.

\subsection{Construction of All Candidate Squares}\label{sec:all_candidate_square}

In this section, we define all the candidate squares of $u$ based on which sub-squares of $R(u)$ contain red points. If all four sub-squares of $R(u)$ contain red points, then due to Observation~\ref{obs:main}, there does not exist any feasible solution. As a result, we set $R_i=\emptyset$ for each $i\in [5]$. Hence, without loss of generality, we assume that at least one sub-square of $R(u)$ does not contain any red point, say $R(u)_{NE}$. Now, if the other three sub-squares of $R(u)$ do not contain any red point, we set $R_4=R(u)$ and $R_1=R_2=R_3=R_5=\emptyset$. Thus, $\S(u)=\{R(u)\}$. Now, depending on which sub-square of $R(u)$, except $R(u)_{NE}$, contains some red points, we have the following cases. See Figure~\ref{fig:Candidate_Cases} for illustration.

% The squares $R_1, R_2, R_3, R_4$ and $R_5$ will be called the upper, the corner, the lower, the top side and the right side candidate square of $u$.

\begin{figure}
    \centering
    \includegraphics[page=2, width=90mm]{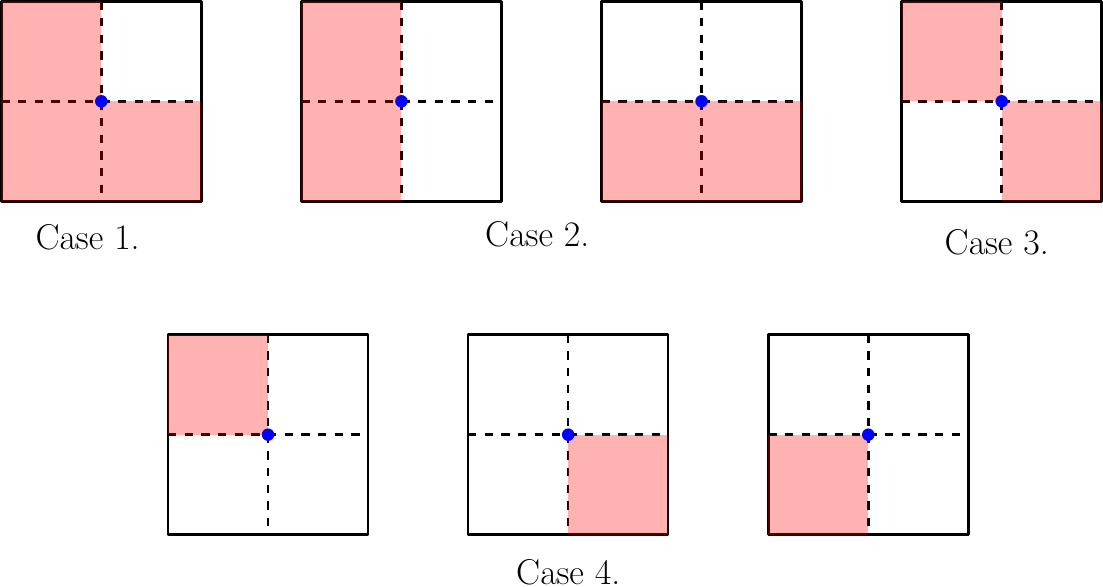}
    \caption{Let $R(u)$ (black) be the square centered at a blue point $u$. The Red shaded sub-squares of $R(u)$ contain some red points.}
    \label{fig:Candidate_Cases}
\end{figure}

\begin{itemize}
    % \item \textbf{Case~1: All four sub-squares of $R(u)$ contain red points.} In this case, finding a feasible solution is impossible due to Observation~\ref{obs:main}. As a result of this, there does not exist any square that covers $u$ but does not contain any red points. So, we set $R_i=\emptyset$ for $i\in [5]$.

    \item \textbf{Case~1: All other sub-squares of $R(u)$ contain red points.} See Figure~\ref{fig:unit_square_case1.1}. Let $r_1 \in R(u)_{NW}\cap \P_r$ and $r_2 \in R(u)_{SE}\cap \P_r$ be two red points nearest to the line segments $\overline{us}$ and $\overline{ur}$, respectively. Let $R'$ be a square whose left edge and bottom edge pass through the points $r_1$ and $r_2$, respectively. No feasible solution exists if $R'_{NE}$ contains some red points. In this case, we set $R_i=\emptyset$ for each $i\in [5]$. So, without loss of generality, assume that $R'_{NE}$ contains no red points. Depending upon whether $R'_{SW}$ contains some red points or not, we have the following two sub-cases.

    \begin{figure}[htbp]
    \centering
    \begin{subfigure}[b]{.49\textwidth}
        \centering
        \includegraphics[page=4, width=75mm]{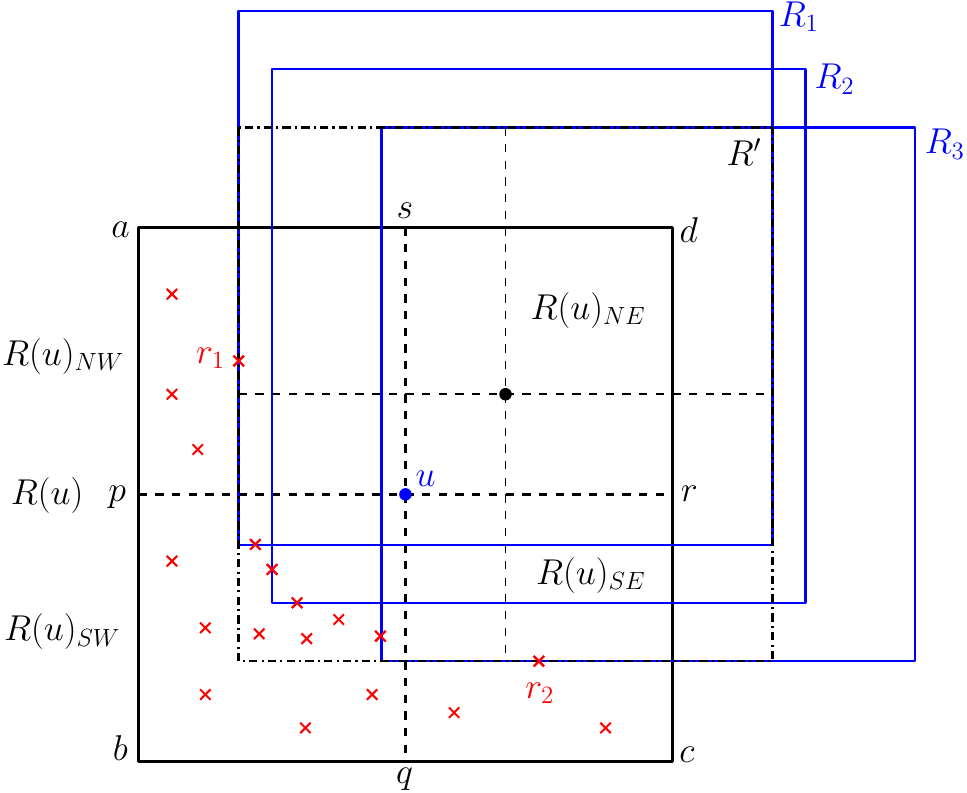}
        \caption{}
        \label{fig:unit_square_case1.2}
    \end{subfigure}
    \hfill
    \begin{subfigure}[b]{.49\textwidth}
        \centering
        \includegraphics[page=5, width=70mm]{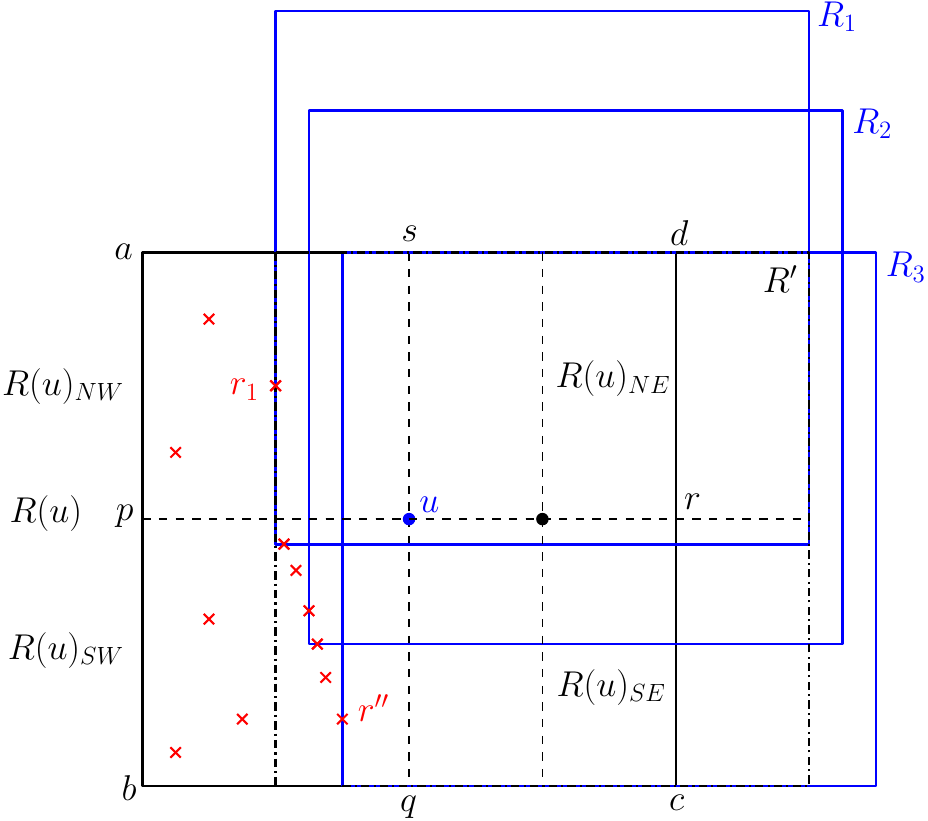}
        \caption{}
        \label{fig:unit_square_case2.1}
    \end{subfigure}
    \caption{Let $R(u)$ (black) be the square centered at a blue point $u$. (a) Case 1.2: $R'_{SW}$ contains some red points. Here, $R_1$, $R_2$ and $R_3$ are the candidate squares (blue) of $u$ constructed by using $R'$ (black/dashed). (b) Case 2: $R(u)_{NW}$ and $R(u)_{SW}$ both contain some red points. Here, $R_1$, $R_2$ and $R_3$ are the candidate squares (blue) of $u$ constructed by using $R'$ (black/dashed).}
    \label{fig:unit_square_cases2}
    \end{figure}

    \hspace{5mm} \emph{\underline{Case 1.1:}} $R'_{SW}$ contains no red point (see Figure~\ref{fig:unit_square_case1.1}). In this case, we set $R_1=R_2=R_3=R_5=\emptyset$, and $R_4=\UR(u, R')$ as defined in Algorithm~\ref{sub:1}. 
    % In other words, if $R'_{NW}\cup R'_{SE}$ contains some red points, first we move $R'$ rightwards until $R'_{NW}$ contains no red points, and then move it upwards until $R'_{SE}$ contains no red points. Repeat the process until $R'$ becomes red point free or does not contain $u$. In the former case, we set $R_4=R'$ and $\S(u)=\{R_4\}$. In the latter case, we set $R_4=\emptyset$.

    % Now a point $q\in Q'$ is said to be a dominating point if there does not exist any point $q'\in Q'$  such that $q(x)<q'(x)$ and $q(y)<q'(y)$. Consider the set $Q''$ containing all the dominating points of $Q'$.

    \hspace{5mm} \emph{\underline{Case 1.2:}} $R'_{SW}$ contains some red points (see Figure~\ref{fig:unit_square_case1.2}). In this case, we set $R_4=R_5=\emptyset$. Since $R(u)_{SW}\cap \P_r \neq \emptyset$ and  $R(u)_{NE}\cap \P_r = \emptyset$, we defined the candidate squares $R_1, R_2$ and $R_3$ as defined in Section~\ref{sec:type1}.

    \item \textbf{Case~2: $R(u)_{SW}$, and one of $R(u)_{NW}$ and $R(u)_{SE}$ contain red points.} Assume, without loss of generality, $R(u)_{NW}$ contains some red points (see Figure~\ref{fig:unit_square_case2.1}). The other case, i.e., $R(u)_{SE}$ contains some red points, is similar. Here, we set $R_4=\emptyset$. Let $r'' \in (R(u)_{NW} \cup R(u)_{SW})\cap \P_r$ be a red point nearest to the line segments $\overline{qs}$. If $r''\in R(u)_{SW}$, since $R(u)_{SW}\cap \P_r \neq \emptyset$ and  $R(u)_{NE}\cap \P_r = \emptyset$, we define the candidate squares $R_1, R_2$ and $R_3$ as defined in Section~\ref{sec:type1}. If $r''\in R(u)_{NW}$, since $R(u)_{NW}\cap \P_r \neq \emptyset$ and  $R(u)_{SE}\cap \P_r = \emptyset$, we define the candidate squares $R_1, R_2$ and $R_3$ in a similar way as defined in Section~\ref{sec:type1}.
    Next, we define the candidate square $R_5$ irrespective of whether $r''\in R(u)_{SW}$ or $r''\in R(u)_{NW}$ (see Figure~\ref{fig:unit_square_case2.2}). Let $H_5$ be a horizontal translated copy of $R(u)$ whose left edge passes through the point $r''$. Observe that $H_{5(NW)} \cup H_{5(SW)}$ contains no red point. If both $H_{5(NE)}$ and $H_{5(SE)}$ do not contain any red point, we set $R_5=H_5$, and if both sub-squares contain some red points, we set $R_5=\emptyset$. 
    % Otherwise, depending on which sub-squares among $H_{5(NE)}$ and $H_{5(SE)}$ contains some red points, we move $H_5$ as follows: (1) If $H_{5(NE)}$ and $H_{5(SE)}$ both sub-squares contain some red points, we set $R_5=\emptyset$. (2)~
    Now, consider the case when exactly one of $H_{5(NE)}$ and $H_{5(SE)}$ contains some red points. W.l.o.g., assume that only $H_{5(SE)}$ contains some red points. The other case is similar. In this case, we set $R_5=\UR(u, H_5)$ as defined in Algorithm~\ref{sub:1}.

    \begin{figure}[htbp]
    \centering
    \begin{subfigure}[b]{.49\textwidth}
        \centering
        \includegraphics[page=6, width=85mm]{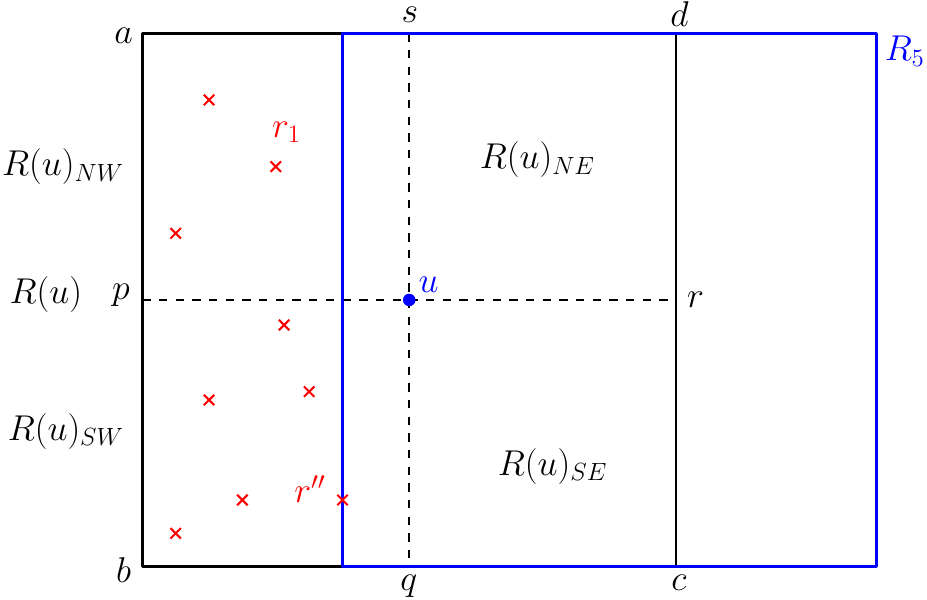}
        \caption{}
        \label{fig:unit_square_case2.2}
    \end{subfigure}
    \hfill
    \begin{subfigure}[b]{.49\textwidth}
        \centering
        \includegraphics[page=7, width=57mm]{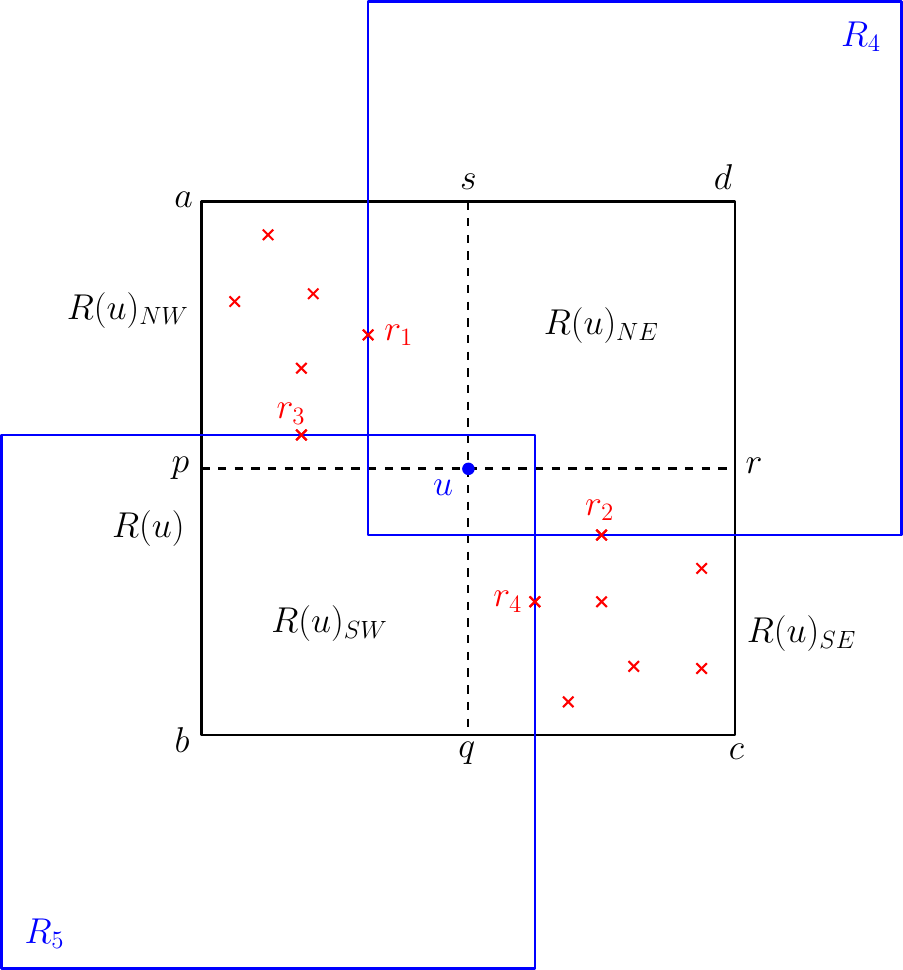}
        \caption{}
        \label{fig:unit_square_case3}
    \end{subfigure}
    \caption{Let $R(u)$ (black) be the square centered at a blue point $u$. (a) Case 2: $R(u)_{NW}$ and $R(u)_{SW}$ both contain some red points. Here, $R_5$ is a candidate square (blue) of $u$. (b) Case~3: $R(u)_{NW}$ and $R(u)_{SE}$ both contain some red points. Here, $R_4$ and $R_5$ are the two candidate squares (blue) of $u$.}
    \label{fig:unit_square_cases3}
    \end{figure}

    \item \textbf{Case~3: Both $R(u)_{NW}$ and $R(u)_{SE}$  contain red points.} See Figure~\ref{fig:unit_square_case3}. For this case, we set $R_1=R_2=R_3=\emptyset$. Let $r_1 \in R(u)_{NW}\cap \P_r$ and $r_2 \in R(u)_{SE}\cap \P_r$ be two red points nearest to the line segments $\overline{us}$ and $\overline{ur}$, respectively. Let $H_4$ be a square whose left and bottom edges pass through the points $r_1$ and $r_2$, respectively. Observe that $H_{4(SW)}$ contains no red point. If $H_{4(NE)}$ contains some red points, we set $R_4=\emptyset$. Otherwise, $H_{4(SE)} \cup H_{4(NW)}$ may contain some red points. Then, we set $R_4=\UR(u, H_4)$ as defined in Algorithm~\ref{sub:1}.
        
    \hspace{5mm} Next, we define the candidate square $R_5$. Let $r_3 \in R(u)_{NW}\cap \P_r$ and $r_4 \in R(u)_{SE}\cap \P_r$ be two red points nearest to the line segments $\overline{up}$ and $\overline{uq}$, respectively (see Figure~\ref{fig:unit_square_case3}). Let $H_5$ be a square whose top and right edges pass through the points $r_3$ and $r_4$, respectively. Observed that $H_{5(NE)}$ contains no red points. If $H_{5(SW)}$ contains some red points, we set $R_5=\emptyset$. Otherwise, $H_{5(SE)} \cup H_{5(NW)}$ may contain some red points. Then, we set $R_5=\DL(u, H_5)$\footnote{Here, D and L stand for downwards and leftwards, respectively.}, where {\DL} is similar to \UR, move the square downwards and leftwards instead of rightwards and upwards, respectively.

\begin{figure}[htbp]
    \centering
    \begin{subfigure}[b]{.49\textwidth}
        \centering
        \includegraphics[page=8, width=75mm]{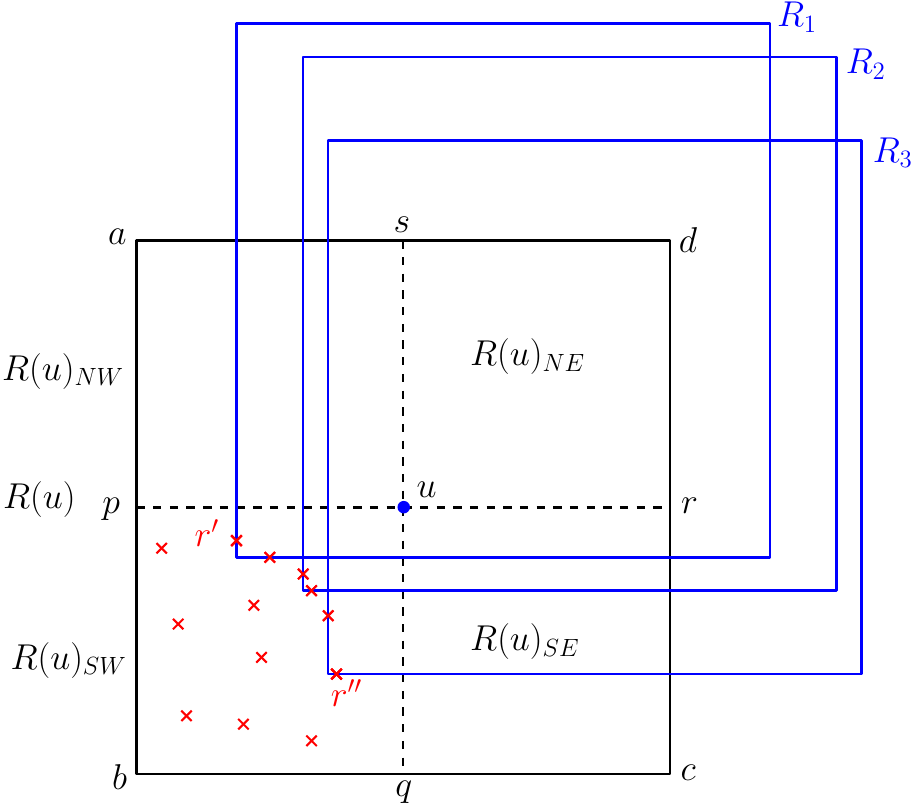}
        \caption{}
        \label{fig:unit_square_case4.1}
    \end{subfigure}
    \hfill
    \begin{subfigure}[b]{.49\textwidth}
        \centering
        \includegraphics[page=9, width=72mm]{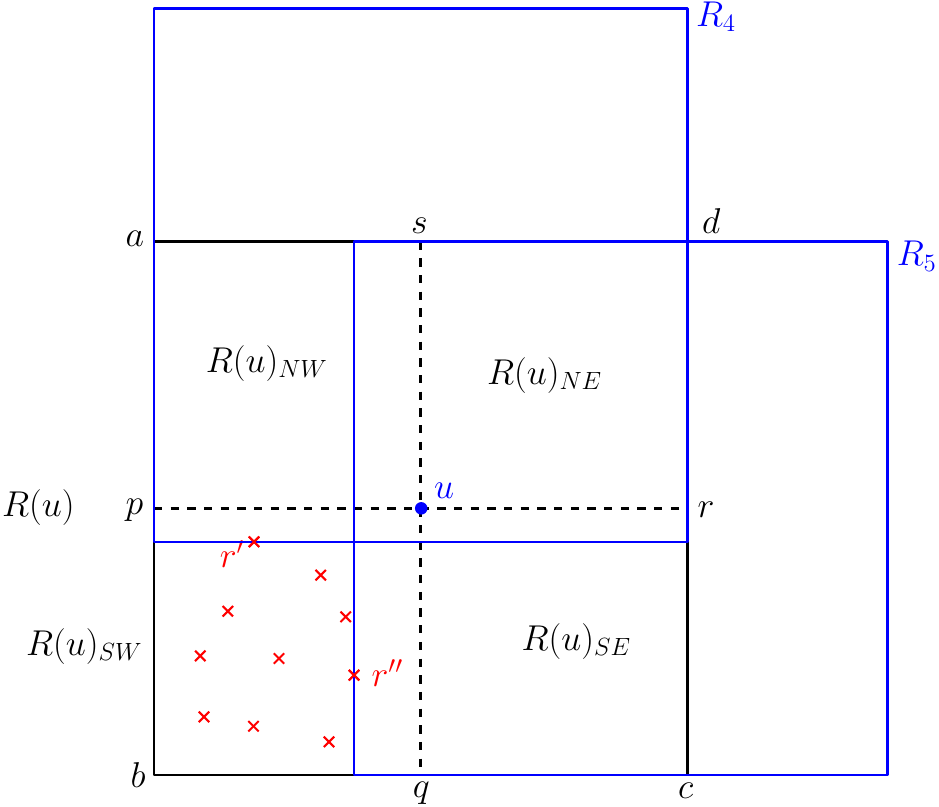}
        \caption{}
        \label{fig:unit_square_case4.2}
    \end{subfigure}
    \caption{Let $R(u)$ (black) be the square centered at a blue point $u$. Case 4: $R(u)_{SW}$ contains some red points. (a) $R_1$, $R_2$ and $R_3$ are the candidate squares (blue) of $u$ constructed by using $R'=R(u)$. (b) $R_4$ and $R_5$ are the two candidate squares (blue) of $u$.}
    \label{fig:unit_square_cases4}
\end{figure}

    \item \textbf{Case~4: One of $R(u)_{NW}$, $R(u)_{SW}$ and $R(u)_{SE}$ contains red points.} Assume, without loss of generality, $R(u)_{SW}$ contains some red points (see Figure~\ref{fig:unit_square_case4.1}). The other cases, i.e., $R(u)_{NW}$ or $R(u)_{SE}$ contains some red points, are similar. Let $r', r'' \in R(u)_{SW}\cap \P_r$ be two red points nearest to the line segments $\overline{up}$ and $\overline{uq}$, respectively. If $r'= r''$, we set the candidate squares $R_1$, $R_2$ and $R_3$ to be $\emptyset$. Otherwise, since $R(u)_{SW}\cap \P_r \neq \emptyset$ and  $R(u)_{NE}\cap \P_r = \emptyset$, we obtain the candidate squares $R_1, R_2$ and $R_3$ as defined in Section~\ref{sec:type1}. Next, we define $R_4$ and $R_5$ irrespective of whether $r'= r''$ or not (see Figure~\ref{fig:unit_square_case4.2}). Let $H_4$ and $H_5$ be two vertical and horizontal translated copies of $R(u)$ such that the bottom edge of $H_4$ passes through $r'$ and the left edge of $H_5$ passes through $r''$. Observe that $H_{5(NW)} \cup H_{5(SW)}$ contains no red point. If both $H_{5(NE)}$ and $H_{5(SE)}$ do not contain any red point, we set $R_5=H_5$; if both of them contain some red points, we set $R_5=\emptyset$. 
    % Otherwise, depending on which sub-squares among $H_{5(NE)}$ and $H_{5(SE)}$ contain red points, we move $H_5$ as follows: (1) If $H_{5(NE)}$ and $H_{5(SE)}$ both sub-squares contain some red points, we set $R_5=\emptyset$. (2)~
    Now, consider the case when exactly one of $H_{5(NE)}$ and $H_{5(SE)}$ contains some red points. W.l.o.g., assume that only $H_{5(SE)}$ contains some red points. The other case is similar. In this case, we set
    % $R_5=\text{U(SE)-R(NW)/L(NE)-MOVE}(\langle u, H_5\rangle)$\footnote{Here, U, R and L stand for upwards, rightwards and leftwards, respectively.} as defined in Algorithm~\ref{sub:4}. In other words, we move $H_5$ upwards until $H_{5(SE)}$ becomes red point free. Then, if $H_{5(NW)}$ and $H_{5(NE)}$ both contain some red points, we set $R_5=\emptyset$. If only $H_{5(NW)}$ contains some red points, we set
    $R_5=\UR(u, H_5)$ as defined in Algorithm~\ref{sub:1}.
    % $H_5$ rightwards until $H_{5(NW)}$ becomes red point free,
    % and if only $H_{5(NE)}$ contains some red points, we set $R_5=\text{U(SW)-L(NE)-Move}(\langle u, H_5\rangle)$\footnote{Here, U and L stand for upwards and leftwards, respectively, and U(SW)-L(NE)-Move is defined in a similar way as U(SE)-R(NW)-Move defined in Algorithm~\ref{sub:1}.}.
    % move $H_5$ leftwards until $H_{5(NE)}$ becomes red point free. Repeat the process until $H_5$ becomes red point free or does not contain $u$. In the former case, we set $R_5=H_5$; in the latter case, we set $R_5=\emptyset$. 
    Similarly, we define the candidate square $R_4$, which will be constructed using $H_4$.

    \end{itemize}

From the above description, we have that $0\leq |\S(u)|\leq 5$ for any blue point $u$, and each (non-empty) candidate square of $u$ always contains $u$ but no red points. Later, we will show that if a blue point $u$ can be covered by a square in the optimum, at least one candidate square exists, i.e., $\S(u) \neq \emptyset$.

\subsection{Algorithm Description}\label{sec:algo_desc}
The algorithm always maintains a set $\S_i$ of $\P_r$-empty squares to cover all the blue points $\{u_1, u_2,\ldots, u_i\}$ that have been part of the input so far. Note that, initially, $\S_0 = \emptyset$. The description of the algorithm is given in Algorithm~\ref{alg:occp}.

\begin{algorithm}[htbp]
\caption{An online algorithm for the class cover problem}\label{alg:occp}
\begin{algorithmic}[1]
\State $\S_0\leftarrow \emptyset$
\For {$i = 1$ to $\infty$} \Comment{arrival of a blue point $u_i$}
    \If {$u_i$ is not covered by $\S_{i-1}$}
        \State $\S_i \leftarrow \S_{i-1}\cup \S(u_i)$ \Comment{$\S(u_i)$ is the set of candidate squares of $u_i$}
    \Else
    \State $\S_i\leftarrow \S_{i-1}$
    \EndIf
\EndFor
\end{algorithmic}
\end{algorithm}

\section{Analysis of the Online Algorithm}\label{sec:ana_algo}
 
 In this section, we analyse the performance of our algorithm, i.e., Algorithm~\ref{alg:occp}.
First, in Section~\ref{sec:algo_analy}, we give the correctness of the algorithm. 
% Let $C(m)$ be the competitive ratio of the algorithm where $|\P_r|=m$. 
We mention some important structural properties of candidate squares in Section~\ref{sec:algo_analy}.
% in Lemma~\ref{lema:candidate_squares}.
Using these properties, in Section~\ref{sec:comp_analy}, we give a simpler analysis to show that the competitive ratio of our algorithm is at most $\max\{6, 20\log_2(m)\}$ for $m\geq 1$. Finally, in Section~\ref{sec:improved_comp}, for $m\geq 2$, we give a tighter analysis to show that the competitive ratio is at most $10+10\log_2 (m)$.

\subsection{Correctness of the Algorithm}\label{sec:algo_analy} 
Let OPT$(\I)$ be an optimal solution of the class cover problem given by an offline algorithm for an input sequence $\I$ of blue points. Let $\O =\square o_1o_2o_3o_4$ be a square in OPT$(\I)$ centered at $o$. Let $o_5, o_6, o_7$ and $o_8$ are the midpoints of the edges $o_1o_2, o_2o_3, o_3o_4$ and $o_4o_1$ of $\O$, respectively.
% Then, the four sub-squares of $\O$ are $\O_{NW}=\square a_1b_1ob_4, \O_{SW}=\square a_2b_2ob_1, \O_{SE}=\square a_3b_3ob_2$ and $\O_{NE}=\square a_4b_4ob_3$.
Let $u$ be an input blue point lying inside $\O$. Without loss of generality, assume that $u \in \O_{SW}$. 
First, consider the following lemma.

\begin{lemma}\label{lm:type 2 existance}
    Let $u$ be an input blue point in $\O_{SW}$. Also, let $H$ be a square that contains $\O\cap R(u)$, and the set $R(u)\cap H$ contains no red point.
    % and $H_{SW}$ contains no red point. Also assume that each red point in $R(u)_{NW}$ (respectively, $R(u)_{SE}$) lies on the left side of the left edge (respectively, below the bottom edge) of $H$. 
    Now, if $P=\UR(u, H)$, then $P\neq \emptyset$, and the square $P$ contains $\O_{SW}$.
\end{lemma}
    \begin{proof}
        First, observe that $H_{NE}$ contains no red points as $H_{NE}\subseteq \O$ (see Figure~\ref{fig:type 2 existance}). So, only $H_{NW}\cup H_{SE}$ may contain some red points. Now, as the left edge of $H$ lies to the left side of the left edge of $\O$ and the square $\O$ is red point free, we have that $H$ will be moved rightwards until $H_{NW}$ becomes red point free or the left edge of $H$ touches the left edge of $\O$. Similarly, since the bottom edge of $H$ lies below the bottom edge of $\O$ and the square $\O$ is red point free, we have that $H$ will be moved upwards until $H_{SE}$ becomes red point free or the bottom edge of $H$ touches the bottom edge of $\O$. Also, observe that during the movement of $H$, the sub-squares $H_{NE}$ and $H_{SW}$ never contain a red point. Therefore, during the movement of $H$, we are guaranteed to find a square that contains $u$ but no red points, and we set such a square as an output square $P$. Hence, $P\neq \emptyset$. Note that during the movement of $H$, the left edge of $H$ always lies to the left side of the left edge of $\O$, and the bottom edge of $H$ always lies below the bottom edge of $\O$. Therefore, $H$ always contains $\O_{SW}$, and thus, $P$ also contains $\O_{SW}$.
    \end{proof}

  \begin{figure}[htbp]
    \centering
    \begin{subfigure}[b]{.49\textwidth}
        \centering
        \includegraphics[page=1, width=75mm]{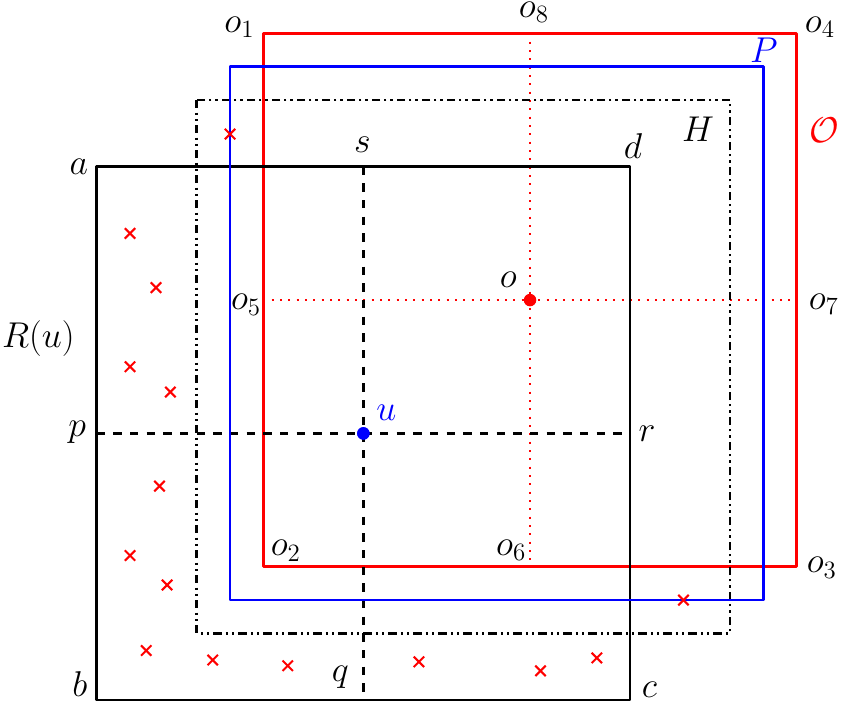}
        \caption{}
        \label{fig:type 2 existance}
    \end{subfigure}
    \hfill
    \begin{subfigure}[b]{.49\textwidth}
        \centering
        \includegraphics[page=2, width=75mm]{OPT_Squares_Existance.pdf}
        \caption{}
        \label{fig:opt_square_case_1.1}
    \end{subfigure}
    \caption{Let $\O$ (red) be an optimum square containing a blue point $u$. Also, let $R(u)$ (black) be the square centered at the blue point $u$. (a) $H$ is a square such that $\O\cap R(u) \subseteq H$, and $P$ (blue) is the square such that $P=\UR(u, H)$ and it contains $\O_{SW}$. (b) Case 1.1: $R'_{SW}(u)$ contains no red points. Then, $R_4$ (blue) is a Type 2 candidate square of $u$ that contains $\O_{SW}$.}
    \label{fig:Correctness_Lemma_1}
    \end{figure}

Now, using Lemma~\ref{lm:type 2 existance}, we prove the following lemma that will be needed in the subsequent sections.

% In this lemma, we also mention important properties of the candidate squares of $u$, which are essential to obtain the required competitive ratio in the subsequent subsection.

\begin{lemma}\label{lema:candidate_squares}
   Let $u$ be an input blue point in $\O_{SW}$. Then, one of the following properties holds.
    \begin{enumerate}
        \item[P1.] At least one non-empty Type~2 candidate square of $u$ exists and it contains~$\O_{SW}$.

        \noindent or,
        \item[P2.] $\Q(u)_{SW}\cap R(u)_{SW}\neq \emptyset$. 
    % \begin{enumerate} 
    %     \item[(a)] $\Q(u)_{SW}\cap R(u)_{SW}\neq \emptyset$. 
        % and $\mathcal{PS}(u)_{SW}\neq \emptyset$.
        % \item[(b)] {There exists a square $P\in\mathcal{PS}(u)_{SW}$ that contains $\O_{SW}$.}
        % \item[(c)] {There exist three Type~1 candidate squares $R_1, R_2$ and $R_3$ of $u$ (they may not be distinct), constructed by using three squares $H_1, H_2$ and $H_3\in \mathcal{PS}(u)_{SW}$, respectively, such that $R_j\neq \emptyset$ for all $j$.}
        % \item[(d)] {If $H_j$, at starting position, contains $\O_{SW}$, then $R_j$ contains $\O_{SW}$.}
        % \item[(e)] {If $H_j$, at starting position, contains $\O_{NW}$, then $R_j$ contains either $\O_{NW}$ or $\O_{SW}$, and the bottom edge of $R_j$ does not lie above the bottom edge of $H_j$.}
        % \item[(f)] {Similarly, if $H_j$, at starting position, contains $\O_{SE}$, then $R_j$ contains either $\O_{SE}$ or $\O_{SW}$, and the left edge of $R_j$ does not lie on the right side of the left edge of~$H_j$.}
        % \item[(b)] Either at least one Type~1 candidate square of $u$ contains $\O_{SW}$ or the union of the Type~1 candidate squares of $u$  contains $\O_{NW} \cup \O_{SE}$.
        
        % {Either at least one $R_j$ contains $\O_{SW}$ or $\bigcup_{i=j}^3 R_j$ contains $\O_{NW} \cup \O_{SE}$.}

    % \end{enumerate}
    \end{enumerate}
\end{lemma}
\begin{proof}
   First, note that, as $R(u)$ and $\O$ are translated copies of each other, and $u\in \O_{SW}$, the square $R(u)$ contains $\O_{SW}$. Now, if $R(u)\cap\P_r=\emptyset$, then according to Section~\ref{sec:all_candidate_square}, we have $\S(u)=\{R(u)\}$. In this case, property P1 holds. So, w.l.o.g., assume that at least one sub-squares of $R(u)$ contains some red points. But, observe that, as $R(u)_{NE}$ is contained in $\O$, the sub-square $R(u)_{NE}$ contains no red point. Now, we have the following cases depending on which sub-squares of $R(u)$, except $R(u)_{NE}$, contain some red points. See Figure~\ref{fig:OPT_all_Cases} for illustration.

    \begin{figure}[htbp]
        \centering
        \includegraphics[page=1, width=90mm]{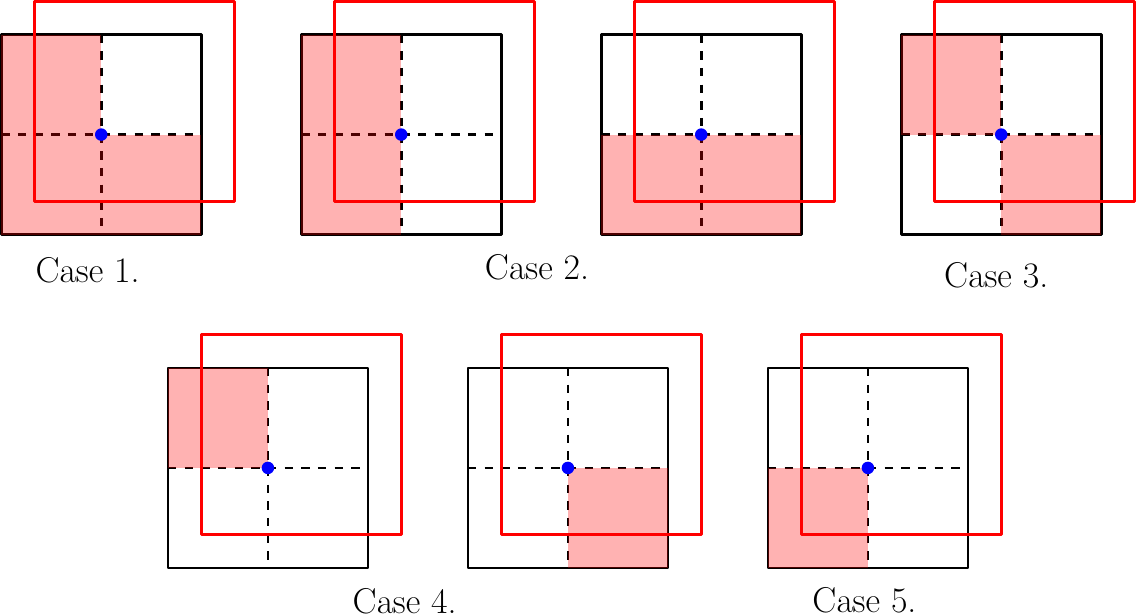}
        \caption{Let $\O$ (red) be an optimum square containing a blue point $u$. Also, let $R(u)$ (black) be the square centered at the blue point $u$. The Red shaded sub-squares of $R(u)$ contain some red points.}
        \label{fig:OPT_all_Cases}
    \end{figure}

\begin{itemize}

       \item \textbf{Case 1: Only $R(u)_{NE}$ contains no red points.} Let $r_1 \in R(u)_{NW}\cap \P_r$ and $r_2 \in R(u)_{SE}\cap \P_r$ be two red points nearest to the line segments $\overline{us}$ and $\overline{ur}$, respectively (see Figure~\ref{fig:opt_square_case_1.1}). As defined in Section~\ref{sec:all_candidate_square} (Case~1), consider the square $R'$ whose left edge and bottom edge pass through the points $r_1$ and $r_2$, respectively.
       \begin{claim}\label{clm:R_RSW}
           $\O\cap R(u)\subseteq R'$, and $R'$ contains $\O_{SW}$.
       \end{claim}
       \begin{proof}
           Observe that the left edge of $R'$ lies between the left edge of $R(u)$ and the left edge of $\O$. Similarly, the bottom edge of $R'$ lies between the bottom edge of $R(u)$ and the bottom edge of $\O$. Due to this, $R'$, a translated copy of both $R(u)$ and $\O$, contains $R(u)\cap\O$. Now, since $R(u)$ contains $\O_{SW}$ and $\O\cap R(u)\subseteq R'$, we have that $R'$ contains $\O_{SW}$.
       \end{proof}
       Due to the above claim, $R'_{NE}$ contains no red point. Now, depending on whether $R'_{SW}$ contains some red points or not, we have the following two sub-cases.

       \hspace{5mm} \underline{\emph{Case~1.1:}} $R'_{SW}$ contains no red point (see Figure~\ref{fig:opt_square_case_1.1}). So, only $R'_{NW}\cup R'_{SE}$ may contain some red points. In that case (see Section~\ref{sec:all_candidate_square} (Case~1.1)), we set $R_4=\UR(u, R')$. Note that because of the definition of $R'$, the set $R(u)\cap R'$ contains no red point. Thus, due to Lemma~\ref{lm:type 2 existance}, we have that the candidate square $R_4\neq\emptyset$, and $R_4$ contains $O_{SW}$.
       % {Since the left edge of $R'$ lies to the left side of the left edge of $\O$ and the square $\O$ is red point free, we have that $R'$ will be moved rightwards until $R'_{NW}$ becomes red point free or the left edge of $R'$ touches the left edge of $\O$. Similarly, since the bottom edge of $R'$ lies below the bottom edge of $\O$ and the square $\O$ is red point free, we have that $R'$ will be moved upwards until $R'_{SE}$ becomes red point free or the bottom edge of $R'$ touches the bottom edge of $\O$. Also, observe that during the movement of $R'$, the sub-squares $R'_{NE}$ and $R'_{SW}$ never contain a red point. Therefore, during the movement of $R'$, we are guaranteed to find a candidate square that contains $u$ but no red points, and we set such a square as a candidate square $R_4$ of $u$. Hence, $\S(u)\neq \emptyset$. Note that during the movement of $R'$, the left edge of $R'$ lies always to the left side of the left edge of $\O$, and the bottom edge of $R'$ lies always below the bottom edge of $\O$. Therefore, $R'$ always contains $\O_{SW}$, and thus, $R_4$ also contains $\O_{SW}$.} 
       So, in this case, property P1 holds.

    % \begin{figure}[htbp]
    % \centering
    % \begin{subfigure}[b]{.49\textwidth}
    %     \centering
    %     \includegraphics[page=2, width=55mm]{OPT_Squares_Existance.pdf}
    %     \caption{}
    %     \label{fig:opt_square_case_1.1}
    % \end{subfigure}
    % \hfill
    % \begin{subfigure}[b]{.49\textwidth}
    %     \centering
    %     \includegraphics[page=3, width=55mm]{OPT_Squares_Existance.pdf}
    %     \caption{}
    %     \label{fig:opt_square_case_1.2.p}
    % \end{subfigure}
    % \caption{$\O$ (red) be an optimum square and $R(u)$ (black) be the square centered at the blue point $u$. (a) Case 1.1: $R'_{SW}(u)$ contains no red points. Then, $R_4$ (blue) is a candidate square of $u$ that contains $\O_{SW}$. (b) Case 1.2: The square $P$ (blue) contains $\O_{SW}$.}
    % \label{fig:opt_square_cases1}
    % \end{figure}

    \begin{figure}[htbp]
    \centering
    \begin{subfigure}[b]{.49\textwidth}
        \centering
        \includegraphics[page=5, width=75mm]{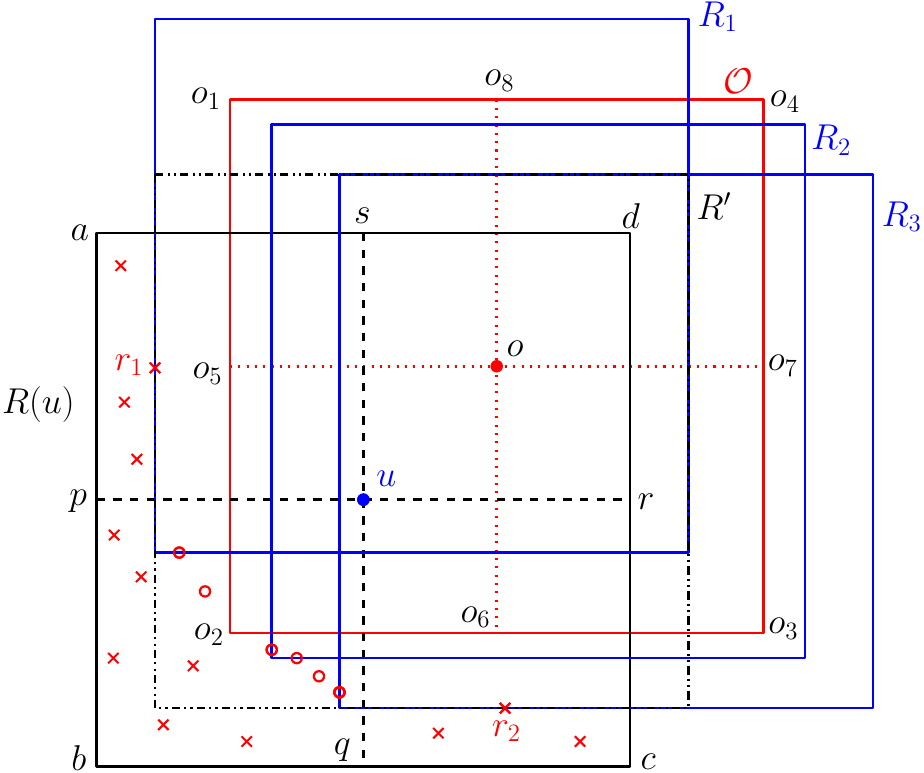}
        \caption{}
        \label{fig:opt_square_case_1.2}
    \end{subfigure}
    \hfill
    \begin{subfigure}[b]{.49\textwidth}
        \centering
        \includegraphics[page=6, width=75mm]{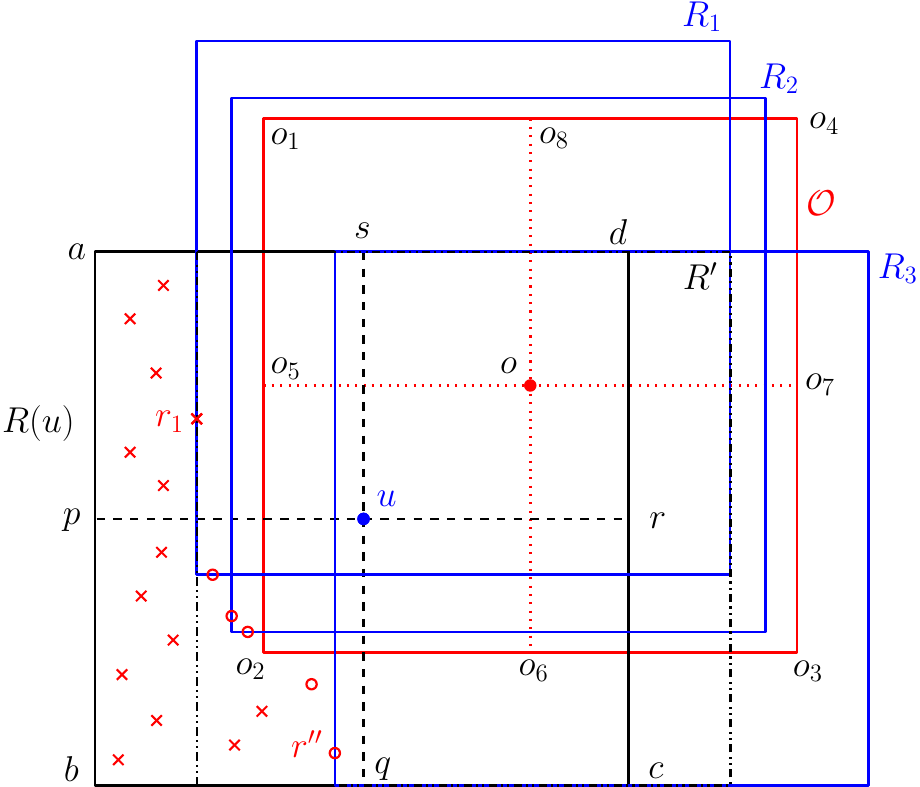}
        \caption{}
        \label{fig:opt_square_case_2.1}
    \end{subfigure}
    \caption{Let $\O$ (red) be an optimum square containing a blue point $u$. Also, let $R(u)$ (black) be the square centered at the blue point $u$. (a) Case 1.2: $R'_{SW}$ contains some red points. Here, $\Q(u)_{SW}\cap R(u)_{SW}$ (red/circle) is non-empty. Then, $R_1, R_2$ and $R_3$ (blue) are the candidate squares of $u$ that contains $\O_{NW}\cup \O_{SE}$, constructed by using $R'$ (black/dashed). (b) Case~2.1: $R(u)_{NW}$ and $R(u)_{SW}$ contain some red points such that $r''$ does not lie on the left side of the left edge of $\O$. Here, $\Q(u)_{SW}\cap R(u)_{SW}$ (red/circle) is non-empty. Then, $R_1$, $R_2$ and $R_3$ (blue) are the candidate squares of $u$ that contains $\O_{NW}\cup\O_{SE}$, constructed by using $R'$ (black/dashed).}
    \label{fig:opt_square_cases1}
    \end{figure}

       \hspace{5mm} \underline{\emph{Case~1.2:}} $R'_{SW}$ contains some red points (see Figure~\ref{fig:opt_square_case_1.2}). In this case, we define the Type~1 candidate squares of $u$ (see Section~\ref{sec:all_candidate_square} (Case~1.2)). Since $R(u)_{SW}\cap \P_r \neq \emptyset$ and  $R(u)_{NE}\cap \P_r = \emptyset$, as defined in  Section~\ref{sec:possible_candidate_square}, consider the sequence of $SW$ staircase points $\Q(u)_{SW}$ of $u$ constructed by using the points in $R(u)_{SW}\cap R'_{SW}\cap\P_r$. Because of the definitions of the two points $r_1, r_2$ and the square $R'$, observe that the red points in $R'_{SW}$ must belong to $R(u)_{SW}$. Therefore, $R(u)_{SW}\cap R'_{SW}\cap\P_r\neq \emptyset$. As a result, $\Q(u)_{SW}\cap R(u)_{SW}\neq \emptyset$. Thus, property P2 holds. 

    \begin{figure}[htbp]
    \centering
    \begin{subfigure}[b]{.49\textwidth}
        \centering
        \includegraphics[page=7, width=75mm]{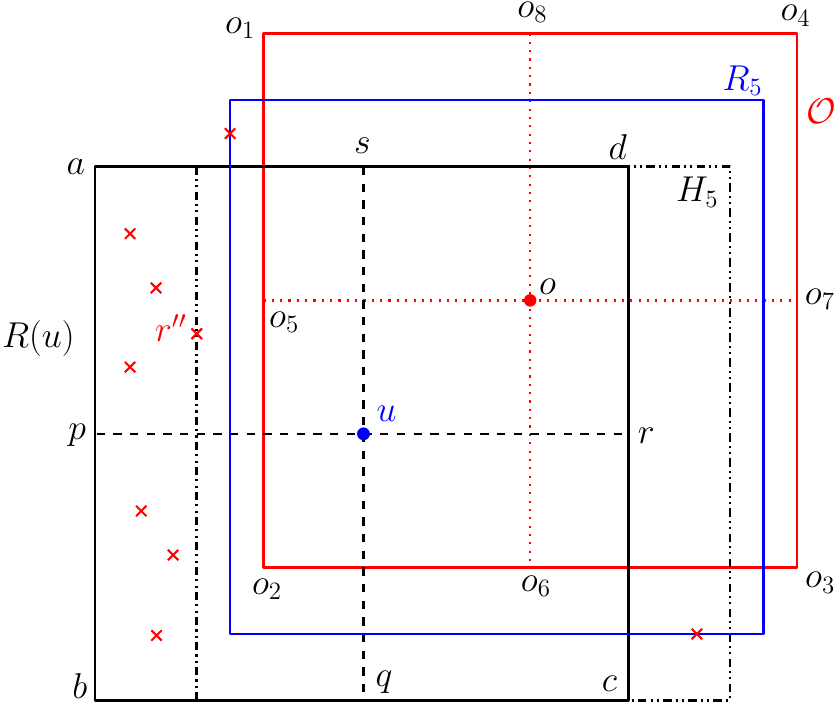}
        \caption{}
        \label{fig:opt_square_case_2.2}
    \end{subfigure}
    \hfill
    \begin{subfigure}[b]{.49\textwidth}
        \centering
        \includegraphics[page=8, width=75mm]{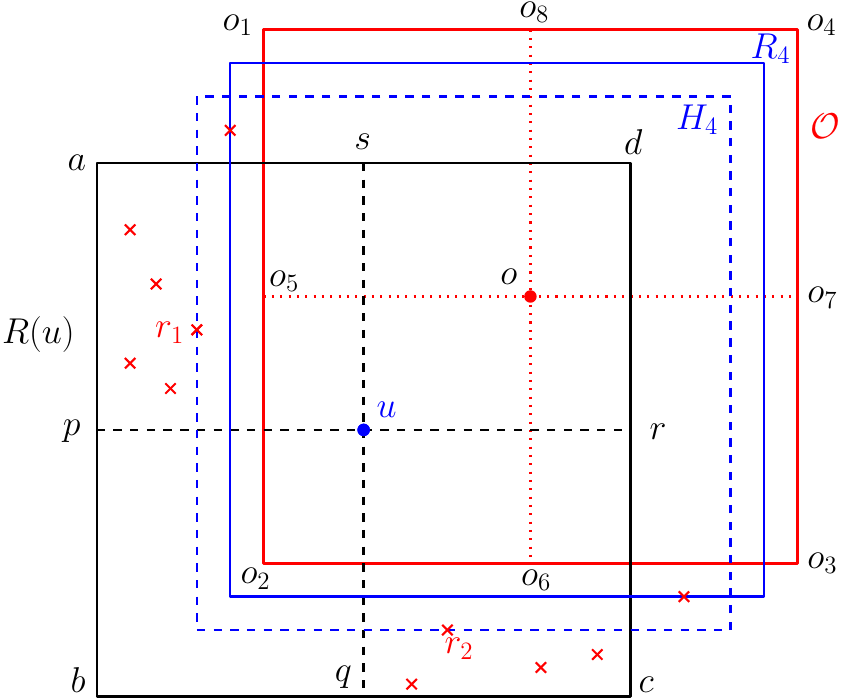}
        \caption{}
        \label{fig:opt_square_case_3}
    \end{subfigure}
    \caption{Let $\O$ (red) be an optimum square containing a blue point $u$. Also, let $R(u)$ (black) be the square centered at the blue point $u$. (a) Case 2.2: $R(u)_{NW}$ and $R(u)_{SW}$ contain some red points such that $r''$ lies on the left side of the left edge of $\O$. Then, $R_5$ (blue) is a candidate square of $u$ that contains $\O_{SW}$, obtained by moving the square $H_5$ (black/dashed). (b) Case~3: $R_{NW}(u)$ and $R_{SE}(u)$ both contain some red points. Then, $R_4$ (blue) is a candidate square of $u$ that contains $\O_{SW}$, obtained by moving the square $H_4$ (blue/dashed).}
    \label{fig:opt_square_cases3}
    \end{figure}

       % \hspace{5mm} {Assuming $H_j$, at starting position, contains $\O_{NW}$, here we argue that the bottom edge of $R_j$ does not lie above the bottom edge of $H_j$ (see Figure~\ref{fig:opt_square_case_1.2.h1}). In this case, the square $H_j$, first, is moved in the rightward direction along the staircase, and when it is moved in the upward direction, the bottom edge of $H_j$ already lies below the bottom edge of $\O$. So, it is moved upwards until the bottom edge of $H_j$ touches the bottom edge of $\O$. Since, at the starting position, the bottom edge of $H_j$ lies above the bottom edge of $\O$, the bottom edge of $R_j$ does not lie above the bottom edge of $H_j$. Therefore, P2(e) holds. Similarly, we can prove that P2(d) and P2(f) also hold.} 

       \item \textbf{Case 2: $R(u)_{SW}$ contains some red points, and one of $R(u)_{NW}$ and $R(u)_{SE}$ contains no red point.} W.l.o.g., assume that $R(u)_{SE}$ contains no red points; the other case is similar. Let $r'' \in (R(u)_{NW}\cup R(u)_{SW})\cap \P_r$ be a red point nearest to the line segments $\overline{qs}$ (see Figure~\ref{fig:opt_square_case_2.1}). Now, depending on the position of $r''$, we have the following two sub-cases.

       \hspace{5mm} \underline{\emph{Case~2.1:}} The red point $r''$ lies on the right side of the left edge of $\O$ (see Figure~\ref{fig:opt_square_case_2.1}). Hence, $r''\in R(u)_{SW}\cap \P_r$. In this case, we define the Type~1 candidate squares of $u$ (see Section~\ref{sec:all_candidate_square} (Case~2)). Since $R(u)_{SW}\cap \P_r \neq \emptyset$ and  $R(u)_{NE}\cap \P_r = \emptyset$, as defined in  Section~\ref{sec:possible_candidate_square}, consider the sequence of $SW$ staircase points $\Q(u)_{SW}$ of $u$ constructed by using the points in $R(u)_{SW}\cap R'_{SW}\cap\P_r$. Recall that $R'$ is a square whose left edge and bottom edge pass through the points $r_1$ and $r_2$, respectively, where $r_1 \in R(u)_{NW}\cap \P_r$ be a red point nearest to the line segments $\overline{us}$ and $r_2=c$, the $SE$ vertex of $R(u)$ (as $R(u)_{SE}$ contains no red point). Observe that $R'$ is a horizontal translated copy of $R(u)$ whose left edge passes through the point~$r_1$. 
       \begin{claim}\label{clm:H5_RSW}
           $\O\cap R(u)\subseteq R'$, and $R'$ contains $\O_{SW}$.
       \end{claim}
       \begin{proof}
           Observe that the left edge of $R'$ lies between the left edge of $R(u)$ and the left edge of $\O$. As a result, we have that $R'$, a horizontal translated copy of both $R(u)$ and $\O$, contains $R(u)\cap\O$. Since $R(u)$ contains $\O_{SW}$ and $\O\cap R(u)\subseteq R'$, we have that $R'$ contains $\O_{SW}$.
       \end{proof}
       Due to Claim~\ref{clm:H5_RSW}, $R'_{NE}$ does not contain any red point. But $r''\in R'_{SW}$, i.e., $R'_{SW}$ contains some red points. Therefore, $R(u)_{SW}\cap R'_{SW}\cap\P_r\neq \emptyset$. As a result, $\Q(u)_{SW}\cap R(u)_{SW}\neq \emptyset$. So, property P2 holds.
       % Let $\Q(u)_{SW}=\{r_1, q_0, q_1, \dots, q_k, r''\}$ be the staircase points set of $u$ such that for each $i\in\{0\} \cup [k]$, the point $q_i\in R(u)_{SW}\cap R'_{SW}\cap \P_r$. 
       % Now, in a  similar way as to above Case~1.2, we can show that the candidate squares $R_1$, $R_2$ and $R_3$ exist, and property P2 also holds.

       % In this case (see Section~\ref{sec:candidate_square} (Case~3a)), we consider the set of staircase points $\Q(u)$ and the set of staircase squares $\mathcal{PS}(u)$ of $u$.
       % In this case (see Section~\ref{sec:candidate_square} (Case~3a)), we prove $\S(u)\neq \emptyset$ and P2 holds in a similar way as Case~3.1.2. Let $r_2 \in R(u)_{NW}\cap \P_r$ be a red point nearest to the line segments $\overline{us}$ (see Figure~\ref{fig:opt_square_case_3.2.2}).

       % $H_1$ rightwards and upwards in a similar way as $H_1$ in Case~3.1. Now, the existence of the left side candidate square $R_1$ of $u$ follows from Case~C1. Hence, $\S(u) \neq \emptyset$, and also $R_1$ contains~$R_{SW}$.

       % Note that in this case (See Section~\ref{sec:candidate_square} (Case~2.2)), the red points in $R'_{SW}$ must belong to $R(u)_{SW}$ because of the definitions of of $r_1$ and $R'$. Therefore, $R(u)_{SW}\cap R'_{SW}\cap\P_r\neq \emptyset$. As a result, $\Q(u)\cap R(u)_{SW}\neq \emptyset$. Let $\Q(u)=\{r_2, q_0, q_1, \dots, q_k, r_1\}$ be the staircase points set of $u$ such that for each $i\in\{0\} \cup [k]$, the point $q_i\in R(u)_{SW}\cap R'_{SW}\cap \P_r$.

       \hspace{5mm} \underline{\emph{Case~2.2:}} The red point $r''$ lies on the left side of the left edge of $\O$ (see Figure~\ref{fig:opt_square_case_2.2}), then as defined in Section~\ref{sec:all_candidate_square} (Case~2), we consider a square $H_5$, a horizontal translated copy of $R(u)$, whose left edge passes through the point $r''$. In a similar way as in Claim~\ref{clm:H5_RSW}, we can show that $\O\cap R(u)\subseteq H_5$, and $H_5$ contains $\O_{SW}$.
       % \begin{claim}\label{clm:H5_RSW}
       %     $\O\cap R(u)\subseteq H_5$, and $H_5$ contains $\O_{SW}$.
       % \end{claim}
       % \begin{proof}
       %     Observe that the left edge of $H_5$ lies between the left edge of $R(u)$ and the left edge of $\O$. As a result, we have that $H_5$, a horizontal translated copy of both $R(u)$ and $\O$, contains $R(u)\cap\O$. Since $R(u)$ contains $\O_{SW}$ and $\O\cap R(u)\subseteq H_5$, we have that $H_5$ contains $\O_{SW}$.
       % \end{proof}
       Due to this, $H_{5(NE)}$ contains no red point. Observe that only $H_{5(SE)}$ may contain some red points. In that case, we evoke \UR$(u, H_5)$. Note that because of the definition of $H_5$, the set $R(u)\cap H_5$ contains no red point.
       % Here at first, only $H_{5(SE)}$ contains some red points, and during the movement of $H_5$, since $H_{5(NE)}$ always lies inside $\O$, the sub-square $H_{5(NE)}$ never contains a red point. As a result, $H_5$ will be moved similarly as Subroutine-\ref{sub:1} on input $(\langle u, H_5\rangle)$. 
       Thus, due to Lemma~\ref{lm:type 2 existance}, we have that the candidate square $R_5\neq\emptyset$, and $R_5$ contains $O_{SW}$. 
       Hence, property P1 holds.

    \begin{figure}[htbp]
    \centering
    \begin{subfigure}[b]{.49\textwidth}
        \centering
        \includegraphics[page=9, width=75mm]{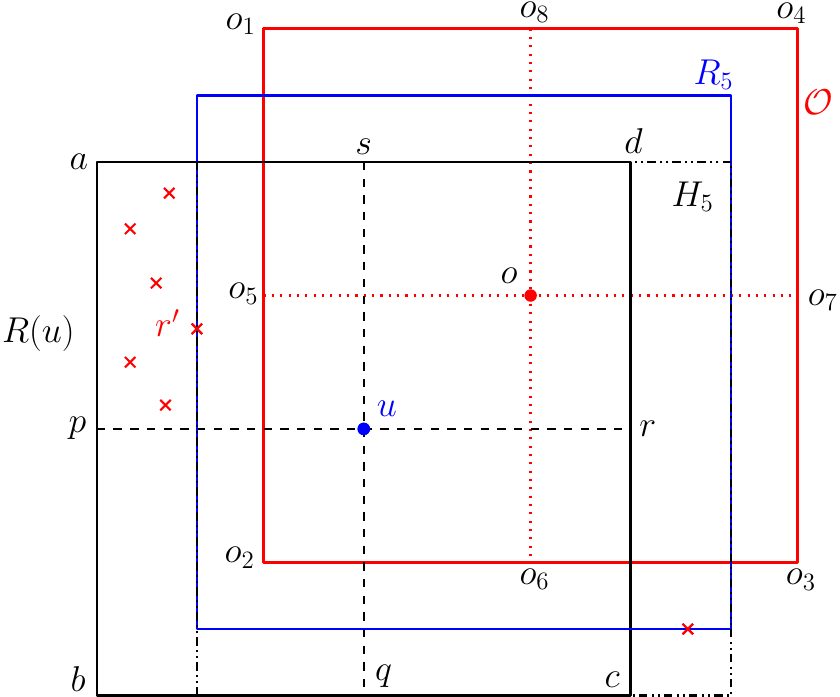}
        \caption{}
        \label{fig:opt_square_case_4}
    \end{subfigure}
    \hfill
    \centering
    \begin{subfigure}[b]{.49\textwidth}
        \centering
        \includegraphics[page=10, width=75mm]{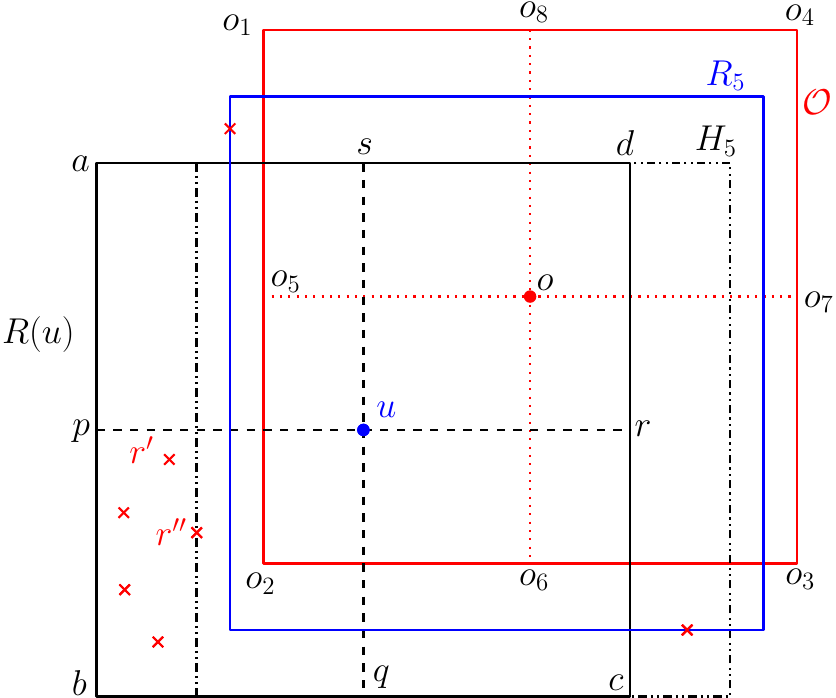}
        \caption{}
        \label{fig:opt_square_case_5.1}
    \end{subfigure}
    \caption{Let $\O$ (red) be an optimum square containing a blue point $u$. Also, let $R(u)$ (black) be the square centered at the blue point $u$. (a) Case 4: Only $R_{NW}(u)$ contains some red points. Then, $R_5$ (blue) is a candidate square of $u$ that contains $\O_{SW}$, obtained by moving the square $H_5$ (black/dashed). (b) Case 5.1: Only $R_{SW}(u)$ contains some red points, and both $r', r''$ lie on the left side of $a_1a_2$. Then, $R_5$ (blue) is a candidate square of $u$ that contains $\O_{SW}$, obtained by moving the square $H_5$ (black/dashed).}
    \label{fig:opt_square_cases4}
    \end{figure}

       \item \textbf{Case~3: $R(u)_{SW}$ contains no red points, but  both $R(u)_{NW}$ and $R(u)_{SE}$ contain some red points.} Let $r_1 \in R(u)_{NW}\cap \P_r$ and $r_2 \in R(u)_{SE}\cap \P_r$ be two red points nearest to the line segments $\overline{us}$ and $\overline{ur}$, respectively (see Figure~\ref{fig:opt_square_case_3}). Then, as defined in Section~\ref{sec:all_candidate_square} (Case~3), we consider the square $H_4$ whose left and bottom edges pass through the points $r_1$ and $r_2$, respectively. In a similar way as in Claim~\ref{clm:R_RSW}, we can show that $\O\cap R(u)\subseteq H_4$, and $H_4$ contains $\O_{SW}$.
       % \begin{claim}
       %     $\R\cap R(u)\subseteq H_4$, and $H_4$ contains $\R_{SW}$.
       % \end{claim}
       % \begin{proof}
       %      Observe that the left edge of $H_4$ lies between the left edge of $R(u)$ and the left edge of $\R$. Similarly, the bottom edge of $H_4$ lies between the bottom edge of $R(u)$ and the bottom edge of $\R$. Due to this, $H_4$, a translated copy of both $R(u)$ and $\R$, contains $(R(u)\cap\R)$. Now since $R(u)$ contains $\R_{SW}$ and $\R\cap R(u)\subseteq H_4$, we have that $H_4$ contains $\R_{SW}$.
       % \end{proof}
       Due to this, $H_{4(NE)}$ contains no red point. Note that only $H_{4(NW)}\cup H_{4(SE)}$ may contains some red points. If $H_{4(NW)}\cup H_{4(SE)}$ contains some red points, we evoke \UR$(u, H_4)$. Notice that because of the definition of $H_4$,  the set $R(u)\cap H_4$ contains no red point. 
       Thus, due to Lemma~\ref{lm:type 2 existance}, we have that the candidate square $R_4\neq\emptyset$, and $R_4$ contains $O_{SW}$. Hence, property P1 holds.

       \item \textbf{Case~4: $R(u)_{SW}$ contains no red points, and one of $R(u)_{NW}$ and $R(u)_{SE}$ contains some red points.} W.l.o.g., assume that $R(u)_{NW}$ contains some red points; the other case is similar. Let $r' \in R(u)_{NW}\cap \P_r$  be a red point nearest to the line segments $\overline{us}$ (see Figure~\ref{fig:opt_square_case_4}). Then, as defined in Section~\ref{sec:all_candidate_square} (Case~4), we consider the square $H_5$, a horizontal translated copy of $R(u)$, whose left edge passes through $r'$. In a similar way as in Claim~\ref{clm:H5_RSW}, we can show that $\O\cap R(u)\subseteq H_5$, and $H_5$ contains $\O_{SW}$.
       % \begin{claim}
       %     $\R\cap R(u)\subseteq H_5$, and $H_5$ contains $\R_{SW}$.
       % \end{claim}
       % \begin{proof}
       %     Observe that the left edge of $H_5$ lies between the left edge of $R(u)$ and the left edge of $\R$. As a result, we have that $H_5$, a horizontal translated copy of both $R(u)$ and $\R$, contains $(R(u)\cap\R)$. Since $R(u)$ contains $\R_{SW}$ and $\R\cap R(u)\subseteq H_5$, we have that $H_5$ contains $\R_{SW}$.
       % \end{proof}
       Due to this, $H_{5(NE)}$ contains no red point. Note that only $H_{5(SE)}$ may contain some red points. If $H_{5(SE)}$ contains some red points, we evoke \UR$(u, H_5)$. 
       % Here at first, only $H_{5(SE)}$ contains some red points, and during the movement of $H_5$, since $H_{5(NE)}$ always lies inside $\O$, the sub-square $H_{5(NE)}$ never contains a red point. As a result, $H_5$ will be moved similarly to U(SE)-R(NW)-MOVE$(\langle u, H_5\rangle)$. 
       Notice that because of the definition of $H_5$, the set $R(u)\cap H_5$ contains no red point. 
       Thus, due to Lemma~\ref{lm:type 2 existance}, we have that the candidate square $R_5\neq\emptyset$, and $R_5$ contains $O_{SW}$. Hence, property P1 holds.

       \item \textbf{Case 5: Only $R(u)_{SW}$ contains some red points.} Let $r', r'' \in R(u)_{SW}\cap \P_r$ be two red points nearest to the line segments $\overline{up}$ and $\overline{uq}$, respectively (see Figure~\ref{fig:opt_square_case_5.1}). Now, depending on the position of $r'$ and $r''$, we have the following three sub-cases.

    % \begin{figure}[htbp]
    % \centering
    % \includegraphics[width=60 mm]{Figures/OPT_Square_Case_3.5.3.pdf}
    % \caption{Case 3.5.3: $\O$ (red) be an optimum square and $R'=R(u)$ (black) be the square centered at the blue point $u$. Only $R_{SW}(u)$ contains some red points, and the points $r_1$ and $r_2$ are distinct such that exactly one of $r_1$ and $r_2$ lies to the left side of the edge $a_1a_2$ and other one lies below the edge $a_2a_3$. Then, $R_1, R_2$ and $R_3$ (blue) are the candidate squares of $u$ that contains $\O_{NW}\cup \O_{SE}$, constructed by using $R'$.}
    % \label{fig:opt_square_cases6}
    % \end{figure}

       \hspace{5mm} \underline{\emph{Case~5.1:}} The red points $r'$ and $r''$ both lie on the left side of the left edge of $\O$ (see Figure~\ref{fig:opt_square_case_5.1}). We consider the square $H_5$, a horizontal translated copy of $R(u)$, whose left edge passes through the point $r''$ (see Section~\ref{sec:all_candidate_square} (Case~4)). In a similar way as in Claim~\ref{clm:H5_RSW}, we can show that $\O\cap R(u)\subseteq H_5$, and $H_5$ contains $\O_{SW}$.
       % \begin{claim}
       %     $\R\cap R(u)\subseteq H_5$, and $H_5$ contains $\R_{SW}$.
       % \end{claim}
       % \begin{proof}
       %     Observe that the left edge of $H_5$ lies between the left edge of $R(u)$ and the left edge of $\R$. As a result, we have that $H_5$, a horizontal translated copy of both $R(u)$ and $\R$, contains $(R(u)\cap\R)$. Since $R(u)$ contains $\R_{SW}$ and $\R\cap R(u)\subseteq H_5$, we have that $H_5$ contains $\R_{SW}$.
       % \end{proof}
       Due to this, $H_{5(NE)}$ contains no red point. Note that only $H_{5(SE)}$ may contain some red points. In that case, we evoke \UR$(u, H_5)$. Notice that because of the definition of $H_5$, the set $R(u)\cap H_5$ contains no red point. Thus, due to Lemma~\ref{lm:type 2 existance}, we have that the candidate square $R_5\neq\emptyset$, and $R_5$ contains $O_{SW}$. Hence, property P1 holds.
       
       \hspace{5mm} \underline{\emph{Case~5.2:}} The red points $r'$ and $r''$ both lie below the bottom edge $\O$. In this case, in a similar way as in the above Case~5.1, we can prove the existence of the candidate square $R_4$, and the fact that $R_4$ contains $\O_{SW}$. Hence, property P1 also holds.

       \hspace{5mm} \underline{\emph{Case~5.3:}} The red points $r'$ and $r''$ are distinct such that exactly one of $r'$ and $r''$ lies to the left side of the left edge of $\O$ and other one lies below the bottom edge of $\O$ (see Figure~\ref{fig:opt_square_case_5.3}).  In this case, we define the Type~1 candidate squares of $u$ (see Section~\ref{sec:all_candidate_square} (Case~4)). Since $R(u)_{SW}\cap \P_r \neq \emptyset$ and  $R(u)_{NE}\cap \P_r = \emptyset$, as defined in  Section~\ref{sec:possible_candidate_square}, consider the sequence of $SW$ staircase points $\Q(u)_{SW}$ of $u$ constructed by using the points in $R(u)_{SW}\cap R'_{SW}\cap\P_r$. Recall that $R'$ is a square whose left edge and bottom edge pass through the points $r_1$ and $r_2$, respectively, where $r_1=a$, the $NW$ vertex of $R(u)$ and $r_2=c$, the $SE$ vertex of $R(u)$ (as both sub-squares $R(u)_{NW}$ and $R(u)_{SE}$ contain no red point). Observe that $R'=R(u)$. Therefore, $R'_{SW}$ contains some red points. As a result, $\Q(u)_{SW}\cap R(u)_{SW}\neq \emptyset$. So, property P2 holds.
       % Let $\Q(u)_{SW}=\{r_1, q_0, q_1, \dots, q_k, r_2\}$ be the staircase points set of $u$ such that for each $i\in\{0\} \cup [k]$, the point $q_i\in R(u)_{SW}\cap R'_{SW}\cap \P_r$.
       % Now, in a  similar way as to above Case~1.2, we can show that the candidate squares $R_1$, $R_2$ and $R_3$ exist, and property P2 also holds.

       % In this case (see Section~\ref{sec:candidate_square} (Case~4)), we consider the set of staircase points $\Q(u)$ and the set of staircase squares $\mathcal{PS}(u)$ of $u$. Now, as defined in Section~\ref{sec:possible_candidate_square}, consider the square $R'$ whose bottom edge and left edge pass through the points $c$ and $a$, respectively (recall that $R(u)_{SE}$ and $R(u)_{NW}$ both contain no red points). Observe that $R'=R(u)$. Therefore, $R'_{SW}$ contains some red points. As a result, $R(u)_{SW}\cap R'_{SW}\cap\P_r\neq \emptyset$, and also, $\Q(u)\cap R(u)_{SW}\neq \emptyset$. Now, in a  similar way as to above Case~1.2, we can show that the candidate squares $R_1$, $R_2$ and $R_3$ exist, and property P2 also holds.

       % As described in Section~\ref{sec:candidate_square} (Case~4), consider the square $R(u)$ to be the square $R'$. Now, $R(u)_{SW}\cap\P_r\neq \emptyset$. So, similar to Case~1.2, we can show that the candidate squares $R_1$, $R_2$ and $R_3$ exist, and properties P2 also hold.

   \end{itemize}

% So for all the cases, we have proved that if a blue point $u \in \R_{SW}$, then there exists at least one candidate square of $u$, i.e. $\S(u)\neq \emptyset$ and $1\leq |\S(x)|\leq 5$. Also, either of the properties of the candidate squares of $u$ holds.

So the lemma follows.
\end{proof}

    \begin{figure}[htbp]
    \begin{subfigure}[b]{.49\textwidth}
        \centering
        \includegraphics[page=11, width=80mm]{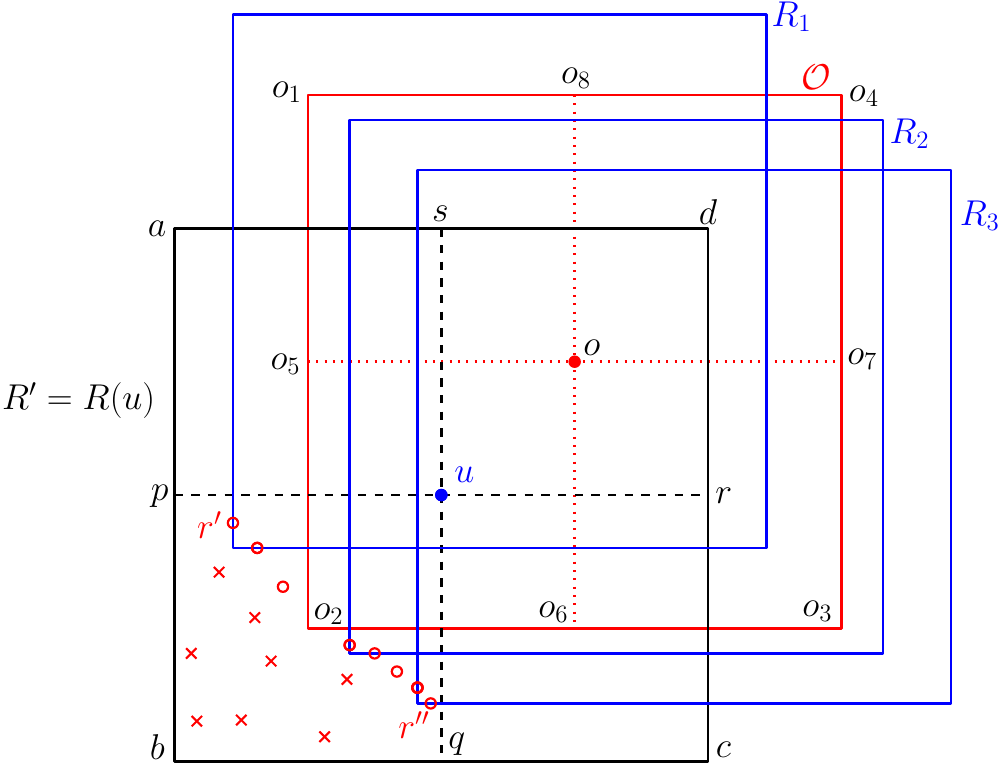}
        \caption{}
        \label{fig:opt_square_case_5.3}
    \end{subfigure}
    \hfill
    \begin{subfigure}[b]{.49\textwidth}
        \centering
        \includegraphics[page=2, width=75mm]{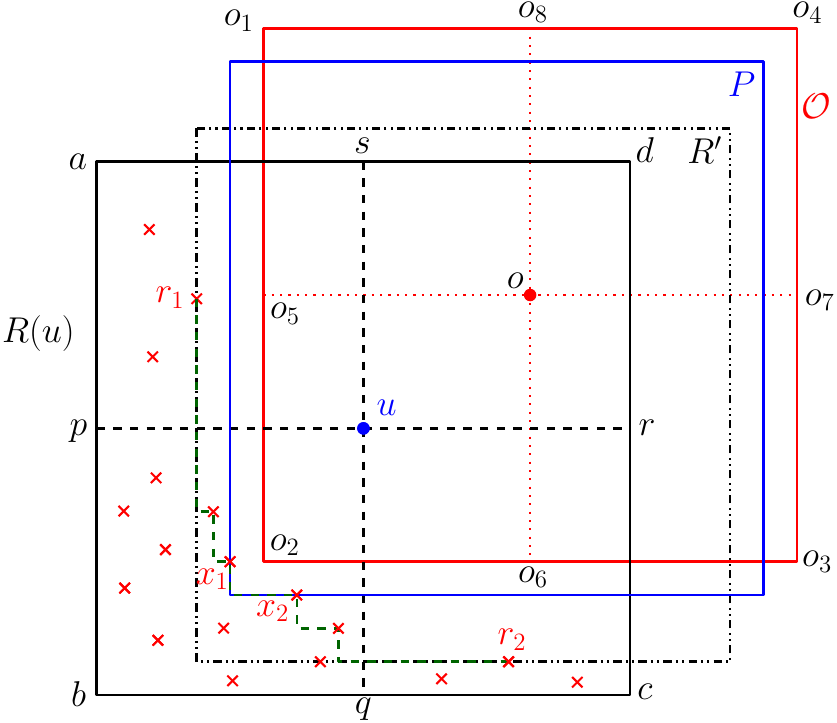}
        \caption{}
        \label{fig:P in PS(u)_SW}
    \end{subfigure}
    \caption{Let $\O$ (red) be an optimum square containing a blue point $u$. Also, let $R(u)$ (black) be the square centered at the blue point $u$. (a) Case 5.3: Only $R_{SW}(u)$ contains some red points, and the points $r'$ and $r''$ are distinct such that exactly one of $r'$ and $r''$ lies to the left side of the left edge of $\O$ and other one lies below the bottom edge of $\O$. Here, $\Q(u)_{SW}\cap R(u)_{SW}$ (red/circle) is non-empty. Then, $R_1, R_2$ and $R_3$ (blue) are the candidate squares of $u$ that contains $\O_{NW}\cup \O_{SE}$, constructed by using $R'$. (b)  The staircase square $P$ (blue) in $\mathcal{PS}(u)_{SW}$ contains $\O_{SW}$.}
    \label{fig:opt_square_cases5}
    \end{figure}

Now, we prove some important properties of the Type 1 candidate square that will be needed to confirm the correctness of our algorithm as well as to obtain the desired competitive ratio of the algorithm.

    % \begin{lemma}\label{lm:P in PS(u)_SW}
    %     Let $u$ be an input blue point in $\O_{SW}$ such that $\Q(u)_{SW}\cap R(u)_{SW}\neq \emptyset$. Then, the set
    %      of $SW$ staircase squares of $u$ is non-empty, i.e.,
    %      $\mathcal{PS}(u)_{SW}\neq \emptyset$, and there exists a square $P\in\mathcal{PS}(u)_{SW}$ that contains $\O_{SW}$.
    % \end{lemma}

    \begin{lemma}\label{lm:P in PS(u)_SW}
        Let $u$ be an input blue point in $\O_{SW}$ such that $\Q(u)_{SW}\cap R(u)_{SW}\neq \emptyset$. Then, the set
         $\mathcal{PS}(u)_{SW}$ is non-empty, and there exists a square $P\in\mathcal{PS}(u)_{SW}$ that contains $\O_{SW}$.
    \end{lemma}
    \begin{proof}
       Note that, as $u$ lies in $\O_{SW}$, the sub-square $R(u)_{NE}$ does not contain any red point. Also, since $\Q(u)_{SW}\cap R(u)_{SW}\neq \emptyset$, we have that the square $R(u)_{SW}$ contains some red points. So, as defined in Section~\ref{sec:possible_candidate_square}, consider the sequence of $SW$ staircase points $\Q(u)_{SW}=(r_1,q_0, q_1, \dots, q_k,r_2)$, where $r_1$ and $r_2$ are the initial and the terminal point of the staircase $ST_{SW}(\Q(u)_{SW})$, respectively, and for each $i\in\{0\} \cup [k]$, the point $q_i\in R(u)_{SW}\cap \P_r$. Recall that $r_1\in R(u)_{NW}$ and $r_2\in R(u)_{SE}$.

        See Figure~\ref{fig:P in PS(u)_SW}. {Since $r_1\in R(u)_{NW}$ lies on the left side of the left edge of $\O$, and the point $r_2\in R(u)_{SE}$ lies below the bottom edge of $\O$, and all the points $q_0,q_1,\dots, q_k$ lie in $R(u)_{SW}$, we have two consecutive points, say $x_1$ and $x_2$, among $r_1, q_0, \dots, q_k, r_2$ such that $x_1$ lies on the left side of the left edge of $\O$, and $x_2$ lies below the bottom edge of $\O$.} So, consider the square $P$ whose left edge contains $x_1$ and the bottom edge contains $x_2$. 
        % In a similar way as in Claim~\ref{clm:R_RSW}, we can prove that $R(u)\cap \O \subseteq P$, and $P$ contains $\O_{SW}$. 
       % Observe that the left edge of $P$ lies between the left edge of $R(u)$ and the left edge of $\R$. Similarly, the bottom edge of $P$ lies between the bottom edge of $R(u)$ and the bottom edge of $\R$. Due to this, $P$, a translated copy of both $R(u)$ and $\R$, contains $(R(u)\cap\R)$. Now, since $R(u)$ contains $\R_{SW}$, we have that $P$ contains $\R_{SW}$.
       Since $P_{NE}$ lies inside $\O$, it contains no red point. As a result, $P\in \mathcal{PS}(u)_{SW}$ and $\mathcal{PS}(u)_{SW} \neq \emptyset$. Also, note that $P$ contains $\O_{SW}$. Hence, the lemma follows.
    \end{proof}

    \begin{lemma}\label{lm:P_1 contains O_NW}
        Let $u$ be an input blue point in $\O_{SW}$ such that $\Q(u)_{SW}\cap R(u)_{SW}\neq \emptyset$. Also, let $\mathcal{PS}(u)_{SW}=(P_{1}, P_{2}, \dots, P_l)$ be the sequence of $SW$ staircase squares of $u$ (as defined in Section~\ref{sec:possible_candidate_square}). Then, there exists some  $i\in[l]$ such that we have the following.
        % and for some $i\in[l]$, a staircase square $P_i\in\mathcal{PS}(u)_{SW}$ contains $\O_{SW}$, then we have the following.
        \begin{enumerate}
        \item[(i)] Each of the staircase squares $P_1,\dots,P_{i-1}$ contains either $\O_{NW}$ or $\O_{SW}$.
        \item[(ii)]Each of the staircase squares $P_{i+1},\dots,P_l$ contains either $\O_{SW}$ or $\O_{SE}$.
    \end{enumerate}
    \end{lemma}
    \begin{proof}
        Due to Lemma~\ref{lm:P in PS(u)_SW}, we have that $\mathcal{PS}(u)_{SW}\neq \emptyset$, and there exists a square $P\in\mathcal{PS}(u)_{SW}$ that contains $\O_{SW}$. Assume that $P_i=P$ for some $i\in[l]$.

        (i) Let $j$ be an index in $\{1,\dots,i-1\}$. Note that the $SW$ vertex of $P_j$ lies in $R(u)_{SW}$ (see Figure~\ref{fig:P_1 contains O_NW}). As a result, the left edge of $P_j$ lies to the right side of the left edge of $R(u)$. Again, since $P_i$ contains $\O_{SW}$, the left edge of $P_j$ lies to the left side of the left edge of $\O$. Therefore, the right edge of $P_j$ lies between the right edge of $R(u)$ and the right edge of $\O$. Similarly, the bottom edge of $P_j$ lies between the line segment $\overline{pr}$ and the bottom edge of $P_i$. Since $u\in \O_{SW}$ and the square $P_i$ contains $\O_{SW}$, the bottom edge of $\O$ also lies between the line segment $\overline{pr}$ and the bottom edge of $P_i$. Now, if the bottom edge of $P_j$ lies above the bottom edge of $\O$, then $P_j$ contains $\O_{NW}$; otherwise, $P_j$ contains $\O_{SW}$.
        
        % Similarly, for each $j\in\{2,\dots,i-1\}$, we can show that $P_j$ contains either $\O_{NW}$ or $\O_{SW}$.
        % Using a similar argument, we can show that at the starting position, if $i\geq \lceil\frac{k'}{2}\rceil$, the square $H_2$ contains either $\O_{NW}$ or $\O_{SW}$, and if $i\leq \lceil\frac{k'}{2}\rceil$, the square $H_2$ contains either $\O_{SW}$ or $\O_{SE}$. Also, the square $H_3$, at starting position, contains either $\O_{SW}$ or $\O_{SE}$.

        (ii) We can similarly prove this as (i).
    \end{proof}

    Using Lemma~\ref{lm:P in PS(u)_SW} and Lemma~\ref{lm:P_1 contains O_NW}, we have the following.

    \begin{figure}[htbp]
    \centering
    \begin{subfigure}[b]{.49\textwidth}
        \centering
        \includegraphics[page=3, width=75mm]{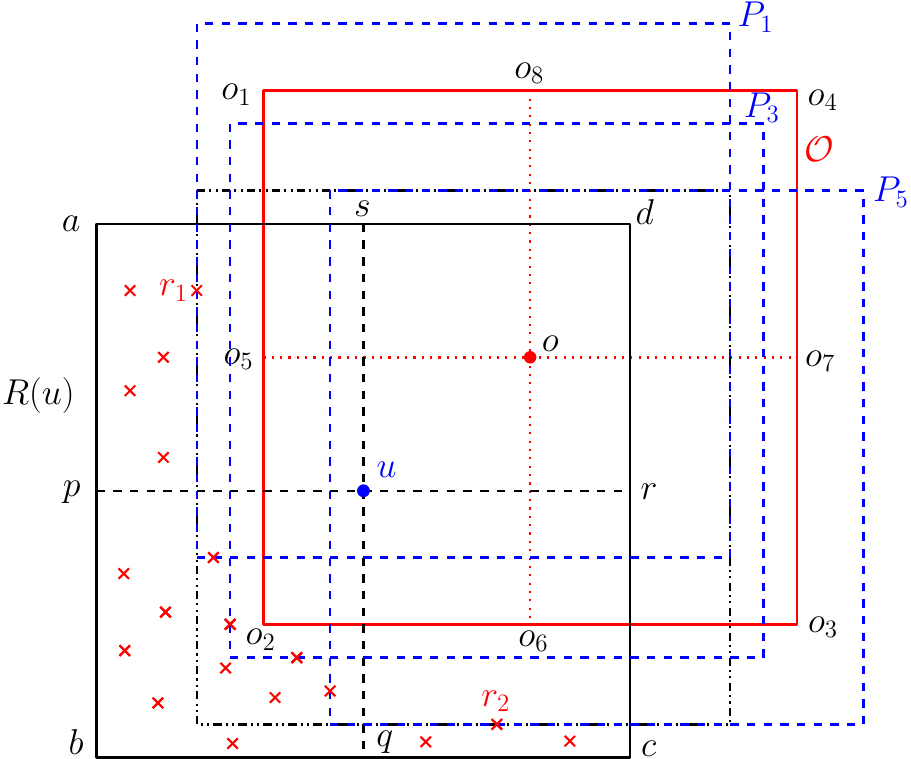}
        \caption{}
        \label{fig:P_1 contains O_NW}
    \end{subfigure}
    \hfill
    \begin{subfigure}[b]{.49\textwidth}
        \centering
        \includegraphics[page=4, width=72mm]{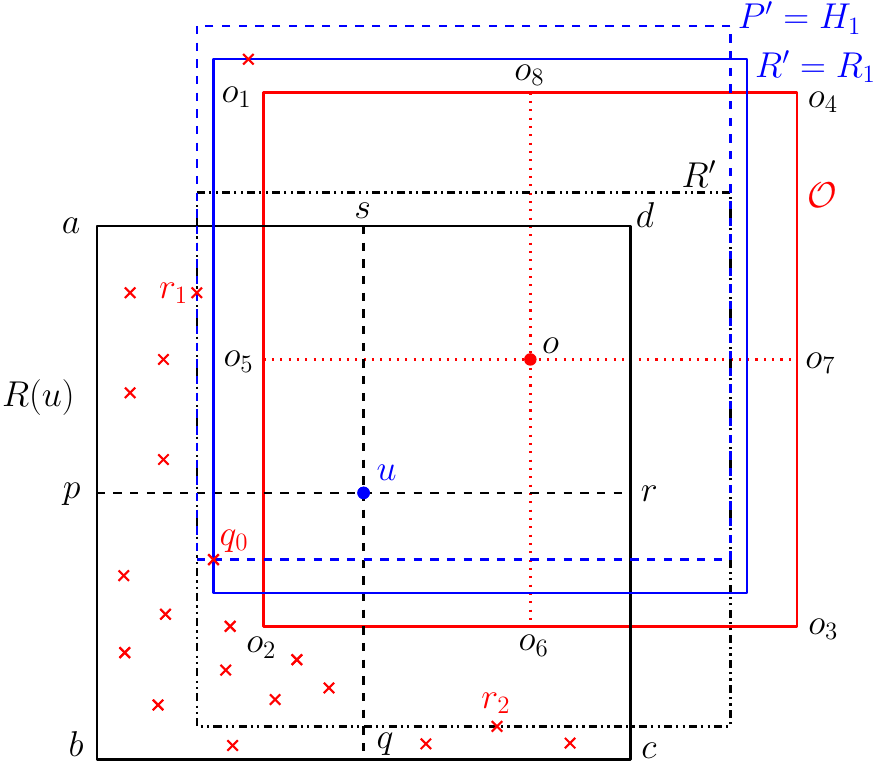}
        \caption{}
        \label{fig:type 1 existance}
    \end{subfigure}
    \caption{Let $\O$ (red) be an optimum square containing a blue point $u$. Also, let $R(u)$ (black) be the square centered at the blue point $u$. (a) The staircase squares $P_1, P_3$ and $P_5$ of $u$ contains $\O_{NW}, \O_{SW}$ and $\O_{SE}$, respectively. (b) The square $H_1$ (blue/dashed) is a staircase square of $u$ that contains $\O_{NW}$, and $R_1$ (blue) is a candidate square of $u$, obtained by moving $H_1$, also contains $\O_{NW}$.}
    \label{fig:Correctness_Lemma_2}
    \end{figure}

    \begin{lemma}\label{lm:type 1 existance}
        Let $u$ be an input blue point in $\O_{SW}$ such that $\Q(u)_{SW}\cap R(u)_{SW}\neq \emptyset$. Also, let $R'$ be a Type~1 candidate square of $u$, constructed by using a staircase square $P'\in \mathcal{PS}(u)_{SW}$ (as defined in Section~\ref{sec:type1}). Then, we have the following.
        \begin{enumerate}
        \item[(i)] If the staircase square $P'$ contains $\O_{NW}$, then $R'\neq \emptyset$, and it contains either $\O_{NW}$ or $\O_{SW}$. Also, the bottom edge of $R'$ does not lie above the bottom edge of $P'$.
        \item[(ii)] If the staircase square $P'$ contains $\O_{SE}$, then $R'\neq \emptyset$, and it contains either $\O_{SE}$ or $\O_{SW}$. Also, the left edge of $R'$ does not lie on the right side of the left edge of~$P'$.
        \item[(iii)] If the staircase square $P'$ contains $\O_{SW}$, then $R'\neq \emptyset$, and it contains $\O_{SW}$.
    \end{enumerate}
    \end{lemma}
    \begin{proof}
         Due to Lemma~\ref{lm:P in PS(u)_SW}, we have that $\mathcal{PS}(u)_{SW}$ is non-empty.
         % , and there exists a square $P\in\mathcal{PS}(u)_{SW}$ that contains $\O_{SW}$. 
         Assume $\mathcal{PS}(u)_{SW}=(P_1,\dots,P_{l})$, for some integer $l$ (as defined in Section~\ref{sec:possible_candidate_square}). Now, as defined in Section~\ref{sec:type1}, consider the three staircase squares $H_1=P_1, H_2=P_{\lceil\frac{l}{2}\rceil}$ and $H_3=P_{l}$. So, due to Lemma~\ref{lm:P_1 contains O_NW}, we have that the staircase square $H_j$, for $j\in [3]$, contains either $\O_{NW}$, $\O_{SW}$ or $\O_{SE}$. Also, recall that $R_j$ is a Type~1 candidate square of $u$, constructed by using the staircase square $H_j$ for each $j\in [3]$ (see Section~\ref{sec:type1}). 
         % Let us assume that $P_i=P$ for some $i\in[l]$. So, due to Lemma~\ref{lm:P_1 contains O_NW}, we have that $H_1=P_1$ contains either $\O_{NW}$ or $\O_{SW}$, and $H_3=P_{l}$ contains either $\O_{SW}$ or $\O_{SE}$. Also, we have that if $i\geq \lceil\frac{l}{2}\rceil$, the square $H_2=P_{\lceil\frac{l}{2}\rceil}$ contains either $\O_{NW}$ or $\O_{SW}$, and if $i\leq \lceil\frac{l}{2}\rceil$, the square $H_2=P_{\lceil\frac{l}{2}\rceil}$ contains either $\O_{SW}$ or $\O_{SE}$.

         (i) Suppose that $P'=H_j$ for some $j\in [3]$ and it contains $\O_{NW}$ (see Figure~\ref{fig:type 1 existance}). Then, $R'=R_j$ is the candidate square of $u$ constructed from the staircase square $P'=H_j$. Now, if $H_j$ contains no red point, then $R_j=H_j\in\S(u)$; in that case, we have nothing to prove. Otherwise, $H_j$ contains some red points. Now, we first argue the existence of $R_j$. Due to Observation~\ref{obs:Pj_SW} and the fact that $H_j$ contains $\O_{NW}$, we have that only the sub-square $H_{j(NW)}$ of $H_j$ contains some red points. In that case, we evoke \STR$(u, H_j)$ (see Section~\ref{sec:type1}). Here, we move it rightwards along the staircase $ST_{SW}(\Q(u)_{SW})$ until $H_{j(NW)}\cup H_{j(NE)}$ becomes red point free or $H_j$ does not contain~$u$. After that, if $H_{j(SE)}$ contains some red points and $u\in H_j$, then we evoke \UR$(u, H_j)$. Hence, before the left edge of $H_j$ touches the left edge of $\O$, either $H_j$ becomes red point free or $H_{j(SE)}$ contains some red points (note that in both cases, the point $u\in H_j$). In the former case, we set $R_j=H_j$, and in the latter case, the existence of the candidate square $R_j$ of $u$ can be shown similarly as Lemma~\ref{lm:type 2 existance}. Therefore, $R_j\neq \emptyset$.

         Now, we argue that $R_j$ contains either $\O_{NW}$ or $\O_{SW}$ (see Figure~\ref{fig:type 1 existance}). Recall that in this case, during the movement of $H_j$, the left edge of $H_j$ always lies between the left edge of $\O$ and the left edge of $R(u)$. Therefore, the left edge of $H_j$ never lies on the right side of the left edge of $\O$. As a result, $H_j$, a translated copy of $\O$, always contains either $\O_{NW}$ or $\O_{SW}$. Thus, $R_j$ contains either $\O_{NW}$ or $\O_{SW}$.

         Next, we argue that the bottom edge of $R_j$ does not lie above the bottom edge of $P'$ (see Figure~\ref{fig:type 1 existance}). During the movement of $H_j$, the square $H_j$, first, is moved in the rightward direction along the staircase, and when it is moved in the upward direction, the bottom edge of $H_j$ already lies below the bottom edge of $\O$. So, it is moved upwards until the bottom edge of $H_j$ touches the bottom edge of $\O$. Since the bottom edge of the staircase square $P'$ lies above the bottom edge of $\O$, the bottom edge of $R_j$ does not lie above the bottom edge of $P'$.

        (ii) We can similarly prove this as (i).
        
        (iii) We can similarly prove this as (i).
    \end{proof}

In the next lemma, we will prove that $\S(u)\neq\emptyset$, which guarantees that the Algorithm~\ref{alg:occp} indeed outputs a valid solution.

    \begin{lemma}\label{lm:set of candidate squares is non-empty}
         Let $u$ be an input blue point in $\O_{SW}$. Then, at least one non-empty candidate square of $u$ exists, i.e., $\S(u)\neq\emptyset$.
    \end{lemma}
    \begin{proof}
        Due to Lemma~\ref{lema:candidate_squares}, we have that either at least one non-empty Type~2 candidate square of $u$ exists or $\Q(u)_{SW}\cap R(u)_{SW}\neq \emptyset$. For the former case, we are done, i.e., $\S(u)\neq\emptyset$. In the latter case, the existence of a non-empty candidate square of $u$ follows from Lemma~\ref{lm:type 1 existance}.
    \end{proof}

\subsection{Competitive Analysis}\label{sec:comp_analy}

In this section, we analyze our algorithm to obtain its competitive ratio. First, we consider the following lemma.

    \begin{lemma}\label{lm: type 1 either or}
        Let $u$ be an input blue point in $\O_{SW}$ such that $\Q(u)_{SW}\cap R(u)_{SW}\neq \emptyset$. Then, either at least one Type~1 candidate square of $u$ contains $\O_{SW}$ or the union of the Type~1 candidate squares of $u$  contains $\O_{NW} \cup \O_{SE}$.
    \end{lemma}
    \begin{proof}
        First, due to Lemma~\ref{lm:P in PS(u)_SW}, we have that $\mathcal{PS}(u)_{SW}\neq \emptyset$.
        % , and there exist a square $P\in \mathcal{PS}(u)_{SW}$ that contains $\O_{SW}$. 
        Assume $\mathcal{PS}(u)_{SW}=(P_1,\dots,P_{l})$, for some integer $l$, is the sequence of $SW$ staircase squares of~$u$. Now, as defined in Section~\ref{sec:type1}, consider the three staircase squares $H_1=P_1, H_2=P_{\lceil\frac{l}{2}\rceil}$ and $H_3=P_{l}$. 
        % Let us assume that $P_i=P$ for some $i\in [l]$. 
        Due to Lemma~\ref{lm:P_1 contains O_NW}, we have that $P_1$ contains either $\O_{NW}$ or $\O_{SW}$, and the square $P_l$ contains either $\O_{SW}$ or $\O_{SE}$.
        Also, recall that $R_j$ is a candidate square of $u$ constructed by using the staircase square $H_j$ for each $j\in [3]$ (see Section~\ref{sec:type1}). So, due to Lemma~\ref{lm:type 1 existance}, we have that for each $j\in[3]$, the candidate square $R_j\neq\emptyset$.

        Since $P_1$ contains either $\O_{NW}$ or $\O_{SW}$, then $R_1$ contains either $\O_{NW}$ or $\O_{SW}$ (due to Lemma~\ref{lm:P_1 contains O_NW}~(i) and Lemma~\ref{lm:type 1 existance}~(i)). Since $P_l$ contains either $\O_{SW}$ or $\O_{SE}$, then $R_3$ contains either $\O_{SW}$ or $\O_{SE}$ (due to Lemma~\ref{lm:P_1 contains O_NW}~(ii) and Lemma~\ref{lm:type 1 existance}~(ii)). Therefore, we have either at least one $R_j$ contains $\O_{SW}$ or $\bigcup_{j=1}^{3} R_j$ contains $\O_{NW}\cup \O_{SE}$. Thus, the lemma follows. See Figure~\ref{fig:opt_square_cases1} and  Figure~\ref{fig:opt_square_case_5.3} for illustrations of the lemma in different cases.
    \end{proof}

In the following lemma, we prove that Algorithm~\ref{alg:occp} constructs at most $5\log_2 (m)$ squares that cover all the blue points contained in $\O_{SW}$, where $m\geq 2$.

\begin{lemma}\label{lema:5log(m)}
    Algorithm~\ref{alg:occp} constructs at most $5\log_2 (m)$ candidate squares that cover all the blue points in $\O_{SW}$, where $m\ (\geq 2)$ is the number of red points.
\end{lemma}

    \begin{figure}[htbp]
    \centering
    \begin{subfigure}[b]{.49\textwidth}
        \centering
        \includegraphics[page=1, width=75mm]{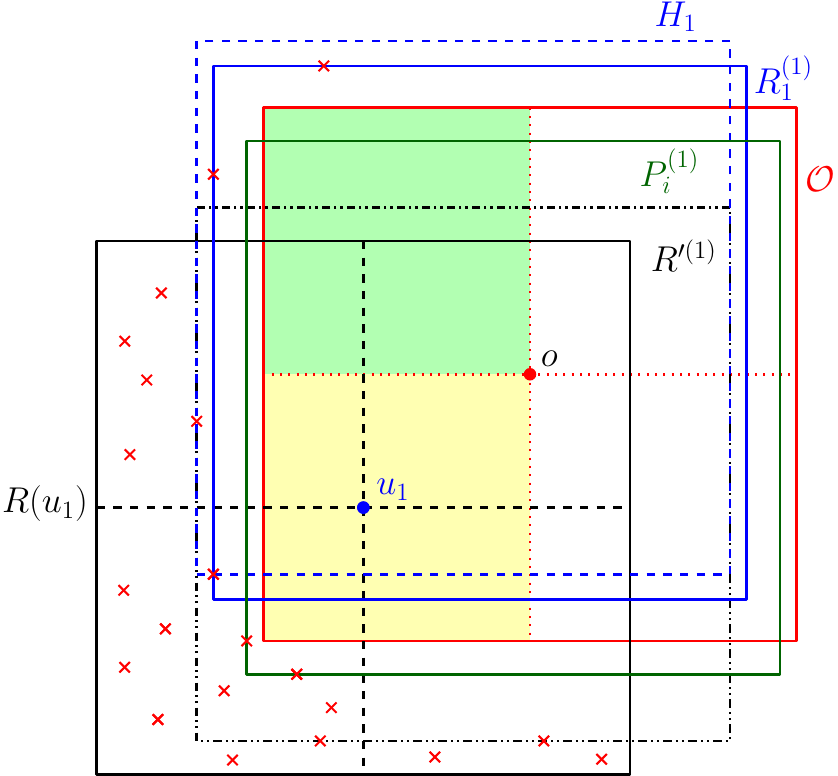}
        \caption{}
        \label{fig:main_lemma_1}
    \end{subfigure}
    \hfill
    \begin{subfigure}[b]{.49\textwidth}
        \centering
        \includegraphics[page=2, width=75mm]{Main_Lemma.pdf}
        \caption{}
        \label{fig:main_lemma_2}
    \end{subfigure}
    \caption{Let $\O$ (red) be an optimum square containing a blue point $u_1$. Also, let $R(u_1)$ (black) be the square centered at the blue point $u_1$. (a) Both the squares $H_1$ (blue/dashed) and $R^{(1)}_1$ (blue) contain $\O_{NW}$ (green shaded area). Also, the staircase square $P^{(1)}_i$ (dark green) contain $\O_{SW}$ (yellow shaded area). (b)~The candidate square $R^{(1)}_1$ (blue) of $u_1$ contains $(R(u_1)_{NW}\cup R(u_1)_{NE})\cap\O_{SW}$ (green shaded area).}
    \label{fig:main_lemma1}
    \end{figure}

\begin{proof}
    Let $u_1, u_2, \dots, u_k$ be the largest sub-sequence of the input sequence such that for each $i\in [k]$, $u_i\in\O_{SW}$ and it is uncovered upon its arrival. If any candidate square of $u_1$ contains $\O_{SW}$, then the lemma follows. W.l.o.g., assume that none of the candidate squares of~$u_1$ contains $\O_{SW}$. {Thus, due to Lemma~\ref{lema:candidate_squares}, we have $\Q(u_1)_{SW}\cap R(u_1)_{SW}\neq \emptyset$. As a result of this and Lemma~\ref{lm:P in PS(u)_SW}, $\mathcal{PS}(u_1)_{SW}\neq \emptyset$.} Let $\mathcal{PS}(u_1)_{SW}=(P^{(1)}_1,$ $P^{(1)}_2, \dots, P^{(1)}_{k_1})$, for some integer $k_1 \leq m$, be the sequence of $SW$ staircase squares of $u_1$ constructed by using the sequence of $SW$ staircase points $\Q(u_1)_{SW}$. Let $R_1^{(1)}, R_2^{(1)},$ $R_3^{(1)}~\in~\S(u_1)$ be the three candidate squares constructed by using three staircase squares $H_1, H_2$ and $H_3$, respectively, where $H_1=P^{(1)}_1, H_2=P^{(1)}_{\lceil\frac{k_1}{2}\rceil}$ and $H_3=P^{(1)}_{k_1}$ (as defined in Section~\ref{sec:type1}). Due to Lemma~\ref{lm:P_1 contains O_NW} and Lemma~\ref{lm:type 1 existance}, we have $R^{(1)}_j\neq\emptyset$ for each $j\in [3]$. Note that due to Lemma~\ref{lm: type 1 either or}, we have $\bigcup_{j=1}^3 R_j^{(1)}$ contains $\O_{NW}\cup\O_{SE}$. As a result, $u_1$ has at least two distinct candidate squares. 
    % From Lemma~\ref{lm:P in PS(u)_SW}, we have a staircase square of $u_1$, say $P^{(1)}_i\in \mathcal{PS}(u_1)_{SW}$, for some $i\in[k_1]$, that contains $\O_{SW}$ (see Figure~\ref{fig:main_lemma_1}).
    Now, we claim that $H_1$ contains $\O_{NW}$, and the square $H_3$ contains $\O_{SE}$.
    \begin{claim}\label{clm:H1_contains_R-NW}
    \begin{enumerate}
        \item[(i)] The staircase square $H_1=P^{(1)}_1$ contains $\O_{NW}$.
        \item[(ii)] The staircase square $H_3=P^{(1)}_{k_1}$ contains $\O_{SE}$.
    \end{enumerate}
        % The staircase square $H_1=P^{(1)}_1$ contains $\O_{NW}$, and the staircase square $H_3=P^{(1)}_{k_1}$ contains $\O_{SE}$.
    \end{claim}
    \begin{proof}
        (i) Due to Lemma~\ref{lm:P_1 contains O_NW}, we have that the square $H_1=P^{(1)}_1$ contains either $\O_{NW}$ or $\O_{SW}$. Since $R_1^{(1)}$ does not contain $\O_{SW}$, the square $H_1$ also does not contain $\O_{SW}$ (due to Lemma~\ref{lm:type 1 existance}~(iii)). As a result, we have that the staircase square $H_1=P^{(1)}_1$ contains $\O_{NW}$.
        
        % {Since the $SW$ vertex of $H_1$ lies in $R(u_1)_{SW}$, the left edge of $H_1$ lies on the right side of the left edge of $R(u_1)$ (see Figure~\ref{fig:main_lemma_1}). Again, since $P^{(1)}_i$ contains $\O_{SW}$, the left edge of $H_1=P^{(1)}_1$ lies on the left side of the left edge of $\O$. Therefore, the right edge of $H_1$ lies between the right edge of $R(u_1)$ and the right edge of $\O$. Similarly, the bottom edge of $H_1$ lies between the top edge of $R(u_1)_{SW}$ and the bottom edge of $P^{(1)}_i$. Since the blue point $u_1\in \O_{SW}$ and $P^{(1)}_i$ contains $\O_{SW}$, we have that the bottom edges of $\O$ also lies between the top edge of $R(u_1)_{SW}$ and the bottom edge of $P^{(1)}_i$. Since $R_1^{(1)}$ does not contain $\O_{SW}$, from Lemma~\ref{lema:candidate_squares}~(P2(d)), we have that $H_1$ also does not contain $\O_{SW}$. As a result, the bottom edge of $H_1$ lies above the bottom edge of $\O$; therefore, $H_1$, a translated copy of $\O$, contains $\O_{NW}$.}

        (ii) We can prove this in a similar way as (i).
    \end{proof}
    \noindent Due to Claim~\ref{clm:H1_contains_R-NW}~(i) and Lemma~\ref{lm:type 1 existance}~(i), the candidate square $R_1^{(1)}$ contains either $\O_{NW}$ or $\O_{SW}$. Recall that none of the candidate squares of $u_1$ contains $\O_{SW}$. Thus, the candidate square $R_1^{(1)}$ contains $\O_{NW}$ (see Figure~\ref{fig:main_lemma_1}). Similarly, due to Claim~\ref{clm:H1_contains_R-NW}~(ii) and Lemma~\ref{lm:type 1 existance}~(ii), the candidate square $R_3^{(1)}$ contains $\O_{SE}$.

    Since $R_2^{(1)}$ does not contain $\O_{SW}$, the staircase square $H_2=P^{(1)}_{\lceil\frac{k_1}{2}\rceil}$ also does not contain $\O_{SW}$ (due to Lemma~\ref{lm:type 1 existance}~(iii)). From Lemma~\ref{lm:P in PS(u)_SW}, we have a staircase square of $u_1$, say $P^{(1)}_i\in \mathcal{PS}(u_1)_{SW}$, for some $i\in[k_1]$, that contains $\O_{SW}$ (see Figure~\ref{fig:main_lemma_1}).  As a result, $i\neq \lceil\frac{k_1}{2}\rceil$. Now, in a similar way as  Claim~\ref{clm:H1_contains_R-NW}, one can prove the following claim.
    % Now, in a similar way as in Claim~\ref{clm:H1_contains_R-SW}, we can show that if $i > \lceil\frac{k_1}{2}\rceil$, the staircase square $H_2=P^{(1)}_{\lceil\frac{k_1}{2}\rceil}$ contains $\R_{NW}$, and if $i < \lceil\frac{k_1}{2}\rceil$, the square $H_2=P^{(1)}_{\lceil\frac{k_1}{2}\rceil}$ contains $\R_{SE}$. 
    \begin{claim}\label{clm:H2_contains_R-NW_R-SE}
    \begin{enumerate}
        \item[(i)] If $i > \lceil\frac{k_1}{2}\rceil$, the staircase square $H_2=P^{(1)}_{\lceil\frac{k_1}{2}\rceil}$ contains $\O_{NW}$.
        \item[(ii)] If $i < \lceil\frac{k_1}{2}\rceil$, the square $H_2=P^{(1)}_{\lceil\frac{k_1}{2}\rceil}$ contains $\O_{SE}$.
    \end{enumerate}
    \end{claim}
    % \begin{proof}
    % We can prove this in a similar way as  Claim~\ref{clm:H1_contains_R-NW}.
        % (i) We can prove this in a similar way as in Claim~\ref{clm:H1_contains_R-NW}(i).

        % (ii) We can prove this in a similar way as in Claim~\ref{clm:H1_contains_R-NW}(ii).
    % \end{proof}
    \noindent Due to Claim~\ref{clm:H2_contains_R-NW_R-SE}, Lemma~\ref{lm:type 1 existance} and the fact that none of the candidate squares of $u_1$ contains $\O_{SW}$, we have that if $i > \lceil\frac{k_1}{2}\rceil$, the candidate square $R^{(1)}_{2}$ contains $\O_{NW}$, and if $i < \lceil\frac{k_1}{2}\rceil$, the square $R^{(1)}_{2}$ contains $\O_{SE}$.

   \begin{figure}[htbp]
    \centering
    \begin{subfigure}[b]{.49\textwidth}
        \centering
        \includegraphics[page=3, width=75mm]{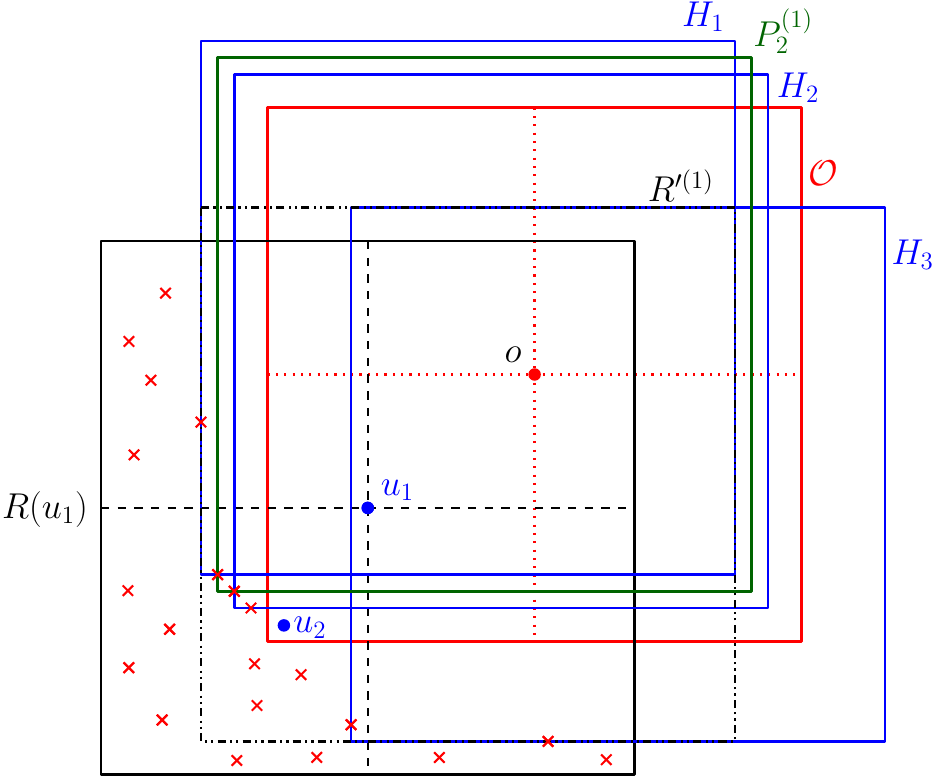}
        \caption{}
        \label{fig:main_lemma_3}
    \end{subfigure}
    \hfill
    \begin{subfigure}[b]{.49\textwidth}
        \centering
        \includegraphics[page=4, width=75mm]{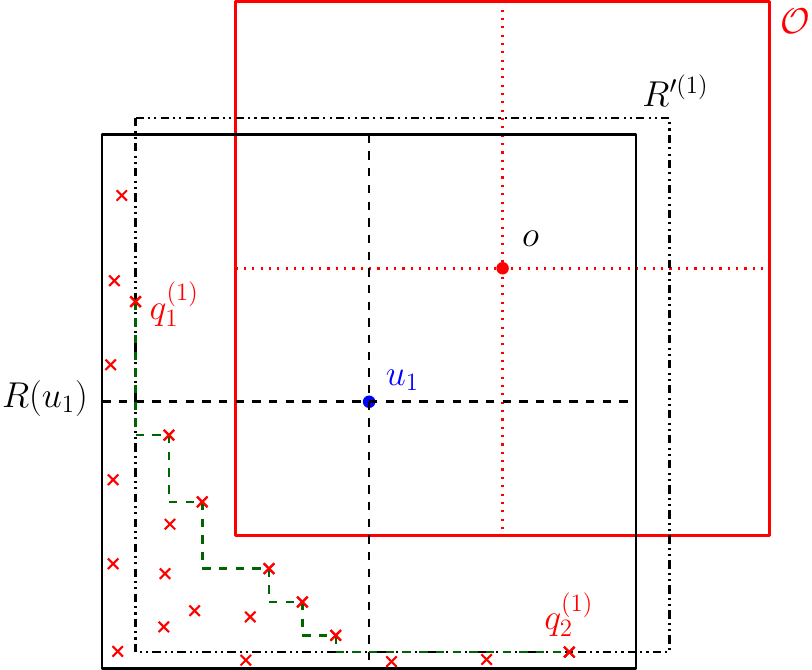}
        \caption{}
        \label{fig:main_lemma_4}
    \end{subfigure}
    \caption{Let $\O$ (red) be an optimum square containing a blue point $u_1$. Also, let $R(u_1)$ (black) be the square centered at the blue point $u_1$. (a) The blue point $u_2$ does not lie in the staircase square $P^{(1)}_2$ (dark green) of $u_1$. (b)~$R'^{(1)}$ (black/dashed) is a square whose left edge and the bottom edge pass through the points $q_1^{(1)}$ and $q_2^{(1)}$, respectively, and $\Q(u_1)_{SW}$ is the set of all the red (cross) points lying in the staircase (dark green/dashed).}
    \label{fig:main_lemma2}
    \end{figure}

    Now, consider the next blue point $u_2$. We claim that $u_2$ must be in $R(u_1)_{SW}\cap\O_{SW}$.
    \begin{claim}\label{clm:b2_Rb1SW}
        The blue point $u_2\in R(u_1)_{SW}\cap\O_{SW}$.
    \end{claim}
    \begin{proof}
        Since $u_1\in \O_{SW}$, each of the sub-squares of $R(u_1)$ has a non-empty intersection with $\O_{SW}$ (see Figure~\ref{fig:main_lemma_2}). Now, as $u_1\in R_1^{(1)}$ and the square $R_1^{(1)}$ does not contain $\O_{SW}$,  the bottom edge of $R_1^{(1)}$ lies between the top edge of $R(u_1)_{SW}$ and the bottom edge of $\O$. This implies that the top edge of $R_1^{(1)}$ lies above the top edge of $\O$. Similarly, we can show that the left edge of $R_1^{(1)}$ lies between the left edge of $R(u_1)$ and the left edge of $\O$. This implies that the right edge of $R_1^{(1)}$ lies between the right edge of $R(u_1)$ and the right edge of $\O$. Hence, we have that $R^{(1)}_1$ contains $(R(u_1)_{NW}\cup R(u_1)_{NE})\cap\O_{SW}$. Similarly, we can show that $R^{(1)}_3$ contains $(R(u_1)_{SE}\cup R(u_1)_{NE})\cap\O_{SW}$. Due to this and since $R(u_1)$ contains $\O_{SW}$, the next blue point $u_2$ must lie in $R(u_1)_{SW}$. From our assumption, we have that $u_2\in \O_{SW}$; hence, the claim follows.
    \end{proof}
    \noindent Now again, if a candidate square of $u_2$ contains $\O_{SW}$, then we are done. Otherwise, we consider the sequence of $SW$ staircase squares $\mathcal{PS}(u_2)_{SW}= (P^{(2)}_1, P^{(2)}_2, \dots, P^{(2)}_{k_2})$ of $u_2$ constructed by using the sequence of $SW$ staircase points $\Q(u_2)_{SW}\ (\neq \emptyset)$, and select at most three candidate squares of $u_2$ obtained from three squares in $\mathcal{PS}(u_2)_{SW}$.
    % Next we claim that $u_2\notin P^{(1)}_j$ for each $j\in [\frac{k_1}{2}]\cup\{k_1\}$.

    \begin{claim}\label{clm:b2_P1j}
    \begin{enumerate}
        \item[(i)] If $i > \lceil\frac{k_1}{2}\rceil$, the blue point $u_2\notin P^{(1)}_j$ for each $j\in \{1,\dots,\lceil\frac{k_1}{2}\rceil\}\cup\{k_1\}$.
        \item[(ii)] If $i < \lceil\frac{k_1}{2}\rceil$, the blue point $u_2\notin P^{(1)}_j$ for each $j\in \{1\}\cup \{\lceil\frac{k_1}{2}\rceil,\dots, k_1\}$.
    \end{enumerate}
    \end{claim}
    \begin{proof}
        (i) Due to Lemma~\ref{lm:type 1 existance}~((i)-(ii)), Claim~\ref{clm:H1_contains_R-NW} and Claim~\ref{clm:H2_contains_R-NW_R-SE}(i), we have that the bottom edge of $R_1^{(1)}$ and $R_2^{(1)}$ do not lie above the bottom edge of $H_1=P^{(1)}_1$ and $H_2=P^{(1)}_{\lceil\frac{k_1}{2}\rceil}$, respectively, and the left edge of $R_3^{(1)}$ does not lie on the right side of the left edge of $H_3=P^{(1)}_{k_1}$ (see Figure~\ref{fig:main_lemma_3}). Also, $u_2\in\O_{SW}\cap R(u_1)_{SW}$ (due to Claim~\ref{clm:b2_Rb1SW}). As a result, the point $u_2$ lies below the bottom edges of $H_1, H_2$, and on the left side of the left edge of $H_3$. Therefore, we have $u_2\notin P^{(1)}_j$ for each $j\in \{1,\dots,\lceil\frac{k_1}{2}\rceil\}\cup\{k_1\}$.

        (ii) We can similarly prove this as (i).
    \end{proof}
    \noindent Now, we claim that $\Q(u_2)_{SW}\subseteq \Q(u_1)_{SW}$. As a result of this and due to  Claim~\ref{clm:b2_P1j}, we have that if $i > \lceil\frac{k_1}{2}\rceil$, then $\mathcal{PS}(u_2)_{SW}$ is a sub-sequence of $(P^{(1)}_{(\lceil \frac{k_1}{2}\rceil+1)},P^{(1)}_{(\lceil \frac{k_1}{2}\rceil+2)},$ $\dots,P^{(1)}_{(k_1-1)})$, and if $i < \lceil\frac{k_1}{2}\rceil$, then $\mathcal{PS}(u_2)_{SW}$ is a sub-sequence of $(P^{(1)}_{2},P^{(1)}_{3},$ $\dots,P^{(1)}_{(\lceil \frac{k_1}{2}\rceil-1)})$. Therefore, $|\mathcal{PS}(u_2)_{SW}|=k_2<\frac{k_1}{2}\leq \frac{m}{2}$.

    \begin{claim}\label{clm:Q(u)}
        $\Q(u_2)_{SW}\subseteq \Q(u_1)_{SW}$, where $\Q(u_1)_{SW}$ and $\Q(u_2)_{SW}$ are the staircase points set of $u_1$ and $u_2$, respectively.
    \end{claim}
    \begin{proof}
        W.l.o.g, we assume that $i > \lceil\frac{k_1}{2}\rceil$ and the other case $i < \lceil\frac{k_1}{2}\rceil$ is similar in nature. Let $q_1^{(1)}$ and $q_2^{(1)}$ be the initial and the terminal points of the staircase $ST_{SW}(\Q(u_1)_{SW})$, respectively  (see Figure~\ref{fig:main_lemma_4}). As defined in Section~\ref{sec:possible_candidate_square}, consider the square $R'^{(1)}$ whose left and bottom edges pass through the points $q_1^{(1)}$ and $q_2^{(1)}$, respectively.

        % Let $q_1^{(1)}\in \Q(u_1)$ be the point such that $q_1^{(1)}(x)\leq v(x)$ and $q_1^{(1)}(y)\geq v(y)$ for all $v \in \Q(u_1)$, and $q_2^{(1)}\in \Q(u_1)$ be the point such that $q_2^{(1)}(x)\geq v(x)$ and $q_2^{(1)}(y)\leq v(y)$ for all $v \in \Q(u_1)$ (see Figure~\ref{fig:main_lemma_4}). As defined in Section~\ref{sec:candidate_square} (Case~2.2, Case~3a and Case~4), consider the square $R'^{(1)}$ whose left and bottom edges pass through the points $q_1^{(1)}$ and $q_2^{(1)}$, respectively.
        
        % Let $q_1^{(2)}\in \Q(u_2)$ be the point such that $q_1^{(2)}(x)\leq u(x)$ and $q_1^{(2)}(y)\geq u(y)$ for all $u \in \Q(u_2)$, and $q_2^{(2)}\in \Q(u_2)$ be the point such that $q_2^{(2)}(x)\geq u(x)$ and $q_2^{(2)}(y)\leq u(y)$ for all $u \in \Q(u_2)$. Let $R'^{(2)}$ be a square whose left edge and bottom edge pass through the points $q_1^{(1)}$ and $q_2^{(1)}$, respectively.
        
        Let $v_1$ and $v_2$ be the two red points such that $v_1\in \Q(u_1)_{SW}$ lies on the bottom edge of $H_2=P^{(1)}_{\lceil\frac{k_1}{2}\rceil}$ and $v_2\in \Q(u_1)_{SW}$ lies on the left edge of $H_3=P^{(1)}_{k_1}$ (see Figure~\ref{fig:main_lemma_5}). Now, we claim that $v_1\in R(u_2)_{NW}$ and $v_2\in R(u_2)_{SE}$. Since $v_1, v_2, u_2\in R(u_1)_{SW}$ and the fact that $R(u_2)$ is a translated copy of $R(u_1)$, the square $R(u_2)$ contains the points $v_1$ and $v_2$. Due to Claim~\ref{clm:H1_contains_R-NW} and Claim~\ref{clm:b2_P1j}, we have $u_2(y)\leq v_1(y)$ and $u_2(x)\geq v_1(x)$. As a result of this and  the fact $v_1\in R(u_2)$, the point $v_1 \in R(u_2)_{NW}$. Similarly, we can argue that $v_2 \in R(u_2)_{SE}$.

    % \begin{figure}[htbp]
    % \centering
    % \includegraphics[width=80 mm]{Figures/main_lemma_5.pdf}
    % \caption{$\R$ (red) is an optimum square and $R(u_1), R(u_2)$ (black) be the squares centered at the blue point $u_1, u_2$, respectively. $\Q(u_1)$ is the set of all the red (cross/box) points while $\Q(u_2)$ only contains all the red (box) points. $R'^{(2)}\cap R(u_2)_{SW}$ is the green shaded area contained in $R'^{(1)}\cap R(u_1)_{SW}$.}
    % \label{fig:main_lemma3}
    % \end{figure}

    \begin{figure}[htbp]
    \centering
    \begin{subfigure}[b]{.49\textwidth}
        \centering
        \includegraphics[page=5, width=75mm]{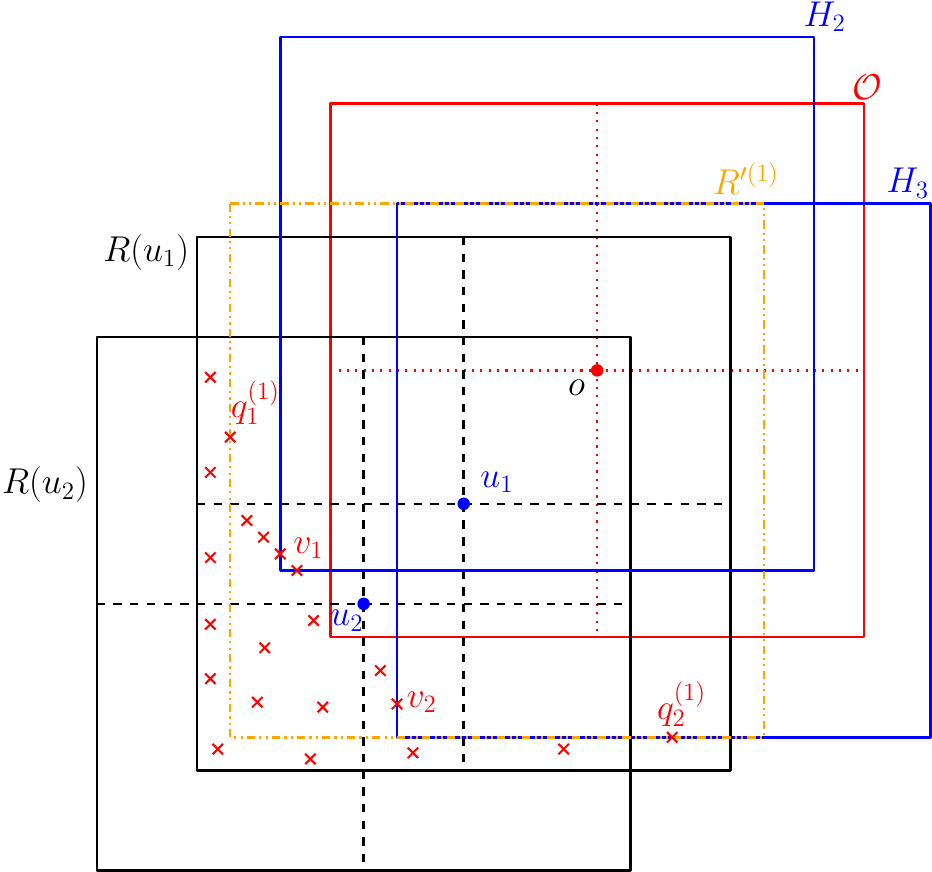}
        \caption{}
        \label{fig:main_lemma_5}
    \end{subfigure}
    \hfill
    \begin{subfigure}[b]{.49\textwidth}
        \centering
        \includegraphics[page=6, width=75mm]{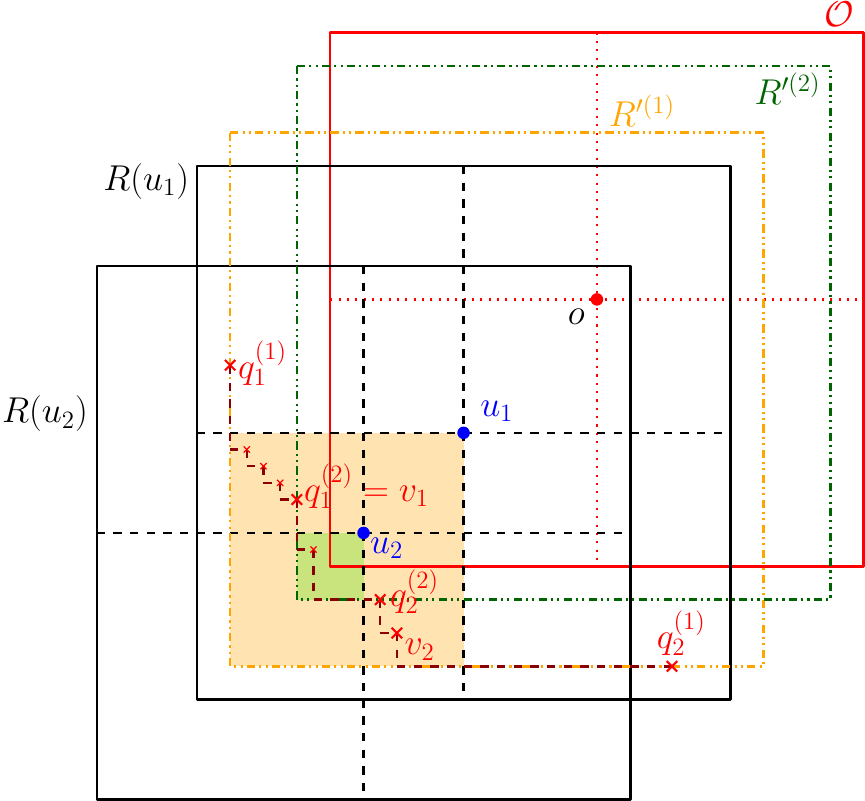}
        \caption{}
        \label{fig:main_lemma_6}
    \end{subfigure}
    \caption{Let $\O$ (red) be an optimum square containing two blue points $u_1$ and $u_2$. Also, let $R(u_1)$ and $R(u_2)$ (black) be two squares centered at the blue point $u_1$ and $u_2$, respectively. (a) The red points $v_1$ and $v_2$ lie in $R(u_2)_{NW}$ and $R(u_2)_{SE}$, respectively. (b)~$R'^{(2)}_{SW}\cap R(u_2)_{SW}$ (green shaded area) is contained in $R'^{(1)}_{SW}\cap R(u_1)_{SW}$ (orange shaded area).}
    \label{fig:main_lemma3}
    \end{figure}

        Let $q_1^{(2)}$ be the nearest red point in $R(u_2)_{NW}$ to the right edge of $R(u_2)_{NW}$, and $q_2^{(2)}$ be the nearest red point in $R(u_2)_{SE}$ to the top edge of $R(u_2)_{SE}$ (see Figure~\ref{fig:main_lemma_6} and \ref{fig:main_lemma_7}). In other words, $q_1^{(2)}$ and $q_2^{(2)}$ be the initial and the terminal points of the staircase $ST_{SW}(\Q(u_2)_{SW})$, respectively. Since $v_1 \in R(u_2)_{NW}$ and $v_2 \in R(u_2)_{SE}$, we have $q_1^{(2)}(x) \geq v_1(x)$ and $q_2^{(2)}(y) \geq v_2(y)$, respectively. As defined in Section~\ref{sec:possible_candidate_square}, consider the square $R'^{(2)}$ whose left edge and bottom edge pass through the points $q_1^{(2)}$ and $q_2^{(2)}$, respectively. 
        % Note that the $SW$ vertex of $R'^{(2)}$ lies inside the sub-square $R(u_1)_{SW}$, and the point $u_2\in R'^{(2)}$.
        
        % Due to Lemma~\ref{lema:5log(m)} (P2(a)) and $R'^{(2)}_{SW}\cap R(u_2)_{SW}\cap \P_r \supseteq \Q(u_2)\cap R'^{(2)}_{SW}$, we have that $R'^{(2)}_{SW}\cap R(u_2)_{SW}$ must contain some red points. 
        Observe that the blue point $u_2$ and the $SW$ vertex of $R'^{(2)}$ both lie inside $R(u_1)_{SW}$. As a result, we have that $R'^{(2)}_{SW}\cap R(u_2)_{SW} \subseteq R(u_1)_{SW}$ (see Figure~\ref{fig:main_lemma_6} and \ref{fig:main_lemma_7}). Now, we also show that $R'^{(2)}_{SW}\cap R(u_2)_{SW}\subseteq R'^{(1)}_{SW}$. Since $q_1^{(1)}(x)\leq v_1(x)\leq q_1^{(2)}(x)$, the left edge of $R'^{(1)}$ lies on the left side of the left edge of $R'^{(2)}$, and since $q_2^{(1)}(y)\leq v_2(y)\leq q_2^{(2)}(y)$, the bottom edge of $R'^{(1)}$ lies below the bottom edge of $R'^{(2)}$. Hence, the $SW$ vertex of $R'^{(2)}$ lies inside the sub-square $R'^{(1)}_{SW}$ (see Figure~\ref{fig:main_lemma_6} and \ref{fig:main_lemma_7}). Due to this and the fact that the blue point $u_2\in R'^{(1)}_{SW}$, we have that $R'^{(2)}_{SW}\cap R(u_2)_{SW} \subseteq R'^{(1)}_{SW}$. In other words, $R'^{(2)}_{SW}\cap R(u_2)_{SW} \subseteq R(u_1)_{SW}\cap R'^{(1)}_{SW}$ (see Figure~\ref{fig:main_lemma_6} and \ref{fig:main_lemma_7}).

        % Now we also show that $R'^{(2)}_{SW}\cap R(u_2)_{SW}\cap \P_r\subseteq R'^{(1)}_{SW}\cap \P_r$. Since $q_1^{(1)}(x)\leq v_1(x)$, the left edge of $R'^{(1)}$ lies on the left side of the left edge of $R'^{(2)}$, and since $q_2^{(1)}(y)\leq v_2(y)$, the bottom edge of $R'^{(1)}$ lies below the bottom edge of $R'^{(2)}$. Hence, the $SW$ vertex of $R'^{(2)}$ lies inside the sub-square $R'^{(1)}_{SW}$ (see Figure~\ref{fig:main_lemma_6} and \ref{fig:main_lemma4}). Recall that $q_1^{(1)}\in R(u_1)_{NW}$. Because of this, the left edge of $R'^{(1)}$ lies between the left edge of $R(u_1)$ and the left edge of $\R$. Similarly, $q_2^{(1)}\in R(u_1)_{SE}$ implies that the bottom edge of $R'^{(1)}$ lies between the bottom edge of $R(u_1)$ and the bottom edge of $\R$. Thus, we have that $R'^{(1)}$ contains $\R_{SW}$, and therefore, $u_2\in R'^{(1)}$. Since $u_1\in R'^{(1)}_{SW}$ and $u_1(x)>u_2(x)$ and $u_1(y)>u_2(y)$, we have $u_2\in R'^{(1)}_{SW}$. Now, since $u_2$ and the $SW$ vertex of $R'^{(2)}$ both lie inside $R'^{(1)}_{SW}$, we have that $R'^{(2)}_{SW}\cap R(u_2)_{SW}\cap \P_r\subseteq R'^{(1)}_{SW}\cap \P_r$.

    % \begin{figure}[htbp]
    % \centering
    % \includegraphics[width=90 mm]{Figures/main_lemma_7.pdf}
    % \caption{$\O$ (red) is an optimum square and $R(u_1), R(u_2)$ (black) be the squares centered at the blue point $u_1, u_2$, respectively. $R'^{(2)}_{SW}\cap R(u_2)_{SW}$ (green shaded area) is contained in $R'^{(1)}_{SW}\cap R(u_1)_{SW}$ (orange shaded area).}
    % \label{fig:main_lemma4}
    % \end{figure}

    \begin{figure}[htbp]
    \centering
    \begin{subfigure}[b]{.49\textwidth}
        \centering
        \includegraphics[page=7, width=80mm]{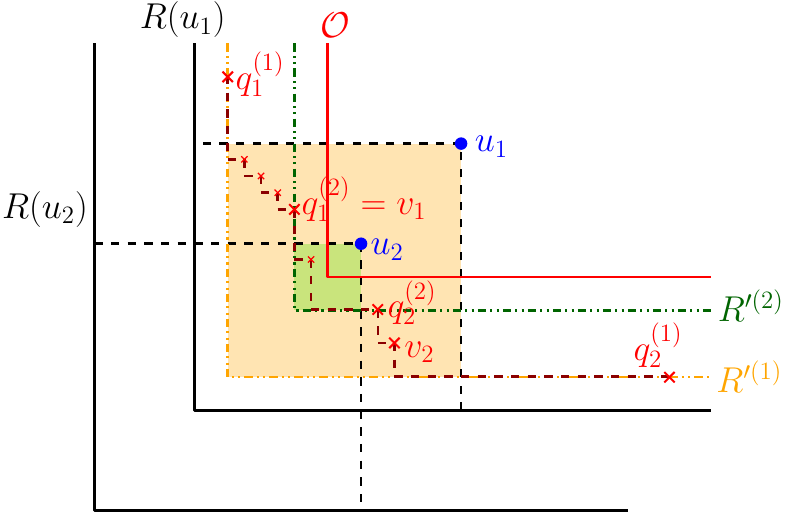}
        \caption{}
        \label{fig:main_lemma_7}
    \end{subfigure}
    \hfill
    \begin{subfigure}[b]{.49\textwidth}
        \centering
        \includegraphics[page=8, width=80mm]{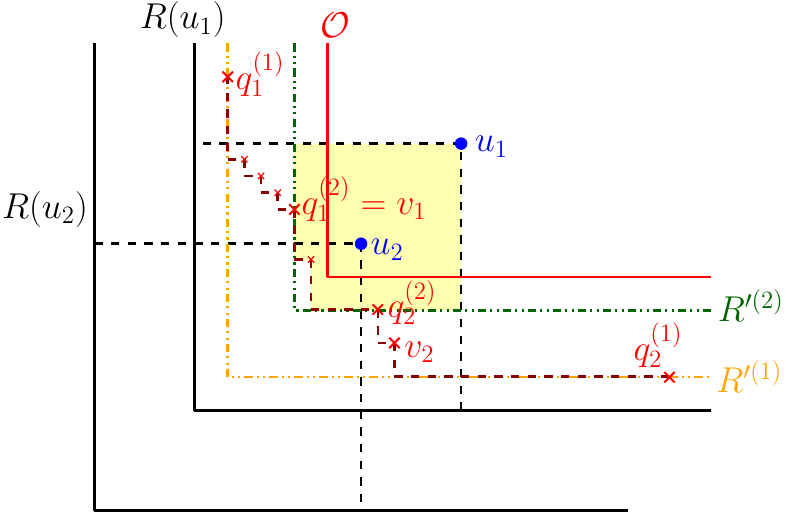}
        \caption{}
        \label{fig:main_lemma_8}
    \end{subfigure}
    \caption{Let $\O$ (red) be an optimum square containing two blue points $u_1, u_2$, and $R(u_1), R(u_2)$ (black) be two squares centered at the blue point $u_1, u_2$, respectively. (a) $R'^{(2)}_{SW}\cap R(u_2)_{SW}$ (green shaded area) is contained in $R'^{(1)}_{SW}\cap R(u_1)_{SW}$ (orange shaded area). (b)~The yellow shaded region is the area bounded by the left edge of $R'^{(2)}$, the staircase $ST_{SW}(\Q(u_1)_{SW})$, the bottom edge of $R'^{(2)}$, the right edge of $R(u_1)_{SW}$ and the top edge of $R(u_1)_{SW}$ whose interior is red point free.}
    \label{fig:main_lemma4}
    \end{figure}

        % Therefore, $R'^{(2)}_{SW}\cap R(u_2)_{SW}\cap \P_r\subseteq R(u_1)_{SW}\cap R'^{(1)}_{SW}\cap \P_r$. Recall, from the definition of the staircase points set, that $\Q(u_1)\cap R'^{(1)}_{SW} \subseteq R(u_1)_{SW}\cap R'^{(1)}_{SW}\cap \P_r$ and $\Q(u_2)\cap R'^{(2)}_{SW} \subseteq R(u_2)_{SW}\cap R'^{(2)}_{SW}\cap \P_r$. Now, $\Q(u_2)\setminus\{q_1^{(2)}, q_2^{(2)}\}\subseteq R'^{(2)}_{SW}$. Since $q_1^{(2)}, q_2^{(2)} \in \Q(u_1)$, we only need to show that $\Q(u_2)\cap R'^{(2)}_{SW}\subseteq \Q(u_1)$. Due to Lemma~\ref{lema:5log(m)} (P2(a)), note that $\Q(u_2)\cap R'^{(2)}_{SW}\neq \emptyset$.
        % Let $v$ be a point in $\Q(u_2)\cap R'^{(2)}_{SW}$. Then, there does not exist any point $v'\in R'^{(2)}_{SW}\cap R(u_2)_{SW}\cap\P_r$ such that $v(x)<v'(x)$ and $v(y)< v'(y)$.

        Since $q_1^{(2)}, q_2^{(2)} \in \Q(u_1)_{SW}$ and the fact that the left edge and bottom edge of $R'^{(2)}$ pass through the points $q_1^{(2)}$ and $q_2^{(2)}$, respectively, we have that the square $R'^{(2)}$ contains a subset $S\subseteq\Q(u_1)_{SW}$ of red points. Note that $S\setminus\{q_1^{(2)}, q_2^{(2)}\} \subseteq R'^{(2)}_{SW}\cap R(u_2)_{SW}$. 
        Since, from the definition, $\Q(u_1)_{SW}\setminus\{q_1^{(1)}, q_2^{(1)}\}=D_{SW}(R'^{(1)}_{SW}\cap R(u_1)_{SW}\cap \P_r)$ and also $R'^{(2)}$ contains the subset $S\subseteq\Q(u_1)_{SW}$, we have that the interior of the area bounded by the left edge of $R'^{(2)}$, the staircase $ST_{SW}(\Q(u_1)_{SW})$, the bottom edge of $R'^{(2)}$, the right edge of $R(u_1)_{SW}$ and the top edge of $R(u_1)_{SW}$ is red point free (see Figure~\ref{fig:main_lemma_8}). Due to this and $R'^{(2)}_{SW}\cap R(u_2)_{SW} \subseteq R(u_1)_{SW}\cap R'^{(1)}_{SW}$ and $S\setminus\{q_1^{(2)}, q_2^{(2)}\} \subseteq R'^{(2)}_{SW}\cap R(u_2)_{SW}$, we have that $S\setminus\{q_1^{(2)}, q_2^{(2)}\}$ is the set of $SW$ dominating point of $R'^{(2)}_{SW}\cap R(u_2)_{SW}\cap \P_r$, i.e., $S\setminus\{q_1^{(2)}, q_2^{(2)}\}= D_{SW}(R'^{(2)}_{SW}\cap R(u_2)_{SW}\cap \P_r)$. Hence, from the definition of the staircase point set, we have $\Q(u_2)_{SW}=S\subseteq \Q(u_1)_{SW}$, and so the claim follows.
        \end{proof}
        
        % Now, we show that $v\in \Q(u_1)$. For a contradiction, let us assume that $v\notin \Q(u_1)$. Then there exists a point $v''\in R'^{(1)}_{SW}\cap R(u_1)_{SW}\cap\P_r$ such that $v(x)<v''(x)$ and $v(y)< v''(y)$. Since $u_2\in R'^{(1)}_{SW}\cap R(u_1)_{SW}$, the square $R(u_2)$ contains $R'^{(1)}_{SW}\cap R(u_1)_{SW}$, and so $v''\in R(u_2)$.  Now, $v''\notin R(u_2)_{NE}$ because $R(u_2)_{NE}$ is contained in $\R$. Again, $v''\notin R(u_2)_{NW}$ because $q_1^{(2)}$ is the nearest red point in $R(u_2)_{NW}$ to the right edge of $R(u_2)_{NW}$ and $q_1^{(2)}\leq v(x)< v''(x)$. Similarly, we can show that $v''\notin R(u_2)_{SE}$. Also $v''\notin R(u_2)_{SW}$ because $v\in \Q(u_2)$. This contradicts our assumption that $v\notin \Q(u_1)$, and therefore, $\Q(u_2)\cap R'^{(2)}_{SW}\subseteq \Q(u_1)$. Hence,~the~claim~follows.

    % \begin{figure}[htbp]
    % \centering
    % \includegraphics[width=90 mm]{Figures/main_lemma_8.pdf}
    % \caption{The yellow shaded region is the area bounded by the left edge of $R'^{(2)}$, the staircase $ST_{SW}(\Q(u_1))$, the bottom edge of $R'^{(2)}$, the right edge of $R(u_1)_{SW}$ and the top edge of $R(u_1)_{SW}$ whose interior is red point free.}
    % \label{fig:main_lemma5}
    % \end{figure}

        % \vspace{5mm}
        % Let $v$ be a point in $\Q(u_2)\cap R'^{(2)}_{SW}$. Then there does not exist any point $v'\in R'^{(2)}_{SW}\cap R(u_2)_{SW}\cap\P_r$ such that $v(x)<v'(x)$ and $v(y)< v'(y)$. Since $R'^{(2)}_{SW}\cap R(u_2)_{SW}\cap\P_r\subseteq R(u_1)_{SW}\cap R'^{(1)}_{SW}\cap \P_r$, we have $v\in \Q(u_1)$. Therefore, $\Q(u_2)\subseteq \Q(u_1)$.

    Similarly, for the case of $u_3$, either we have a candidate square of $u_3$ containing $\O_{SW}$ or we select at most three candidate squares of $u_3$ obtained from three squares in $\mathcal{PS}(u_3)_{SW}$ such that the sequence $\mathcal{PS}(u_3)_{SW}$ is a sub-sequence of either $(P^{(2)}_{2},\dots, P^{(2)}_{(\lceil \frac{k_2}{2}\rceil-1)})$ or $(P^{(2)}_{(\lceil \frac{k_2}{2}\rceil+1)},\dots, P^{(2)}_{(k_2-1)})$. Thus, $|\mathcal{PS}(u_3)_{SW}|=k_3<\frac{k_2}{2}<\frac{m}{2^2}$. In each step to cover an uncovered blue point in $\O_{SW}$, either the algorithm selects a candidate square containing $\O_{SW}$ or the size of the staircase squares set is reduced by at least half, and each such staircase squares set always contains a square that covers $\O_{SW}$ (due to Lemma~\ref{lm:P in PS(u)_SW}). As a result, before the blue point $u_{(\log_2 (k_1) + 1)}$ is introduced, Algorithm~\ref{alg:occp} must construct a candidate square for some blue point $u_j$, where $j\in [\log_2 (k_1)]$, such that a candidate square of $u_j$ contains $\O_{SW}$. Hence, after $\log_2 (k_1)$th blue point is introduced, all the blue points $u_j$, for $(\log_2 (k_1) + 1)\leq j\leq k$, will be already covered by the previously constructed candidate squares. Therefore, $k\leq \log_2 (k_1)$. Since for each blue point $u_j$, for $j\in [\log_2 (k_1)]$, Algorithm~\ref{alg:occp} reports at most five candidate squares, we have that Algorithm~\ref{alg:occp} constructs at most $5 \log_2 (k_1) \leq 5 \log_2 (m)$ squares that cover $\O_{SW}$. This completes the proof of the lemma.

    % Similarly, for the other sub-squares of $\O$, the Algorithm~\ref{alg:occp} needs at most $3\cdot\log_2 (m)$ candidate squares to cover the sub-square of $\R$. Therefore, to cover all the blue points covered by $\R$, the Algorithm~\ref{alg:occp} chooses at most $4 \cdot 3\cdot\log_2 (m) = 12\cdot\log_2 (m)$ candidate squares. Therefore, the lemma follows.
\end{proof}

Similarly, we can prove an equivalent statement of Lemma~\ref{lema:5log(m)} for other sub-squares of $\O$ such as $\O_{NW}, \O_{NE}$ and $\O_{SE}$.

% Similarly, we can show that the Algorithm~\ref{alg:occp} needs at most $5\log_2 (m)$ candidate squares to cover all the blue points in other sub-squares of $\R$ such as $\R_{NW}, \R_{NE}$ and $\R_{SE}$, where $m\geq 2$.

\begin{lemma}\label{lema:20log(m)}
    Algorithm~\ref{alg:occp} constructs at most $20\log_2 (m)$ candidate squares that cover all the blue points in $\O$, where $m\ (\geq 2)$ is the number of red points.
\end{lemma}
\begin{proof}
    From Lemma~\ref{lema:5log(m)}, we have that at most $5\log_2 (m)$ squares reported by Algorithm~\ref{alg:occp} are sufficient to cover all the blue points that lie in only one sub-square of $\O$.
    % to cover all the blue points that lie in only one sub-square of $\O$, the Algorithm~\ref{alg:occp} needs at most $5\log_2 (m)$ candidate squares. 
    Hence, Algorithm~\ref{alg:occp} constructs at most $20\log_2 (m)$ candidate squares that cover all the blue points in $\O$. Thus, the lemma follows.
\end{proof}

\begin{figure}[htbp]
\centering
\includegraphics[page=1, width=50mm]{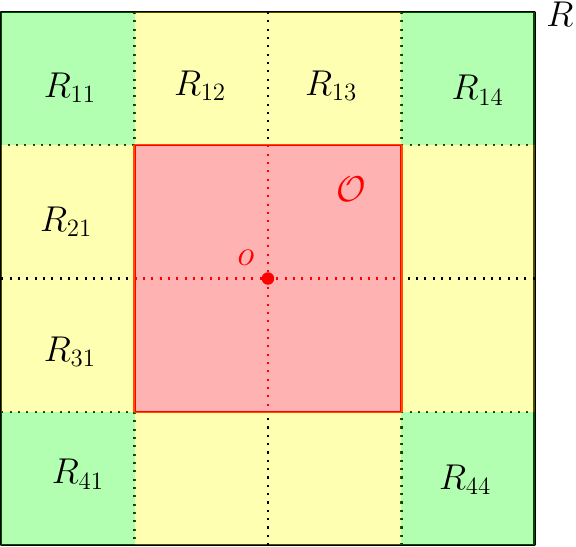}
\caption{Let $\O$ (red) be an optimum square centered at $o$, and $R$ (black) be a square centered at $o$ of side length 2. Here, the green shaded regions denote the sub-squares $R_{ij}$ for $i, j\in\{1,4\}$; the red shaded regions denote the sub-squares $R_{ij}$ for $i, j\in\{2,3\}$, and the yellow shaded regions denote the remaining sub-squares.}
\label{fig:sub_squares_m1}
\end{figure}

\begin{lemma}\label{lema:m_1}
     If $|\P_r|= 1$, i.e., there is exactly one red point, then the Algorithm~\ref{alg:occp} reports at most~$6$ candidate squares that cover all the blue points in $\O$.
\end{lemma}

\begin{proof}
    Let $\P_r=\{r\}$. Recall that $\O$ is an optimum square centered at a point $o$. Let $R$ be an axis-parallel square centered at $o$ of side length 2. We can partition $R$ into 16 sub-squares of side length $\frac{1}{2}$ as shown in Figure~\ref{fig:sub_squares_m1}. We denote the sub-squares by $R_{ij}$ for $i,j\in [4]$. 
    % {We call $R_{ij}$ a \emph{corner sub-square} for $i,j\in\{1,4\}$; a \emph{middle sub-square} for $i,j\in\{2,3\}$, and the remaining sub-square is called a \emph{side sub-square}.}
    Note that the sub-square $R_{ij}$, for $i,j\in\{2,3\}$, is red point free.
   Now, depending on the position of the red point $r$, 
    we have the following three cases.
    \begin{itemize}
        \item \emph{\underline{Case 1:}} The red point $r$ does not lie in $R$.
        % Either $\P_r=\emptyset$, or $\P_r=\{r\}$ but the red point $r$ does not lie in $R$.
        Then, for any blue point $b\in\O$, the square $R(b)$ never contains a red point. As a result, $\S(b)=\{R(b)\}$. Due to Observation~\ref{obs:main}, we have that Algorithm~\ref{alg:occp} constructs at most~$4$ candidate squares that cover all the blue points in $\O$.
        
        \item \emph{\underline{Case 2:}} The red point $r$ lies in a sub-square $R_{ij}$ for $i,j\in\{1,4\}$ (see Figure~\ref{fig:m1_case2}). W.l.o.g., assume that $r\in R_{41}$; the other cases are similar. Now, if $b\in \O_{SW}=R_{32}$, then the square $R(b)$ may contain~$r$. If $r\notin R(b)$, then $\S(b)=\{R(b)\}$ and clearly $\O_{SW}\subseteq R(b)$. Otherwise, $r\in R(b)$. In that case, as defined in Section~\ref{sec:all_candidate_square} (Case~4), there are exactly two (non-empty) candidate squares of $b$. The first will be a horizontal translated copy $H$ of $R(b)$, while the second will be a vertical translated copy $V$ of $R(b)$. Now, it is easy to see that one of them must contain $\O_{SW}$. Now, consider that $b\in \O\setminus\O_{SW}$. In this case, observe that $R(b)$ never contains $r$. As a result, $\S(b)=\{R(b)\}$, and $R(b)$ contains the sub-square of $\O$, where $b$ lies (due to Observation~\ref{obs:main}). Therefore, Algorithm~\ref{alg:occp} constructs at most~$5$ candidate squares that cover all the blue points in $\O$.

\begin{figure}[htbp]
    \centering
    \begin{subfigure}[b]{.49\textwidth}
        \centering
        \includegraphics[page=2, width=60mm]{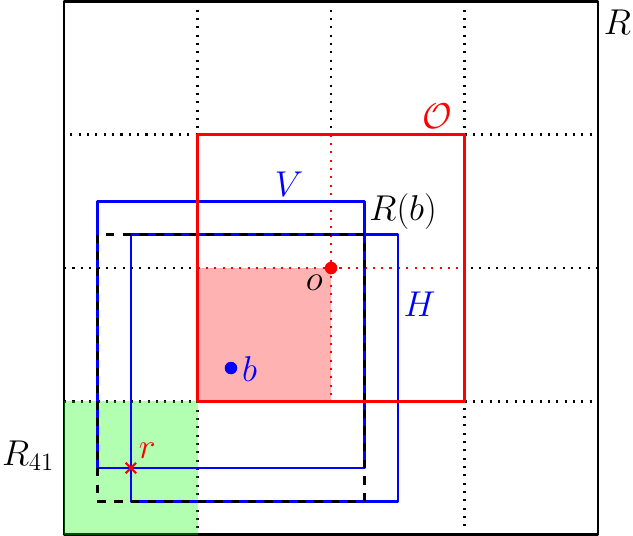}
        \caption{}
        \label{fig:m1_case2}
    \end{subfigure}
    \hfill
    \begin{subfigure}[b]{.49\textwidth}
        \centering
        \includegraphics[page=3, width=55mm]{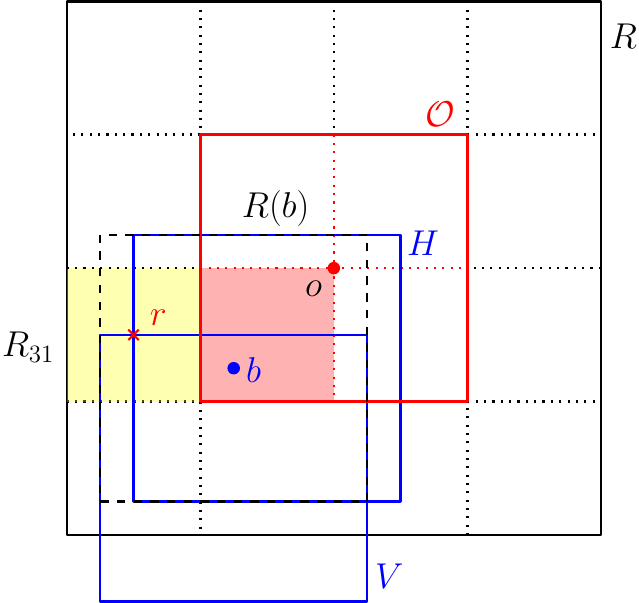}
        \caption{}
        \label{fig:m1_case3}
    \end{subfigure}
    \caption{Let $\O$ (red) be an optimum square centered at a point $o$, and $R$ (black) be a square of length 2 centered at $o$. $R(b)$ is the square centered at a blue point $b\in \O_{SW}$. The squares $H$ and $V$ are the horizontal and the vertical translated copies of $R(b)$, respectively. (a) The red point $r$ lies in the sub-square $R_{41}$ of $R$. (b)~The red point $r$ lies in the sub-square $R_{31}$ of $R$.}
    \label{fig:m_1_2nd}
\end{figure}

        \item \emph{\underline{Case 3:}} The red point $r$ lies in $R$, but does not lie in a sub-square $R_{ij}$, where either $i,j\in\{1,4\}$ or $i,j\in\{2,3\}$ (see Figure~\ref{fig:m1_case3}). W.l.o.g., assume that $r\in R_{31}$; the other cases are similar. Now, if $b$ lies in $\O_{SW}=R_{32}$ or in $\O_{NW}=R_{22}$, then the square $R(b)$ may contain~$r$. Suppose, $b\in \O_{SW}$. If $r\notin R(b)$, then $\S(b)=\{R(b)\}$ and clearly $\O_{SW}\subseteq R(b)$. Otherwise, $r\in R(b)$. In that case, again as defined in Section~\ref{sec:all_candidate_square} (Case~4), there are exactly two (non-empty) candidate squares of $b$, and it is easy to see that one of them must contain $\O_{SW}$. Similarly, if $b\in \O_{NW}$, we can show that $b$ has at most two (non-empty) candidate squares, and one of them must contain $\O_{NW}$. Now, consider that $b\in \O_{NE}\cup\O_{SE}$. In this case, observe that $R(b)$ never contains $r$. As a result, $\S(b)=\{R(b)\}$, and $R(b)$ contains the sub-square of $\O$, where $b$ lies (due to Observation~\ref{obs:main}). Therefore, Algorithm~\ref{alg:occp} constructs at most~$6$ candidate squares that cover all the blue points in $\O$.
    \end{itemize}

    Hence, the lemma follows.

    % Consider a blue point $u_1$ such that it is the first uncovered blue point in $\R$ to be introduced to the Algorithm~\ref{alg:occp}. Without loss of generality, assume that $u_1 \in \R_{SW}$. Now we claim that $|\S(u_1)|\leq 2$ and $\R_{SW}\subseteq \left(\cup_{P\in \S(u_1)} P\right)$. If $r\notin R(u_1)$, then $\S(u_1)=\{R(u_1)\}$ and clearly $\R_{SW}\subseteq R(u_1)$.

    % Otherwise, $r\in R(u_1)$. Note that $r\notin R(u_1)_{NE}$. W.l.o.g, assume that $r\in R(u_1)_{SW}$ (see Figure~\ref{fig:m_1}). The other cases such as $r\in R(u_1)_{NW}$ and $r\in R(u_1)_{SE}$ are similar in nature. In this case (as defined in Section~\ref{sec:candidate_square} (Case~4)), there are exactly two candidate squares of $u_1$. The first will be a horizontal translated copy $H$ of $R(u_1)$ with its left edge passing through $r$, while the second will be a vertical translated copy $V$ of $R(u_1)$ with its bottom edge passing through $r$. Now it is easy to see that one of them must contain $\R_{SW}$.

    % Similarly, we can prove that Algorithm~\ref{alg:occp} places at most $2$ candidate squares when $u_1$ lies on the other sub-squares of $\R$ such as $\R_{NW}, \R_{NE}$ and $\R_{SE}$. Also, the union of the candidate squares contains the sub-square of $\R$. Hence, to cover all the blue points, contained in $\R$, it needs at most $8$ candidate squares. Thus, the lemma follows.
\end{proof}

\begin{theorem}\label{theo:upperbound}
    There exists an algorithm for the online class cover problem for squares that achieves a competitive ratio of $\max\{6, 20\log_2 (m)\}$, where $m\ (\geq 1)$ is the number of red points.
\end{theorem}
\begin{proof}
    Let $\I$ be an input sequence of blue points. Suppose that $\text{OPT}(\I)$ be an optimal solution given by an offline algorithm, and $\text{ALG}(\I)$ be the solution given by Algorithm~\ref{alg:occp} for the sequence $\I$. Also, let $\O \in \text{OPT}(\I)$. 
    % Consider first that $m=0$, i.e. there are no red points. Whenever a blue point, say $u$, is introduced to the algorithm and not already covered by the previously selected squares, Algorithm~\ref{alg:occp} places $R(u)$ to cover $u$. As a result of Observation~\ref{obs:main}, we have that Algorithm~\ref{alg:occp} needs at most $4$ squares to cover all the blue points of $\I\cap\O$. Therefore,~$|\text{ALG}(\I)| \leq \sum_{\O\in\text{OPT}(\I)} 4 = 4 |\text{OPT}(\I)|$.
    Consider first that $m=1$, i.e., there is exactly one red point. From Lemma~\ref{lema:m_1}, we have that Algorithm~\ref{alg:occp} constructs at most $6$ squares that cover all the blue points in $\I\cap\O$. Therefore,~$|\text{ALG}(\I)| \leq \sum_{\O\in\text{OPT}(\I)} 6 = 6 |\text{OPT}(\I)|$.
    
    Next, we consider $m\geq 2$. From Lemma~\ref{lema:20log(m)}, we have that Algorithm~\ref{alg:occp} constructs at most $20\log_2 (m)$ squares that cover all the blue points in $\I\cap\O$. Therefore, $|\text{ALG}(\I)| \leq \sum_{\O\in\text{OPT}(\I)} 20\log_2 (m) = 20\log_2 (m) |\text{OPT}(\I)|$. Thus, the theorem follows.
\end{proof}

%  = \{u_1,u_2, \dots, u_k\}

\subsection{Improving the Competitive Ratio Further}\label{sec:improved_comp}
From Lemma~\ref{lema:candidate_squares} and Lemma~\ref{lm: type 1 either or}, we have if a blue point $u\in \O_{SW}$ such that any candidate squares of $u$ does not contain $\O_{SW}$, then their union contains $\O_{NW}\cup \O_{SE}$. The equivalent statement is true when $u$ lies in other sub-squares of $\O$. So, a careful analysis of Algorithm~\ref{alg:occp} reduces the competitive ratio to $10+10\log_2 (m)$ from $20\log_2 (m)$, where  $m\ (\geq 2)$ is the number of red points (see Figure~\ref{fig:improve_ratio}). Let $u_1, u_2, \dots, u_k$ be the largest sub-sequence of the input sequence such that $u_i\in\O$ for $i\in [k]$, and they are uncovered upon arrival. W.l.o.g., we assume that $u_1\in \O_{SW}$. So due to Lemma~\ref{lema:candidate_squares} and Lemma~\ref{lm: type 1 either or}, we have the following two cases.

\begin{figure}[htbp]
\centering
\includegraphics[width=70 mm]{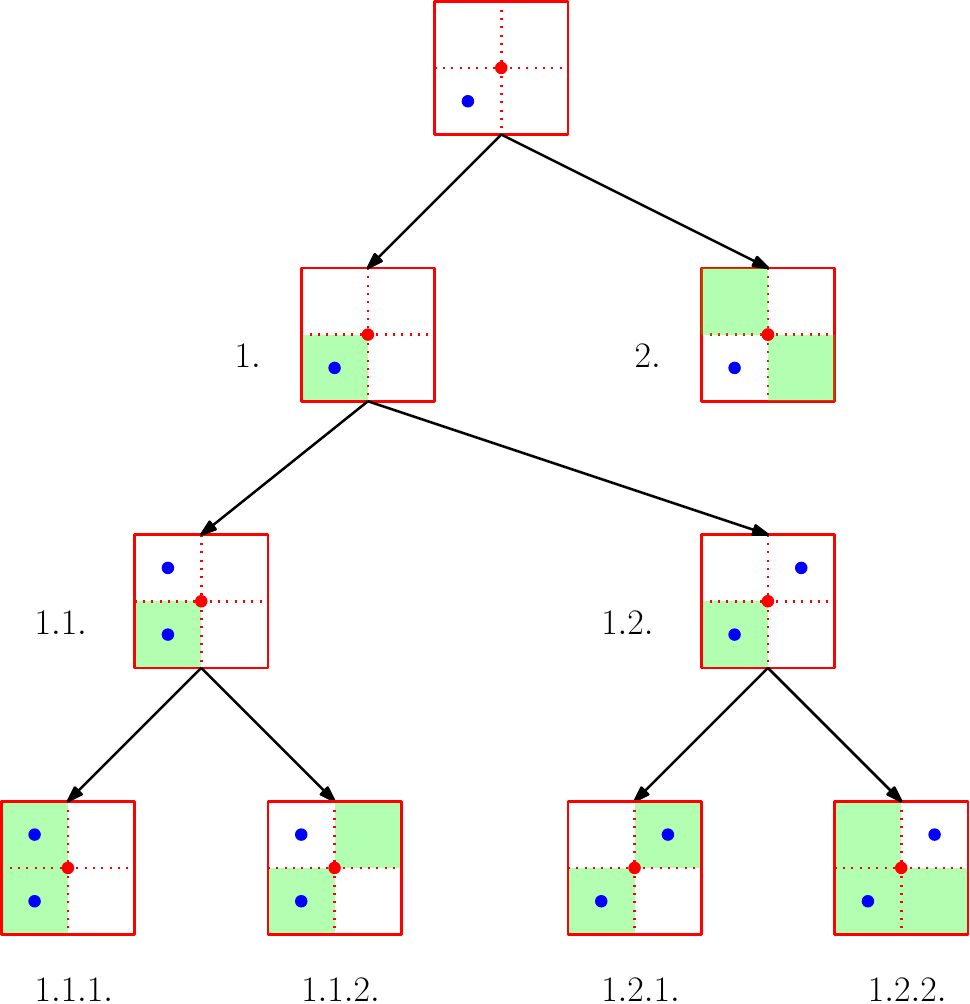}
\caption{Red square is an optimum square. Blue points are the input points. Some candidate squares cover green-shaded regions.}
\label{fig:improve_ratio}
\end{figure}

\begin{enumerate}
    \item A candidate square of $u_1$ contains $\O_{SW}$. Hence, next blue point $u_2$ lies on either $\O_{NW}\cup\O_{SE}$ or $\O_{NE}$.
    \begin{enumerate}
        \item[1.1.] Let $u_2$ lies on $\O_{NW}\cup \O_{SE}$. W.l.o.g., assume that $u_2\in \O_{NW}$. The other case is similar. Then, due to Lemma~\ref{lema:candidate_squares} and Lemma~\ref{lm: type 1 either or}, we have the following two cases.
        \begin{enumerate}
            \item[1.1.1.] A candidate square of $u_2$ contains $\O_{NW}$. In this case, Algorithm~\ref{alg:occp} constructs at most $2\cdot 5\log_2 (m)$ candidate squares that cover $\O_{NE}\cup \O_{SE}$ (due to Lemma~\ref{lema:5log(m)}). Therefore, our algorithm constructs at most $10+10\log_2 (m)$ candidate squares that cover $\O$.
            \item[1.1.2.] Union of the candidate squares of $u_2$ contains $\O_{NE}\cup \O_{SW}$. In this case, Algorithm~\ref{alg:occp} constructs at most $2\cdot 5\log_2 (m)$ candidate squares that cover $\O_{NW}\cup \O_{SE}$ (due to Lemma~\ref{lema:5log(m)}). Therefore, our algorithm constructs at most $10+10\log_2 (m)$ candidate squares that cover~$\O$.
        \end{enumerate}
        \item[1.2.] Let $u_2$ lies $\O_{NE}$. Then, due to Lemma~\ref{lema:candidate_squares} and Lemma~\ref{lm: type 1 either or}, we have the following two cases.
        \begin{enumerate}
            \item[1.2.1.] A candidate square of $u_2$ contains $\O_{NE}$. In this case, Algorithm~\ref{alg:occp} constructs at most $2\cdot 5\log_2 (m)$ candidate squares that cover $\O_{NW}\cup \O_{SE}$ (due to Lemma~\ref{lema:5log(m)}). Therefore, our algorithm constructs at most $10+10\log_2 (m)$ candidate squares that cover~$\O$.
            \item[1.2.2.] The union of the candidate squares of $u_2$ contains $\O_{NW}\cup \O_{SE}$. In this case, Algorithm~\ref{alg:occp} constructs at most $5\log_2 (m)$ candidate squares that cover $\O_{NE}$ (due to Lemma~\ref{lema:5log(m)}). Therefore, our algorithm constructs at most $10+5\log_2 (m)$ candidate squares that cover~$\O$.
        \end{enumerate}
    \end{enumerate}
    \item Union of the candidate squares of $u_1$ contains $\O_{NW}\cup \O_{SE}$. In this case, Algorithm~\ref{alg:occp} constructs at most $2\cdot 5\log_2 (m)$ candidate squares that cover $\O_{SW}\cup \O_{NE}$ (due to Lemma~\ref{lema:5log(m)}). Therefore, our algorithm constructs at most $5+10\log_2 (m)$ candidate squares that cover~$\O$.
\end{enumerate}

% Hence, the competitive ratio of our Algorithm is $10+10\log_2(m)$.
Hence, we have the following theorem.

\begin{theorem}\label{theo:upperbound_improved}
    There exists an algorithm for the online class cover problem for squares that achieves a competitive ratio of $10+10\log_2 (m)$, where $m\ (\geq 2)$ is the number of red points.
\end{theorem}

\section{Conclusion}\label{sec:conclud}
In this paper, we have discussed the online class cover problem. Our main results are proofs for a lower and upper bound of the problem for axis-parallel unit squares. 
Observe that the upper bound applies to the problem even if the geometric objects are translated copies of a rectangle. 
Notice that the lower and upper bounds of the problem depend on the maximum number of red points that can lie in a square. If this number is bounded by a constant, then so are the lower and upper bounds of the problem. Obtaining the same bounds for unit disks or translates of a convex object will be an interesting question. Since our algorithm and lower bound are based on the deterministic model, it would be important to question whether randomization helps to obtain a better competitive ratio.

% We only consider the online class cover problem for translated copies of a rectangle. Therefore, we propose the following open problem. What can be said about the lower and upper bounds on the competitive ratio of the class cover problem for other objects, such as disks, convex objects etc.?

\section*{Acknowledgement}
The authors would like to acknowledge Satyam Singh for participating in formulating the problem. We also thank the reviewers for their helpful comments and suggestions, which made the manuscript better.

\section*{Conflict of Interest}
The authors declare that there are no financial and non-financial competing interests that are relevant to the content of this article.

\bibliography{references}

\begin{thebibliography}{10}

\bibitem{AgarwalS98}
Pankaj~K. Agarwal and Subhash Suri.
\newblock Surface approximation and geometric partitions.
\newblock {\em {SIAM} J. Comput.}, 27(4):1016--1035, 1998.

\bibitem{AlonAABN09}
Noga Alon, Baruch Awerbuch, Yossi Azar, Niv Buchbinder, and Joseph Naor.
\newblock The online set cover problem.
\newblock {\em {SIAM} J. Comput.}, 39(2):361--370, 2009.

\bibitem{AschnerKMY13}
Rom Aschner, Matthew~J. Katz, Gila Morgenstern, and Yelena Yuditsky.
\newblock Approximation schemes for covering and packing.
\newblock In {\em {WALCOM:} Algorithms and Computation, 7th International
  Workshop, {WALCOM} 2013, February 14-16, 2013. Proceedings}, volume 7748 of
  {\em Lecture Notes in Computer Science}, pages 89--100. Springer, 2013.

\bibitem{BeregCDPSV12}
Sergey Bereg, Sergio Cabello, Jos{\'{e}}~Miguel
  D{\'{\i}}az{-}B{\'{a}}{\~{n}}ez, Pablo P{\'{e}}rez{-}Lantero, Carlos Seara,
  and Inmaculada Ventura.
\newblock The class cover problem with boxes.
\newblock {\em Comput. Geom.}, 45(7):294--304, 2012.

\bibitem{Borodin98}
Allan Borodin and Ran El-Yaniv.
\newblock {\em Online computation and competitive analysis.}
\newblock Cambridge University Press, 1998.

\bibitem{CannonC04}
Adam Cannon and Lenore Cowen.
\newblock Approximation algorithms for the class cover problem.
\newblock {\em Ann. Math. Artif. Intell.}, 40(3-4):215--224, 2004.

\bibitem{cannon1998approximate}
Adam~H Cannon, Lenore~J Cowen, and Carey~E Priebe.
\newblock Approximate distance classification.
\newblock {\em Computing Science and Statistics}, pages 544--549, 1998.

\bibitem{CaragiannisFKP07}
Ioannis Caragiannis, Aleksei~V. Fishkin, Christos Kaklamanis, and Evi
  Papaioannou.
\newblock Randomized on-line algorithms and lower bounds for computing large
  independent sets in disk graphs.
\newblock {\em Discret. Appl. Math.}, 155(2):119--136, 2007.

\bibitem{CardinalDI21}
Jean Cardinal, Justin Dallant, and John Iacono.
\newblock Approximability of (simultaneous) class cover for boxes.
\newblock In {\em Proceedings of the 33rd Canadian Conference on Computational
  Geometry, {CCCG} 2021, August 10-12, 2021}, pages 149--156, 2021.

\bibitem{ChanZ09}
Timothy~M. Chan and Hamid Zarrabi{-}Zadeh.
\newblock A randomized algorithm for online unit clustering.
\newblock {\em Theory Comput. Syst.}, 45(3):486--496, 2009.

\bibitem{CharikarCFM04}
Moses Charikar, Chandra Chekuri, Tom{\'{a}}s Feder, and Rajeev Motwani.
\newblock Incremental clustering and dynamic information retrieval.
\newblock {\em {SIAM} J. Comput.}, 33(6):1417--1440, 2004.

\bibitem{ChenFKLMMPSSWW07}
Ke~Chen, Amos Fiat, Haim Kaplan, Meital Levy, Jir{\'{\i}} Matousek, Elchanan
  Mossel, J{\'{a}}nos Pach, Micha Sharir, Shakhar Smorodinsky, Uli Wagner, and
  Emo Welzl.
\newblock Online conflict-free coloring for intervals.
\newblock {\em {SIAM} J. Comput.}, 36(5):1342--1359, 2007.

\bibitem{ChenKS09}
Ke~Chen, Haim Kaplan, and Micha Sharir.
\newblock Online conflict-free coloring for halfplanes, congruent disks, and
  axis-parallel rectangles.
\newblock {\em {ACM} Trans. Algorithms}, 5(2):16:1--16:24, 2009.

\bibitem{cowen1997randomized}
Lenore~J Cowen and Carey~E Priebe.
\newblock Randomized nonlinear projections uncover high-dimensional structure.
\newblock {\em Advances in Applied Mathematics}, 19(3):319--331, 1997.

\bibitem{DeJKS24}
Minati De, Saksham Jain, Sarat~Varma Kallepalli, and Satyam Singh.
\newblock Online geometric covering and piercing.
\newblock {\em Algorithmica}, pages 1--27, 2024.

\bibitem{DeKS23}
Minati De, Sambhav Khurana, and Satyam Singh.
\newblock Online dominating set and coloring.
\newblock In {\em Combinatorial Optimization and Applications - 17th
  International Conference, {COCOA} 2023, December 15-17, 2023, Proceedings,
  Part {I}}, volume 14461 of {\em Lecture Notes in Computer Science}, pages
  68--81. Springer, 2023.

\bibitem{DeS24}
Minati De and Satyam Singh.
\newblock Online hitting of unit balls and hypercubes in $\mathbb{R}^d$ using
  points from $\mathbb{Z}^d$.
\newblock {\em Theor. Comput. Sci.}, 992:114452, 2024.

\bibitem{Devinney03}
Jason~Gary DeVinney.
\newblock {\em The class cover problem and its applications in pattern
  recognition}.
\newblock Ph.{D}. dissertation, The Johns Hopkins University, 2003.

\bibitem{DumitrescuGT20}
Adrian Dumitrescu, Anirban Ghosh, and Csaba~D. T{\'{o}}th.
\newblock Online unit covering in {E}uclidean space.
\newblock {\em Theor. Comput. Sci.}, 809:218--230, 2020.

\bibitem{DumitrescuT22}
Adrian Dumitrescu and Csaba~D. T{\'{o}}th.
\newblock Online unit clustering and unit covering in higher dimensions.
\newblock {\em Algorithmica}, 84(5):1213--1231, 2022.

\bibitem{Eidenbenz}
Stephan Eidenbenz.
\newblock Online dominating set and variations on restricted graph classes.
\newblock {\em Technical Report No 380, ETH Library}, 2002.

\bibitem{EvenS14}
Guy Even and Shakhar Smorodinsky.
\newblock Hitting sets online and unique-max coloring.
\newblock {\em Discret. Appl. Math.}, 178:71--82, 2014.

\bibitem{OliverHKSV14}
Oliver G{\"{o}}bel, Martin Hoefer, Thomas Kesselheim, Thomas Schleiden, and
  Berthold V{\"{o}}cking.
\newblock Online independent set beyond the worst-case: Secretaries, prophets,
  and periods.
\newblock In {\em Automata, Languages, and Programming - 41st International
  Colloquium, {ICALP} 2014, July 8-11, 2014, Proceedings, Part {II}}, volume
  8573 of {\em Lecture Notes in Computer Science}, pages 508--519. Springer,
  2014.

\bibitem{ArindamLRSW23}
Arindam Khan, Aditya Lonkar, Saladi Rahul, Aditya Subramanian, and Andreas
  Wiese.
\newblock Online and dynamic algorithms for geometric set cover and hitting
  set.
\newblock In {\em 39th International Symposium on Computational Geometry, SoCG
  2023, June 12-15, 2023}, volume 258 of {\em LIPIcs}, pages 46:1--46:17.
  Schloss Dagstuhl - Leibniz-Zentrum f{\"{u}}r Informatik, 2023.

\bibitem{Mitchell93}
Joseph~SB Mitchell.
\newblock Approximation algorithms for geometric separation problems.
\newblock Technical report, AMS Dept., SUNY Stony Brook, NY, 1993.

\bibitem{Shanjani20}
Sima~Hajiaghaei Shanjani.
\newblock Hardness of approximation for red-blue covering.
\newblock In {\em Proceedings of the 32nd Canadian Conference on Computational
  Geometry, {CCCG} 2020, August 5-7, 2020}, pages 39--48, 2020.

\end{thebibliography}

\end{document}